\documentclass[a4paper]{article}
\usepackage[margin=1in]{geometry}

\usepackage{amsmath,amsthm,amssymb,bbm,hyperref,enumerate,graphicx,caption,subcaption}
\usepackage{natbib}
\usepackage{xcolor}
\usepackage{float}
\usepackage[]{algorithm2e}

\hypersetup{colorlinks,citecolor=blue,urlcolor=black}%

\usepackage{enumitem}

\usepackage{booktabs}

\newcommand{\N}{\mathbb{N}}

\newcommand{\R}{\mathbb{R}}

\newcommand{\E}{\mathbb{E}}
\renewcommand{\Pr}{\mathbb{P}}

\newcommand{\ceil}[1]{\lceil #1 \rceil}

\newcommand{\cP}{\mathcal{P}}

\newcommand{\cX}{\mathcal{X}}

\newcommand{\cH}{\mathcal{H}}

\newcommand{\cI}{\mathcal{I}}
\newcommand{\cE}{\mathcal{E}}
\newcommand{\op}{o_{\cP,\cX,\cH,\Lambda}}
\newcommand{\supoP}{o_{\cP,\cX,\cH,\Lambda}}
\newcommand{\supOP}{O_{\cP,\cX,\cH,\Lambda}}
\newcommand{\ind}{\mathbbm{1}}
\newcommand{\RN}[1]{\textup{\uppercase\expandafter{\romannumeral#1}}}
\newcommand{\opt}a

\newcommand{\indep}{\mathrel{\perp\!\!\!\perp}}

\newcommand{\given}{\,|\,}
\newcommand{\biggiven}{\,\big|\,}
\newcommand{\Biggiven}{\,\Big|\,}
\newcommand{\bigggiven}{\,\bigg|\,}

\newcommand{\outerr}{\mathcal{B}_{x_0}(\lambda h)\setminus\mathcal{B}_{x_0}(h)}
\newcommand{\innerr}{\mathcal{B}_{x_0}(h)}

\makeatletter
\newcommand*\bigcdot{\mathpalette\bigcdot@{.5}}
\newcommand*\bigcdot@[2]{\mathbin{\vcenter{\hbox{\scalebox{#2}{$\m@th#1\bullet$}}}}}
\makeatother

\DeclareMathOperator\Var{Var}

\DeclareMathOperator\tr{tr}
\DeclareMathOperator{\supp}{supp}

\DeclareMathOperator{\sgn}{sgn}
\DeclareMathOperator*{\argmin}{argmin}
\DeclareMathOperator*{\sargmin}{sargmin}

\newtheorem{theorem}{Theorem}

\newtheorem{lemma}[theorem]{Lemma}
\newtheorem{proposition}[theorem]{Proposition}
\newtheorem{corollary}[theorem]{Corollary}

\theoremstyle{definition}

\newtheorem*{remark*}{Remark}
\newtheorem{assumption}{Assumption}

\makeatletter
\def\underbar{\underline}
\makeatother

\usepackage{authblk}
\author[1]{{Elliot H.\ Young}
\thanks{ey244@cam.ac.uk}}
\author[1]{{Rajen D.\ Shah}
\thanks{r.shah@statslab.cam.ac.uk}}
\author[1]{{Richard J.\ Samworth}
\thanks{r.samworth@statslab.cam.ac.uk}}
\affil{Statistical Laboratory, University of Cambridge, UK}

\title{Outrigger local polynomial regression}
\date{}

\DeclareMathAlphabet{\mathpzc}{OT1}{pzc}{m}{it}
\begin{document}

\maketitle

\begin{abstract}
Standard local polynomial estimators of a nonparametric regression function employ a weighted least squares loss function that is tailored to the setting of homoscedastic Gaussian errors.  We introduce the \emph{outrigger local polynomial estimator}, which is designed to achieve distributional adaptivity across different conditional error distributions.  It modifies a standard local polynomial estimator by employing an estimate of the conditional score function of the errors and an `outrigger' that draws on the data in a broader local window to stabilise the influence of the conditional score estimate.  Subject to smoothness and moment conditions, and only requiring consistency of the conditional score estimate, we first establish that even under the least favourable settings for the outrigger estimator, the asymptotic ratio of the worst-case local risks of the two estimators is at most $1$, with equality if and only if the conditional error distribution is Gaussian.  Moreover, we prove that the outrigger estimator is minimax optimal over H\"older classes up to a multiplicative factor $A_{\beta,d}$, depending only on the smoothness $\beta \in (0,\infty)$ of the regression function and the dimension~$d$ of the covariates.  When $\beta \in (0,1]$, we find that $A_{\beta,d} \leq 1.69$, with $\lim_{\beta \searrow 0} A_{\beta,d} = 1$.  A further attraction of our proposal is that we do not require structural assumptions such as independence of errors and covariates, or symmetry of the conditional error distribution.  Numerical results on simulated and real data validate our theoretical findings; our methodology is implemented in \texttt{R} and available at \url{https://github.com/elliot-young/outrigger}.
\end{abstract}

\section{Introduction}\label{sec:intro}

The estimation of a regression function is one of the most important and widely studied problems in statistics and machine learning. Suppose we have access to independent copies $(X_1,Y_1),\ldots,(X_n,Y_n)$ of $(X,Y) \sim P$ on $\mathbb{R}^d \times \mathbb{R}$ with $\E_P(Y^2) < \infty$, and are interested in estimating the \emph{conditional mean function} $f:\mathbb{R}^d \rightarrow \mathbb{R}$, given by
\begin{equation*}
    f(x) := \E_P(Y\given X=x).
\end{equation*}
It is well known that 
\begin{equation*}
    f \in \argmin_{g \in \mathcal{G}} \E_P\bigl\{w(X)\bigl(Y-g(X)\bigr)^2\bigr\},
\end{equation*}
where $\mathcal{G}$ denotes the set of Borel measurable functions from $\mathbb{R}^d$ to $\mathbb{R}$ and $w \in \mathcal{G}$ is a non-negative weight function. 
This fundamental observation underpins empirical risk minimisation with squared error loss. Indeed, many if not most commonly used regression methods are motivated by this core principle, ranging from ordinary and weighted least squares in parametric models to standard implementations of more flexible approaches such as local polynomial regression, random forests, gradient boosting, neural networks and splines.

A more formal justification for minimising the least squares loss comes from semiparametric theory \citep{bickel}.  
Indeed, for an arbitrary parametric conditional mean model, a weighted least squares estimator is \emph{semiparametrically efficient}: no estimator can outperform it in a local asymptotic minimax sense~\citep[Chapter~4]{tsiatis}.  As another example, for the nonparametric problem of estimating~$\E(Y\given X\in R)= \int_R f \,dP$ over a measurable region $R\subseteq\R^d$ subject only to the regularity constraints that $\mathbb{P}(X \in R) > 0$, $\E_P(Y^2)<\infty$ and $f$ is bounded and measurable, a semiparametric efficient estimator is the sample mean $\frac{1}{|\{i:X_i\in R\}|}\sum_{i:X_i\in R}Y_i$, which minimises $\mu \mapsto n^{-1}\sum_{i=1}^n \bigl\{(Y_i-\mu)^2\ind_R(X_i)\bigr\}$~\citep[Example 3.2]{bickel}.  This latter example acts as a heuristic justification for the use of local least squares estimators such as local polynomials and random forests. 
A notable characteristic of these semiparametric optimality results, however, is that they are only valid in settings where the parameter of interest is estimable at rate $n^{-1/2}$.  In other words, they need not apply in the nonparametric settings for which modern statistical and machine learning methods are designed.  
Of course, there is one nonparametric setting where squared error loss remains natural, namely when the errors are conditionally Gaussian, so that (weighted) least squares corresponds simply to (local) maximum likelihood estimation.  On the other hand, different error distributions (if known) would give rise to alternative loss functions based on their respective negative local log-likelihoods, and could potentially lead to improved estimators. 

In this work, we introduce a new estimator of a nonparametric regression function, which we call an \emph{outrigger estimator}, designed to adapt to the unknown error distribution (which may in particular be non-Gaussian).  
A naive first attempt towards this goal would be to replace the conditional score function of the errors in the estimating equations arising from the local likelihood with a data-driven estimate.  
It turns out that this strategy introduces a significant bias, for reasons outlined in Section~\ref{sec:lpr}. 
Our primary methodological idea, then, is to modify a standard local polynomial estimator by combining an estimate of the conditional score of the errors with an `outrigger' that draws on the data in a broader local window to stabilise the influence of the conditional score estimate.  It is these two features that give rise to the method's name, since they are evocative of an outrigger on a boat or crane that projects over the side to provide stability.   

Our main theoretical results are of two flavours: first, a comparison of the local worst-case risks of the outrigger estimator and the standard local polynomial estimator; and second, a minimax analysis that compares the outrigger estimator with any alternative procedure.  In the first case, we are able to establish a strong sense in which the outrigger estimator yields an asymptotic improvement: uniformly over bandwidth sequences, estimation points $x_0$ and a broad class of data generating mechanisms, the ratio of the local worst-case risks is asymptotically at most one, with equality if and only if the error distribution is Gaussian.  In fact, we show that for $\beta$-H\"older smooth regression functions, the asymptotic least-favourable ratio of these local risks is given by 
\begin{equation}
\label{Eq:BasicRatio}
    \biggl(\frac{1/i_P(x_0)}{\sigma_P^2(x_0)}\biggr)^{2\beta/(2\beta+d)}
\end{equation}
under optimal bandwidth choices for both methods, where $\sigma_P^2(x_0)$ denotes the conditional error variance and $i_P(x_0)$, defined in~\eqref{eq:Fisher-info} below, denotes the conditional Fisher information of the errors. 
The ratio~\eqref{Eq:BasicRatio} is indeed at most one, with equality if only if the error distribution is Gaussian, and in this sense standard local polynomial estimators are asymptotically inadmissible; see Theorem~\ref{thm:improvements} for a precise statement.

Theorems~\ref{thm:UB} and~\ref{thm:LAM-LB} allow us to compare the worst-case mean squared error (over regression functions in a H\"older ball) of the outrigger estimator with a minimax lower bound.  Remarkably, the ratio of these two quantities depends asymptotically only on the H\"older smoothness $\beta \in (0,\infty)$ and the covariate dimension $d \in \mathbb{N}$.  In fact, when $\beta \in (0,1]$, the asymptotic ratio is at most $1.69$ for every $d$, and converges to $1$ in the low smoothness limit as $\beta \searrow 0$, showing that the outrigger estimator has almost optimal performance even at the level of constants.

\begin{figure}[ht]
  \centering
  \begin{minipage}{0.61\textwidth}
    \centering
    \includegraphics[width=\linewidth]{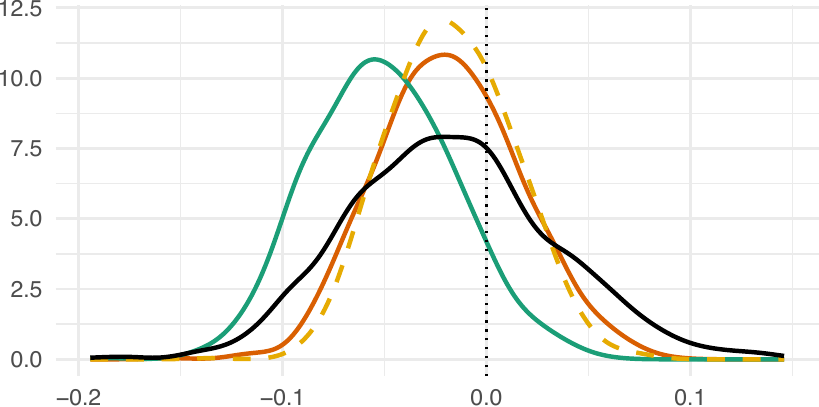}
  \end{minipage}%
  \hfill
  \begin{minipage}{0.38\textwidth}
    \vspace{0cm}\centering
		\begin{tabular}{cc}
			\toprule
			Estimator & {\begin{tabular}{@{}c@{}c@{}}MSE\\{\small$(\times10^3)$}\end{tabular}} \\
			\midrule
			Standard local polynomial  & 3.04     \\
            Oracle (Local likelihood) & 1.23 \\
			{\bf Outrigger} & {\bf1.51}     \\
		      Score plug-in & 3.47     \\
			\bottomrule
		\end{tabular}  
  \end{minipage}

 \vspace{0.3cm} 
  \includegraphics[width=0.82\linewidth]{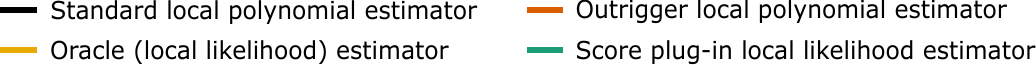}
  \caption{Kernel density estimates of $\hat{f}(0) - f(0)$ for the simulation example of Section~\ref{sec:numerics-nonindep-errors}\ref{item:exp-t3} for different estimators $\hat{f}$, based on 1000 repetitions with sample size $n=10^4$.  A standard local constant estimator (black) does not adapt to the unknown (non-Gaussian) error distribution, so its variance is larger than that of the oracle local likelihood estimator~\eqref{eq:rho-est-oracle} (dashed yellow) that exploits knowledge of the conditional score function $\rho$. The estimator~\eqref{eq:errorO} based on a naive distributional plug-in estimator $\hat\rho$ of $\rho$ (green) reduces variance compared with the standard local polynomial estimator, at the expense of a significant additional bias.  On the other hand, our outrigger estimator (orange) enjoys a very similar reduction in variance to the distributional plug-in estimator, and a similar bias to those of the oracle and standard local polynomial estimators.  The mean squared error of each estimator is given in the adjoining table.
}
\label{fig:intro}
\end{figure}

At this point, it is worth providing some intuition for distributional adaptivity in nonparametric regression.  As alluded to briefly above, if the conditional distribution of the errors given the covariates were known, then it would be natural to employ a locally weighted maximum likelihood estimator, as in the local likelihood approach of~\citet{loc-lik}.  In the more typical case where this conditional error distribution is unknown, a first thought would be to estimate the error distribution and adopt a plug-in local likelihood strategy.  The main issue with this approach is that the bias incurred in the estimation of the conditional error distribution typically leads to poor downstream regression function estimates, as has been observed in many related settings~\citep[e.g.][Example 3.2]{bickel};  see also the simulated example in Figure~\ref{fig:intro}.  Indeed, the conditional error density is a function of $d+1$ variables (where the last $d$ variables are the covariates) and the relevant functional of this density for local likelihood is the partial derivative with respect to its first argument of its logarithm (i.e.~the conditional score function).  
Estimating the conditional score may therefore be regarded as more challenging than estimating the original regression function, which is a function of $d$ variables with no derivatives involved.  In addition to making parametric assumptions on the form of the regression function, prior works have therefore imposed structural assumptions on the conditional error distribution, most commonly that the errors and covariates are independent, or that the conditional error distribution is symmetric; see, e.g.~\citet[Example~5.3, Section~4.3]{bickel},~\citet[Section~25.8.1]{vandervaart},~\citet[Section~5.1]{tsiatis},~\citet[Section~4.1.1]{kosorok},~\citet{deepreg-lik} and~\citet{asm}.  
The main effect of our outrigger is to stabilise the conditional score estimate under smoothness assumptions by eliminating its dominant bias contribution.  To the best of our knowledge then, this is the first work to achieve optimal distributional adaptation in nonparametric regression, and moreover this is achieved without structural assumptions on the conditional error distribution.  

Our outrigger estimator leverages two key assumptions in achieving distributional adaptivity, namely smoothness of the data generating distribution and access to a smooth and consistent conditional score estimator.  With regards to the latter, various (conditional) score estimators have been proposed, including those based on score matching \citep{dennis-cox, scorematching}, generative adversarial networks~\citep{gan}, engression~\citep{engression} and distributional learners such as distributional random forests~\citep{drf}. 
We further note that in the modern debiased learning framework, it is often assumed that nuisance functions can be estimated at rates faster than~$n^{-1/4}$~\citep[e.g.][]{dml}.  An attraction of our novel outrigger scheme is that we require only consistency of the conditional score estimator.   

The remainder of this paper is organised as follows.  After reviewing related literature and defining relevant notation, we present our outrigger local polynomial estimator methodology in Section~\ref{sec:methodology}.  This section begins with a more detailed discussion of oracle local likelihood estimation and the difficulties alluded to above with distributional plug-in estimation.  Our theoretical results on the performance of the outrigger estimator are given in Section~\ref{sec:theory}; these are complemented by numerical experiments on simulated and real data in Section~\ref{sec:sims}.  We conclude in Section~\ref{sec:extensions} by discussing various extensions to our methodology, including settings where we may prefer to estimate a proxy for the conditional score, or where our interest is in some other aspect of the conditional distribution of the response given covariates (e.g.~quantile regression).  All proofs are deferred to the Appendix.
\subsection{Related work}\label{sec:lit-rev}

Outside low-dimensional parametric models where the likelihood is assumed to be known, (weighted) least squares has long formed the data-fidelity terms employed in loss functions for estimators constructed as (penalised) empirical risk minimisers.  Beyond ordinary least squares and its weighted (Aitken estimator) variant, such estimators include ridge regression~\citep{ridge}, the Lasso~\citep{lasso}, standard local polynomial estimators~\citep{nadaraya,watson,stone}, random forests~\citep{randomforest}, gradient boosting~\citep{boosting,statisticalboosting}, neural networks~\citep{goodfellow} and splines~\citep{wahba}.  In nonparametric regression, minimax optimal rates over various smoothness classes are known to be achieved by several estimators formed in this way, e.g.~by local polynomials over H\"older smooth functions \citep[e.g.][]{tsybakov, samworth-shah}; gradient boosting for Sobolev classes~\citep{l2boost}; random forests for H\"older smooth functions~\citep{mondrian-minimax, underwood, crf}, neural networks over compositional H\"older classes~\citep{schmidt-hieber,ma2025deep} and splines over Sobolev classes~\citep{speckman}.  Going further, it is even known that linear estimators constructed via empirical (weighted) least squares can come close to matching minimax lower bounds at the level of constants~\citep{ibra, donoho1, donoho3, fan-minimax, cai}.  

A significant limitation of this impressive body of work that seeks almost optimal constants, however, is that it is restricted to settings with Gaussian error distributions.  In fact, in constructing global minimax lower bounds over classes of error distributions with mean zero and variance $\sigma^2>0$, these works have restricted attention to Gaussian submodels with $N(0,\sigma^2)$ errors; from this perspective, Gaussian errors are (at least almost) the worst case.  On the other hand, this leaves open the possibility that improved estimation performance may be achievable when our error distribution is non-Gaussian, and indeed the main contribution of this work is to show how these potential benefits can be realised.

Local polynomial estimators date back to the seminal works of~\citet{nadaraya}, \citet{watson}, \citet{priestley} and \citet{stone}. For modern treatments, see e.g.,~\citet{hardle}, \citet{wand}, \citet{fan}, \citet{gyorfi}, \citet{wasserman}, \citet{tsybakov}, \citet{samworth-shah}. These estimators have been extended to alternative regression contexts such as quantile regression~\citep{loc-quantile}, robust regression~\citep{robust-under-symmetry}, estimation of causal functions of interest~\citep{kennedy2, kennedy, scheidegger} and likelihood-based methods~\citep{loc-lik, staniswalis}, under additional structural assumptions.  
Over the years there has also been substantial work on inferential questions, including confidence interval and band construction~\citep{hall, xia, fan-inference, claeskens, wasserman,   cattaneo2}. 

Our work connects to the literature on distributional learning and generative modelling, which has emerged as a central pillar of modern machine learning. These methods, which encompass generative adversarial networks~\citep{gan}, autoencoders~\citep{autoencoder, vincent2}, diffusion models~\citep{song1}, and more generally energy-based distributional regression~\citep{energy-score, energy-book, drf, engression}, focus on modelling the entire probability distribution of the data. The Langevin dynamics that underpin diffusion models mean that it is convenient to represent these distributions via their score function, defined to be the derivative of the logarithm of the density.  Such considerations have led to the emergence of score estimation, especially via score matching, as a key estimation vehicle~\citep{dennis-cox, scorematching, wibisono, dou, lewis, asm}.  Our outrigger local polynomial estimator takes as an input a smooth, consistent estimator of the conditional score function, which can be obtained for instance via (conditional) score matching, or extracted from a generative model using the techniques above.

Recent years have seen a large literature in semiparametric statistics, seeking to estimate parameters of interest at rate $n^{-1/2}$ even in the presence of an infinite-dimensional nuisance parameter~\citep{robinson, bickel, vandervaart}.  Debiased machine learning approaches often have the property that such nuisance functions may only need to be estimated at rate $o(n^{-1/4})$, or even the more relaxed property that the product of the root mean squared errors of two nuisance function estimators is $o(n^{-1/2})$~\citep{dml,stijn, kennedy-dml}. 
Our setting is different in that the nonparametric regression function of interest is not estimable at rate $n^{-1/2}$, and in fact this is crucial for the success of our distributional adaptation methodology. 
Another key difference between our work and the semiparametric statistics literature is that we only require a consistent conditional score estimator (with no requirement on its rate of convergence).  This is a particularly attractive feature of our approach, as, especially in moderate- or high-dimensional problems, even $o(n^{-1/4})$ rates may be infeasible, so may result in these methods having poor practical performance~\citep{rose}. 

\subsection{Notation}\label{sec:notation}

We denote the set of natural numbers by $\N$, and write $\N_0:=\N\cup\{0\}$. For $n\in\N$, we write $[n]:=\{1,\ldots,n\}$ and $[n]_0:=[n]\cup\{0\}$.  For $D \in \mathbb{N}$ and $j \in [D]$, the $j$th standard basis vector in $\mathbb{R}^D$ is denoted $\mathrm{e}_j$. Fixing $q \geq 1$, we write $\|\cdot\|$ for the $\ell_q$ norm on $\mathbb{R}^d$, and define the closed ball of radius $r \geq 0$ around $x \in \mathbb{R}^d$ by $\mathcal{B}_{x}(r):=\{\tilde{x}\in\R^d:\|\tilde{x}-x\|\leq r\}$.  
Let $\|\cdot\|_{\mathrm{op}}$ denote the operator (spectral) norm of a square matrix.  For $\beta>0$, we write $\beta_0:=\ceil{\beta}-1$. We use multi-index notation for partial derivatives, so that, for $\alpha = (\alpha_1,\ldots,\alpha_d) \in \N_0^d$, we write $\|\alpha\|_1 := \sum_{j=1}^d \alpha_j$ and, for a sufficiently smooth function $f$ on $\mathcal{X} \subseteq \mathbb{R}^d$, we let $\partial^\alpha f(x):=\frac{\partial^{\beta_0}f(x)}{\prod_{j=1}^d\partial x_j^{\alpha_j}}$ when $\|\alpha\|_1 = \beta_0$; we also use the shorthand $x^\alpha := \prod_{j=1}^d x_j^{\alpha_j}$.   For $L>0$, 
we define the \emph{H\"{o}lder class} $\cH(\beta,L)$ on $\mathcal{X}$ to be the set of $\beta_0$-times differentiable functions $f:\cX\to\R$ satisfying
\[
\max_{\alpha\in\N_0^d:\|\alpha\|_1\leq\beta_0}\|\partial^\alpha f\|_\infty\leq L,\quad
        \max_{\alpha\in\N_0^d:\|\alpha\|_1=\beta_0}\sup_{x\neq y\in\cX}\frac{|\partial^\alpha f(x)-\partial^\alpha f(y)|}{\|x-y\|^{\beta-\beta_0}}\leq L.
\]

Given an index set $\mathcal{R}$, a sequence $(X_{n,R})$ of random variables for each $R \in \mathcal{R}$ and a deterministic, positive sequence $(a_{n,R})$ for each $R \in \mathcal{R}$, we write $X_{n,R} = o_{\mathcal{R}}(a_{n,R})$ if $\lim_{n\to\infty}\sup_{R \in \mathcal{R}}\Pr_R(|X_{n,R}|/a_{n,R} >\epsilon)=0$ for all $\epsilon>0$ and $X_{n,R}= O_{\mathcal{R}}(a_{n,R})$ if for any $\epsilon>0$, there exist $M_\epsilon, N_\epsilon>0$ such that $\sup_{n\geq N_\epsilon}\sup_{R\in\mathcal{R}}\Pr_R(|X_{n,R}|/a_{n,R}>M_\epsilon)<\epsilon$.  When $\mathcal{R}$ is a product set of the form $\mathcal{R} = \mathcal{S} \times \mathcal{T}$, we write, e.g., $X_{n,R} = o_{\mathcal{S},\mathcal{T}}(a_{n,R})$ instead of $X_{n,R} = o_{\mathcal{S} \times \mathcal{T}}(a_{n,R})$.  The standard normal distribution function is denoted by $\Phi$.

A \emph{kernel} $K:\R^d\to\R$ is a Borel measurable function satisfying $\int_{\mathbb{R}^d} K(\nu) \, d\nu = 1$.  For $\ell \in \mathbb{N}$, we say a kernel $K$ is of \emph{order} $\ell$ if $\int_{\R^d}K(\nu)\nu^\alpha \, d\nu=0$ for all $\alpha\in\N_0^d$ with $1\leq \|\alpha\|_1\leq\ell-1$. 
Given $p\in\mathbb{N}$, define $\bar{p} := \binom{d+p}{p}$ and $Q:\mathbb{R}^d \rightarrow \mathbb{R}^{\bar{p}}$ by  $Q(\nu):=\bigl(\frac{1}{\alpha!}\prod_{j=1}^d\nu_r^{\alpha_r}:\|\alpha\|_1\leq p\bigr)$,  with components in increasing lexicographic order.  We also define the quantities
\begin{gather*}
    u(K):=\int_{\R^d}K(\nu)Q(\nu) \, d\nu,
    \qquad
    s_t(K):=\int_{\R^d}K^t(\nu)Q(\nu)Q(\nu)^\top \, d\nu\quad(t\in\{1,2\}),
    \\
    \mu_\beta(K) := \int_{\R^d}K(\nu)\|\nu\|_{\beta}^{\beta} \, d\nu \quad (\beta > 0),
    \qquad
    R_2(K) := \int_{\R^d} K^2(\nu)\,d\nu.
\end{gather*}
For a bandwidth $h>0$, define the \emph{scaled kernel} by $K_h(\cdot):=K(\cdot/h)/h^d$ and also set $Q_h(\cdot):=Q(\cdot/h)$.

\section{The outrigger local polynomial estimator}\label{sec:methodology}

\subsection{Motivation: Local polynomial regression and the hardness of structure-free distributional adaptivity}\label{sec:lpr}

The purpose of this subsection is to provide background and informal arguments to motivate the introduction of our outrigger local polynomial estimator.  Let $(X_1,Y_1),\ldots,(X_n,Y_n)$ be independent and identically distributed covariate-response pairs in $\R^d \times \R$, regarded as copies of $(X,Y)$ satisfying $Y = f(X) + \varepsilon$ with $\mathbb{E}_P(\varepsilon \given X) = 0$.  Let $h > 0$ be a bandwidth, let $p \in \N_0$ denote a local polynomial degree and let $K$ be a kernel with support $\mathcal{B}_{0}(1)$.  The celebrated local polynomial estimator $\hat{f}^{\mathrm{LP}}$ is defined at the point $x_0 \in \R^d$ by $\hat{f}^{\mathrm{LP}}(x_0) := \mathrm{e}_1^\top\hat\theta^{\mathrm{LP}}$, where  $\hat\theta^{\mathrm{LP}}$ solves the linear estimating equation
\begin{equation}\label{eq:lp}
    \sum_{i=1}^nK_h(X_i-x_0)Q_h(X_i-x_0)\bigl(Y_i-Q_h(X_i-x_0)^\top\theta\bigr)=0,
\end{equation}
over $\theta\in\R^{\bar{p}}$.  
Suppose that the conditional variance $\sigma_P ^2(x_0):=\Var_P(Y\given X=x_0)$ is positive and finite and that $X$ has density $p_X$.  Then, provided $f$ is $\beta$-H\"older smooth and the bandwidth $h \equiv h_n$ satisfies $h\to0$ and $nh^d\to\infty$, the local polynomial estimator  $\hat{f}^{\mathrm{LP}}(x_0)$ admits the asymptotic decomposition
\begin{equation}\label{eq:lp-decomp}
    \hat{f}^{\mathrm{LP}}(x_0)-f(x_0)\overset{d}{=}B(f,x_0,K,h)h^{\beta^*}+\frac{1}{\sqrt{nh^d}}\biggl\{\frac{R_2(K)\sigma_P^2(x_0)}{p_X(x_0)}\biggr\}^{1/2}N(0,1)
    +o_P\biggl(h^{\beta^*} +\frac{1}{\sqrt{nh^d}}\biggr)
    ,
\end{equation}
for a bias term $B(f,x_0,K,h)$ satisfying $\limsup_{n \rightarrow \infty} |B(f,x_0,K,h)| < \infty$, and where $\beta^*:=\beta\wedge(p+1)$. 
See, e.g.,~\citet{fan-lp-asymptotics} for the case where $f$ is twice differentiable and~\citet{rupertwand} for extensions to $d \geq 2$ and infinitely differentiable $f$.

Suppose for now that the conditional density $p_{\varepsilon|X}(\cdot \given x)$ of the errors, or equivalently the conditional score function $\rho(\cdot \given x)$, given by $\rho(\varepsilon \given x) := \frac{\partial}{\partial \varepsilon} \log p_{\varepsilon|X}(\varepsilon \given x)$ were known. 
The local likelihood estimator~\citep{loc-lik} is defined as $\hat{f}^{\mathrm{LL}}(x_0)=\mathrm{e}_1^\top\hat\theta^{\mathrm{LL}}$, where $\hat\theta^{\mathrm{LL}}$ is a zero of the function $g:\mathbb{R}^{\bar{p}} \rightarrow \mathbb{R}^{\bar{p}}$, given by 
\begin{equation}\label{eq:lp-ora-conditional}
    g^{\mathrm{LL}}(\theta):=\frac{1}{n}\sum_{i=1}^nK_h(X_i-x_0)Q_h(X_i-x_0)\rho\bigl(Y_i-Q_h(X_i-x_0)^\top\theta\biggiven X_i\bigr).
\end{equation}
Under similar conditions, and provided the \emph{conditional Fisher information} 
\begin{equation}\label{eq:Fisher-info}
    i_P(x_0) := \E_P\bigl(\rho^2(\varepsilon\given X)\biggiven X=x_0\bigr)
    \end{equation}
    is positive and finite, we can write 
\begin{equation}\label{eq:rho-est-oracle}
    \hat{f}^{\mathrm{LL}}(x_0)-f(x_0)\overset{d}{=}B(f,x_0,K,h)h^{\beta^*}+\frac{1}{\sqrt{nh^d}}\bigg\{\frac{R_2(K)}{i_P(x_0)p_X(x_0)}\bigg\}^{1/2}N(0,1)
    + o_P\biggl(h^{\beta^*}+\frac{1}{\sqrt{nh^d}}\biggr)
    ;
\end{equation}
see Theorem~\ref{thm:decomp}.  Thus, the dominant bias term is the same as in~\eqref{eq:lp-decomp}, but the dominant variance term is reduced, since $\sigma_P^2(x_0) \geq 1/i_P(x_0)$, with equality if and only if $\varepsilon\given \{X=x_0\}\sim N\bigl(0,\sigma_P^2(x_0)\bigr)$; see Lemma~\ref{lem:score-min}. 

In practice, the conditional error distribution is rarely known, so suppose instead that we have access to an estimator $\hat\rho(\cdot \given x)$ of $\rho(\cdot \given x)$. For simplicity of exposition here, we assume that $\hat\rho$ has been constructed via independent auxiliary data. A naive plug-in estimator motivated by~\eqref{eq:rho-est-oracle} would take $\hat{f}^{\mathrm{plug.in}}(x_0)=\mathrm{e}_1^\top\hat\theta^{\mathrm{plug.in}}$, where $\hat\theta^{\mathrm{plug.in}}$ is a zero of the function
\begin{equation}\label{eq:errorO}
    g^{\mathrm{plug.in}}(\theta):=\frac{1}{n}\sum_{i=1}^nK_h(X_i-x_0)Q_h(X_i-x_0)\hat\rho\bigl(Y_i-Q_h(X_i-x_0)^\top \theta\biggiven X_i\bigr).
\end{equation}
Now let $r_n:=\frac{1}{\sqrt{nh^d}}+h^{\beta}$.  Provided that $\rho(\cdot\given x)$ is differentiable for each $x$, a Taylor expansion yields that for every $C > 0$,
\begin{align*}
    \sup_{\theta\in\R^{\bar{p}}:\|\theta-f(x_0)\mathrm{e}_1\|\leq C r_n}&\biggl\|g^{\mathrm{LL}}(\theta)
     - \frac{1}{n}\sum_{i=1}^n K_h(X_i-x_0)Q_h(X_i-x_0)\rho\bigl(\varepsilon_i\biggiven X_i\bigr)
     \\
    &-
    \frac{1}{n}\sum_{i=1}^n K_h(X_i-x_0)Q_h(X_i-x_0)\rho'\bigl(\varepsilon_i\biggiven X_i\bigr)\bigl(f(X_i)-Q_h(X_i-x_0)^\top\theta\bigr)
    \biggr\| = o_P(r_n);
\end{align*}
a similar property holds for $g^{\mathrm{plug.in}}$, but with $\rho$ replaced by $\hat{\rho}$.  Noting that $\mathbb{E}_P\bigl\{\rho(\varepsilon \given X)\biggiven X\bigr\} = 0$, and provided $\hat{\rho}'(\cdot \given x)$ is a uniformly consistent estimator of $\rho'(\cdot \given x)$ over $\mathcal{B}_{x_0}(Cr_n)$, we can therefore expect that
\begin{equation*}
    \sup_{\theta \in \mathbb{R}^{\bar{p}}:\|\theta-f(x_0)\mathrm{e}_1\|\leq C r_n} \Bigl\|
    g^{\mathrm{plug.in}}(\theta)- g^{\mathrm{LL}}(\theta)-
     \E_P\bigl\{K_h(X-x_0)Q_h(X-x_0)\hat\rho(\varepsilon\given X)\bigr\}
    \Bigr\| = o_P(r_n).
\end{equation*}
Thus in the typical case that $\mathrm{Bias}_P(\hat\rho(\varepsilon\given X)\given X=x_0)\neq0$ the term
\begin{equation}\label{eq:bias}
    \E_P\bigl\{K_h(X-x_0)Q_h(X-x_0)\hat\rho(\varepsilon\given X)\bigr\}
    = \bigl(1+o(1)\bigr)u(K)p_X(x_0)\,\mathrm{Bias}_P\bigl(\hat\rho(\varepsilon\given X)\biggiven X=x_0\bigr)
\end{equation}
biases the naive plug-in estimating equation, and hence the resulting estimator.  A further complication in seeking to estimate this bias is the fact that we do not have direct access to independent realisations of the errors, and would have to rely on fitted residuals from a pilot estimate.  This means that, outside some special cases discussed below, the order of the bias is reflected in the rate at which we can hope to estimate the conditional score at $x_0$.  Since this is a function of $d+1$ variables and involves a partial derivative of the error density, in general we do not expect to be able to estimate this function at rate $1/r_n$, and the bias of the corresponding estimator will typically be of the same order as the bias of the conditional score estimate.  This explains the empirical findings on the failure of the naive plug-in estimator illustrated in Figure~\ref{fig:intro}.

As mentioned above, there are special cases where the naive plug-in estimator may still perform well.  
One is where the conditional errors are known to be symmetric, so that $\rho(\cdot\given x)$ is antisymmetric for each~$x$. In this setting, if the estimator $\hat\rho(\cdot\given x)$ is constructed to also be antisymmetric, then the bias term~\eqref{eq:bias} is exactly zero~\citep[Example 25.27]{vandervaart}. Another case may be when the errors are independent of the covariates, in which case $\rho(\cdot\given x)=\rho(\cdot)$ is a univariate function, although even in this setting the intricacies of using fitted residuals instead of oracle residuals in this procedure still makes this problem non-trivial. These settings effectively assume away the difficulty of distributional adaptivity; the challenge taken up in this work is to achieve distributional adaptation in a structure-free setting.

\subsection{The outrigger}\label{sec:outrigging}

We are now in a position to provide the intuition behind our outrigger local polynomial estimator; a formal outline of our proposed methodology is given in Algorithm~\ref{alg:outrigger} below.  In addition to the usual ingredients of a local polynomial estimator $\hat{f}^{\mathrm{LP}}$ with kernel $K$ and bandwidth $h$, our outrigger estimator at $x_0$ involves an \emph{outrigger kernel}~$\kappa_\lambda:\mathbb{R}^d \rightarrow \mathbb{R}$, indexed by $\lambda > 1$ and supported on $\mathcal{B}_{x_0}(\lambda)\setminus\mathcal{B}_{x_0}(1)$.  By analogy with the original kernel, we also define the \emph{scaled outrigger kernel} $\kappa_{h,\lambda}(\cdot) := h^{-d}\kappa_\lambda(\cdot/h)$. 
It is now convenient to define, for $i \in [n]$, \emph{population-level outrigger weightings}
\begin{equation*}
    \varphi_{h,\lambda}(X_i-x_0) := K_h(X_i-x_0)Q_h(X_i-x_0)-\mu(x_0)\kappa_{h,\lambda}(X_i-x_0),
    \quad
    \mu(x_0) :=\frac{\E_P\bigl\{K_h(X-x_0)Q_h(X-x_0)\bigr\}}{\E_P\bigl\{\kappa_{h,\lambda}(X-x_0)\bigr\}}
    ,
\end{equation*}
and their empirical analogues 
\begin{equation*}
    \hat\varphi_{h,\lambda}(X_i-x_0) := K_h(X_i-x_0)Q_h(X_i-x_0)-\hat\mu(x_0)\kappa_{h,\lambda}(X_i-x_0),
\end{equation*}
where $\hat\mu(x_0)$ is an empirical estimator of $\mu(x_0)$, assumed for now to be constructed using auxiliary data.  An illustration of the population-level outrigger weightings in the case of a local constant estimator, is given in the lower panel of Figure~\ref{fig:outrigger-kernel}.   By construction,~$\E\bigl(\varphi_{h,\lambda}(X-x_0)\bigr)=0$. Given a pilot estimator, which we take to be the standard local polynomial estimator $\hat{f}^{\mathrm{LP}}$, the starting point for our procedure is to consider  
\begin{equation}\label{eq:orthog-scoring}
    g^{\mathrm{Outrig}}(\theta):=\frac{1}{n}\sum_{i=1}^n\hat\varphi_{h,\lambda}(X_i-x_0)\hat\rho\Bigl(Y_i-Q_h(X_i-x_0)^\top\theta\,\ind_{\innerr}(X_i)-\tilde{f}(X_i)\ind_{\mathcal{B}_{x_0}(\lambda h)\setminus\innerr}(X_i)\Biggiven X_i\Bigr),
\end{equation}
where    
\begin{equation}
    \tilde{f}(X_i) := \hat{f}^{\mathrm{L{P}}}(X_i)+\hat{c}(x_0) \quad \text{and} \quad 
    \hat{c}(x_0):=\frac{\sum_{j=1}^n\kappa_{h,\lambda}(X_j-x_0)\bigl(Y_j-\hat{f}^{\mathrm{LP}}(X_j)\bigr)}{\sum_{j=1}^n\kappa_{h,\lambda}(X_j-x_0)}.
\end{equation}
\begin{figure}[ht]
    \centering
    \includegraphics[width=0.7\linewidth]{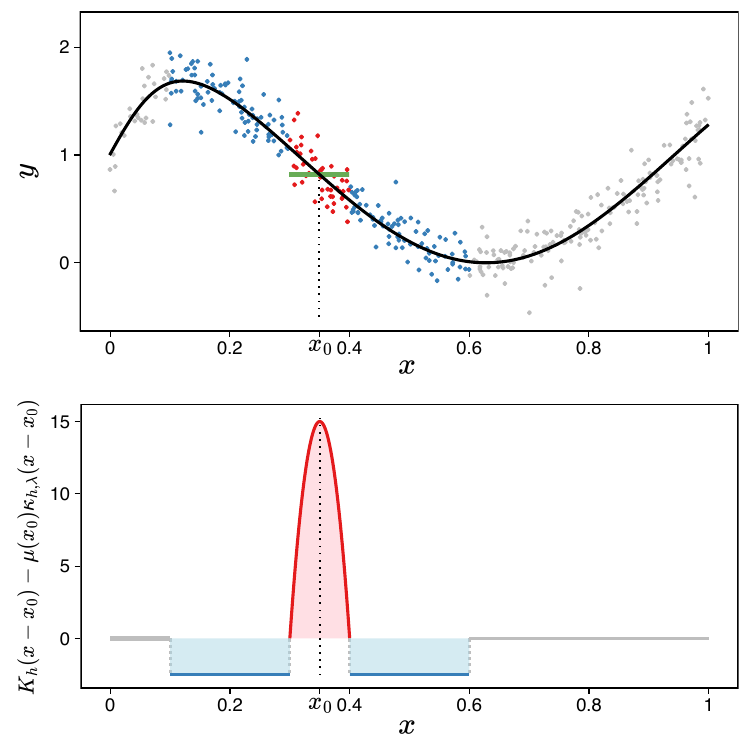}
    \caption{Illustration of a local constant estimator ($p=0$) in the single covariate ($d=1$) case at  $x_0=0.35$ with bandwidth $h=0.05$ and outrigger parameter $\lambda=5$. The solid black curve is the true regression function. The green line shows the outrigger local constant fit at $x_0$. The bottom diagram shows the orthogonal combination of `inner and outer region' kernels about $x_0$. In our schematic,  $K(\nu)=\tfrac{3}{4}\max(1-\nu^2,0)$ is the Epanechnikov kernel and $\kappa_\lambda(\nu)=\tfrac{1}{2(\lambda-1)}\ind_{\mathcal{B}_{x_0}(\lambda)\setminus\mathcal{B}_{x_0}(1)}$ is a uniform kernel. The red points lie within the inner region $\mathcal{B}_{x_0}(h)$, which contribute to our estimating equation via the kernel $K_h(x-x_0)$. The blue points lie within the regions $\mathcal{B}_{x_0}(\lambda h)\setminus\mathcal{B}_{x_0}(h)$ contributing to our estimating equation via the kernel $\kappa_{h,\lambda}(x-x_0)$; the scaling $\mu(x_0)$ is such that the expectation of the function in the bottom diagram is zero. Each of $K$ and $\kappa_\lambda$ are kernels of order 2. The grey points do not contribute to the estimator at $x_0$.}
    \label{fig:outrigger-kernel}
\end{figure}

The function $g^{\mathrm{Outrig}}$ has three differences with $g^{\mathrm{LL}}$ defined in~\eqref{eq:lp-ora-conditional}, the first of which is that the unknown score $\rho$ is replaced with the estimate $\hat{\rho}$.  We know that on its own, this naive replacement yields an estimator with significant bias, so the other two alterations are introduced to compensate.  The weightings $K_h(X_i - x_0)Q_h(X_i-x_0)$ are replaced with empirical outrigger weightings $\hat\varphi_{h,\lambda}(X_i-x_0)$, designed to have population mean close to zero, so as to stabilise the score estimator.  Further, since these empirical outrigger weightings are generally non-vanishing in the larger region $\mathcal{B}_{x_0}(\lambda h)$, we now require meaningful residuals in the support of the outrigger kernel, and these are provided via the intermediate estimator $\tilde{f}$. This intermediate estimator is constructed by debiasing the pilot estimator $\hat{f}^{\mathrm{LP}}$ via the addition of the term $\hat{c}(x_0)$, which represents an in-sample weighted average estimate over the data in the outrigger's support of the average pilot estimator residual.  The net effect of these modifications of $g^{\mathrm{LL}}$ is that, as shown as a significant part of the proof of Theorem~\ref{thm:decomp} below, in the asymptotic regime where $\lambda\to\infty$ and $\lambda h\to0$, we have for every $C > 0$ that 
\begin{align}\label{eq:Out-LL-1}
    \sup_{\theta \in \mathbb{R}^{\bar{p}}:\|\theta-\theta_0\|\leq C r_n} \bigl\|
    g^{\mathrm{Outrig}}(\theta)-g^{\mathrm{LL}}(\theta) 
   \bigr\| = o_P(r_n),
\end{align}
where $\theta_0 := f(x_0)\mathrm{e}_1$.  Moreover, writing $D$ for the derivative operator and $\pi_h(\cdot):=K_h(\cdot)Q_h(\cdot)Q_h(\cdot)^\top$, we have
\begin{align}
    Dg^{\mathrm{LL}}(\theta_0) &=
    \frac{1}{n}\sum_{i=1}^n \pi_h(X_i-x_0) \rho'\bigl(Y_i-Q_h(X_i-x_0)^\top\theta_0\biggiven X_i\bigr)
    \notag
    \\
    &=p_X(x_0) \E_P\bigl\{\rho'(\varepsilon\given X)\biggiven X=x_0\bigr\} s_1(K)+o_P(1),   
    \label{eq:Out-LL-2}
\end{align}
where, under mild conditions, $\E_P\bigl\{\rho'(\varepsilon\given X)\biggiven X=x_0\bigr\} = -\E_P\bigl\{\rho^2(\varepsilon\given X)\biggiven X=x_0\bigr\} < 0$; see~\eqref{Eq:FITwoForms} in the proof of Lemma~\ref{lem:score-min}.  Finally, again under mild smoothness conditions and provided that $\hat{\rho}'(\cdot \given x_0)$ is a uniformly consistent estimator of $\rho'(\cdot \given x_0)$ over $\mathcal{B}_{x_0}(Cr_n)$, 
\begin{align}
    \sup_{\theta \in \mathbb{R}^{\bar{p}}:\|\theta-\theta_0\|\leq C r_n} &\bigl\|D g^{\mathrm{Outrig}}(\theta)
    -
    D g^{\mathrm{LL}}(\theta)
    \bigr\|_{\mathrm{op}}
    \notag
    \\
    &=
    \sup_{\theta \in \mathbb{R}^{\bar{p}}:\|\theta-\theta_0\|\leq C r_n} \biggl\|\frac{1}{n}\sum_{i=1}^n \pi_h(X_i-x_0)(\hat\rho'-\rho')\bigl(Y_i-Q_h(X_i-x_0)^\top\theta\biggiven X_i\bigr)\biggr\|_{\mathrm{op}}
    =
    o_P(1).
    \label{eq:Out-LL-3}
\end{align}
Assuming for simplicity here that the roots of our estimating equations are unique, we see from~\eqref{eq:Out-LL-1},~\eqref{eq:Out-LL-2} and~\eqref{eq:Out-LL-3} that $\hat{\theta}^{\mathrm{Outrig}}$ in Algorithm~\ref{alg:outrigger} satisfies $\hat{\theta}^{\mathrm{Outrig}} - \hat\theta^{\mathrm{LL}} = o_P(r_n)$, and we are further able to establish that $\hat{\theta}^{\mathrm{LL}} - \theta_0 = O_P(r_n)$, so that $\hat{\theta}^{\mathrm{Outrig}}$ and $\hat{\theta}^{\mathrm{LL}}$ have the same asymptotic behaviour. 

This simplified presentation is elaborated and formalised in our theory in Section~\ref{sec:theory} to follow.

\subsubsection{Implementation details}

As we typically do not have access to auxiliary data with which to estimate the conditional score function, we employ a modified form of $\mathcal{K}$-fold cross-fitting~\citep{dml} in Algorithm~\ref{alg:outrigger}.  Specifically, the score estimator $\hat\rho_k$ on the $k$th fold and outrigger weighting quantity $\hat\mu_k$ (used for score stabilisation) are constructed out-of-fold, while $\hat{c}_k$ (used for pilot stabilisation) is constructed in-fold.  Examples of conditional score estimates that could be used as an input in Algorithm~\ref{alg:outrigger} were discussed in Section~\ref{sec:lit-rev}; see also Section~\ref{sec:numerical-implementation}. 

The computation of $\hat\theta^{\mathrm{Outrig}}$ in Algorithm~\ref{alg:outrigger} involves solving a non-linear estimating equation.  
In practice we implement Fisher scoring steps, initialised at the pilot estimator~$\hat\theta^{\mathrm{LP}}$, via the updates
\[
    \theta \leftarrow \theta - \hat{\mathcal{J}}(\theta)^{-1}\hat{\mathcal{S}}(\theta),
\]
    where
    \begin{align*}
    \hat{\mathcal{J}}(\theta) &:= \sum_{k=1}^{\mathcal{K}}\sum_{i
    \in\cI_k}\pi_h(X_i-x_0)\hat\rho_k^2\big(Y_i-Q_h(X_i-x_0)^\top\theta \biggiven X_i\big),
    \\
    \hat{\mathcal{S}}(\theta) &:= 
    \sum_{k=1}^{\mathcal{K}}\sum_{i
    \in\cI_k}\hat\varphi_{h,\lambda,k}(X_i-x_0)\hat\rho_k\big(Y_i-Q_h(X_i-x_0)^\top\theta\,\ind_{\mathcal{B}_{x_0}(h)}(X_i)-\tilde{f}_k(X_i)\ind_{\outerr}\biggiven X_i\big),
    \\
    \hat\varphi_{h,\lambda,k}(\cdot) &:= K_h(\cdot)Q_h(\cdot)-\hat{\mu}_k(x_0)\kappa_{h,\lambda}(\cdot),
    \qquad
    \tilde{f}_k(X_i) := \hat{f}_k^{\mathrm{LP}}(X_i)+\hat{c}_k(x_0),
\end{align*}
where $(\cI_k)_{k\in[\mathcal{K}]}$ forms a partition of $[n]$, where $\cI_k$ indexes the observations in the $k$th fold and where $\hat{f}_k^{\mathrm{LP}}$ is the standard local polynomial estimator computed on the $k$th fold.

{
\RestyleAlgo{ruled}
\begin{algorithm}[H]
\KwIn{Data $(X_i,Y_i)_{i\in[n]}$; point $x_0\in\R^d$; bandwidth $h > 0$; degree of polynomial $p \in \mathbb{N}_0$; kernel $K:\R^d\to\R$ supported on $\mathcal{B}_{x_0}(1)$; outrigger parameter\footnote{
The quantity $\lambda_0(K) \geq 1$, depending only on the kernel $K$ chosen by the practitioner, is defined in Assumption~\ref{ass:kappa}.
} $\lambda\in \bigl(\lambda_0(K),\infty\bigr]$; outrigger kernel $\kappa_{\lambda}:\R^d\to\R$  supported on $\mathcal{B}_{x_0}(\lambda  )\setminus\mathcal{B}_{x_0}(1)$; number $\mathcal{K}\geq2$ of folds for cross-fitting.
}

\smallskip

Partition $[n]$ into $\mathcal{K}$ disjoint sets $(\cI_k)_{k\in[\mathcal{K}]}$ of approximately equal size. 

\For{$k\in[\mathcal{K}]$}{
    $\hat f^{\mathrm{LP}}_k(\cdot) := \mathrm{e}_1^\top\bigl(\sum_{i\in\cI_k^c}K_{h}(X_i-\cdot)Q_{h}(X_i-\cdot)Q_{h}(X_i-\cdot)^\top\bigr)^{-1}\sum_{i\in\cI_k^c}K_{h}(X_i-\cdot)Q_{h}(X_i-\cdot)Y_i$.
    
    Construct an estimator $\hat\rho_k$ of $\rho$ using $(\hat\varepsilon_i,X_i)_{i\in\cI_k^c}:=\bigl(Y_i-\hat{f}^{\mathrm{LP}}_k(X_i),\,X_i\bigr)_{i\in\cI_k^c}$. 

    $\hat\mu_k(x_0) := \big(\sum_{i\in\cI_k^c}\kappa_{h,\lambda}(X_i-x_0)\big)^{-1}\big(\sum_{i\in\cI_k^c}K_h(X_i-x_0)Q_h(X_i-x_0)\big)$.
    
    $\hat{c}_k(x_0):=\bigl(\sum_{i\in\cI_k}\kappa_{h,\lambda}(X_i-x_0)\bigr)^{-1}\bigl\{\sum_{i\in\cI_k}\kappa_{h,\lambda}(X_i-x_0)\bigl(Y_i-\hat{f}_k^{\mathrm{LP}}(X_i)\bigr)\bigr\}$.

}

Let $\Theta \subseteq \R^{\bar{p}}$ denote the set of zeros of 
\begin{align}
    \theta \mapsto \sum_{k=1}^\mathcal{K} &\sum_{i\in\cI_k} 
    \Bigl\{K_h(X_i-x_0)Q_h(X_i-x_0)-
    \hat\mu_k(x_0)\kappa_{h,\lambda}(X_i-x_0)
    \Bigr\}
    \notag
    \\ &\hspace{-0.66cm} \cdot
    \hat\rho_k\Big(Y_i-Q_h(X_i-x_0)^\top\theta\,\ind_{\innerr}(X_i)-\big(\hat{f}^{\mathrm{LP}}_k(X_i)+\hat{c}_k(x_0)\big)\ind_{\outerr}(X_i)\Biggiven X_i\Big),
    \label{eq:algorithm-est-eqn}
\end{align}
and take\footnote{Here, $\mathrm{sargmin}$ denotes the smallest element in the lexicographic ordering of the $\argmin$ set.} $\hat{\theta}^{\mathrm{Outrig}} := \sargmin_{\theta \in \Theta} \|\theta - \hat\theta^{\mathrm{LP}}\|_\infty$, where $\hat\theta^{\mathrm{LP}}:=\bigl(\sum_{i=1}^nK_h(X_i-x_0)Q_h(X_i-x_0)Q_h(X_i-x_0)^\top \bigr)^{-1}\sum_{i=1}^n K_h(X_i-x_0)Q_h(X_i-x_0)Y_i$.

\KwOut{The outrigger estimator $\hat{f}^{\mathrm{Outrig}}(x_0) :=\mathrm{e}_1^\top\hat\theta^{\mathrm{Outrig}}$.}
\caption{The outrigger local polynomial estimator}\label{alg:outrigger}
\end{algorithm}
}

\section{Main results}\label{sec:theory}

Throughout, we will assume that our data $(X_1,Y_1),\ldots,(X_n,Y_n)$ are independent copies of a pair $(X,Y)$, where $X$ takes values in $\R^d$ and 
\begin{equation}\label{eq:y=f(x)+e}
Y = f(X) + \varepsilon,
\end{equation}
with $\mathbb{E}_P(\varepsilon \given X) = 0$.  The joint distribution $P$ of $(X,Y)$ may be described by a triple $(P_X,P_{\varepsilon|X},f)$, where $P_X$ denotes the marginal distribution of $X$, where $P_{\varepsilon|X} = (P_{\varepsilon|x})_{x \in \R^d}$ denotes a disintegration of $P$ into conditional distributions on $\mathbb{R}$ \citep[e.g.][Section~10.9.2]{samworth-shah} and where $f:\R^d \rightarrow \mathbb{R}$ denotes the regression function, given by $f(x) := \mathbb{E}_P(Y \given X=x)$.  The convergence in probability statements in Assumptions~\ref{ass:DGM}~and~\ref{ass:bandwidth-kernels-and-co} below are stated uniformly over the given index sets so as to facilitate uniform conclusions; see Section~\ref{sec:notation} for formal definitions.

\begin{assumption}[Regularity of the data generating mechanism]\label{ass:DGM}
     Let $\mathcal{P}$ denote the class of distributions $P = (P_X,P_{\varepsilon|X},f)$ for $(X,Y)$ that satisfy:
    \begin{enumerate}[label=(A1.\arabic*), leftmargin=*, align=left]
        \item (Covariate distribution) There exist $\beta_X,L_X > 0$, as well as a compact set $\mathcal{X} \subseteq \mathbb{R}^d$ and an open neighbourhood $\mathcal{X}^\circ$ of $\mathcal{X}$, such that the restriction of $P_X$ to $\mathcal{X}^\circ$ is absolutely continuous with respect to Lebesgue measure on $\R^d$, with Radon--Nikodym derivative $p_X \in \cH(\beta_X,L_X)$ satisfying 
        $0 < c_X \leq p_X(x)\leq C_X$ for all $x\in\cX^\circ$.  We assume that our estimation point $x_0$ of interest belongs to $\mathcal{X}$.
        \label{ass:covariates}
        \item (Smoothness of regression function) $f\in\cH(\beta, L)$ on $\mathcal{X}^\circ$ for some $\beta,L > 0$.
        \item (Score function smoothness) We assume that $P_{\varepsilon|x}$ is absolutely continuous with respect to Lebesgue measure for each $x \in \mathcal{X}^\circ$, with density $p \equiv p_{\varepsilon|x}$.  Moreover, we assume that the \emph{score function} $\rho:\R \times \mathcal{X}^\circ \to\R$, given by $\rho(e \given x) := p'(e \given x)/p(e \given x)$, is well-defined and that $e \mapsto \rho(e \given x)$ is differentiable for each $x \in \mathcal{X}^\circ$, with uniformly continuous derivative in the sense that $\sup_{e \in \mathbb{R}} \sup_{x \in \mathcal{X}^\circ} \bigl|\rho'(e+t\given x) - \rho'(e \given x)\bigr| \rightarrow 0$ as $t \rightarrow 0$. 
        Further suppose that $\E_P\bigl(\rho^2(\varepsilon\given X)\biggiven X=\cdot\bigr)$ is uniformly continuous.
        \label{ass:score}
        \item (Score function estimation) We have access to an estimator $\hat{\rho}$ of $\rho$ satisfying, for almost all realisations of data on which it is constructed, that 
        \begin{enumerate}[label=(\roman*)]
            \item $\E_P\bigl\{\hat\rho(\varepsilon\given X)\given X=\cdot \bigr\},\, \E_P\bigl\{\hat\rho'(\varepsilon\given X)\given X=\cdot \bigr\} \in \cH(\beta_\cE, L_\cE)$ for some $\beta_\cE > 0$, $L_\cE \geq 0$;
            \item There exists $\epsilon>0$ such that $\sup_{|\tau|\leq\epsilon}\E_P\bigl\{(\hat\rho-\rho)^2(\varepsilon+\tau\given X)\biggiven\hat\rho,X=\cdot\bigr\}=o_{\mathcal{P},\mathcal{X}}(1)$.
            \item $e \mapsto \hat{\rho}(e \given x)$ is differentiable for each $x\in\cX^\circ$, with uniformly continuous derivative $\hat{\rho}'(\cdot\given x)$.  Moreover, $\E_P\bigl\{(\hat\rho'-\rho')^2(\varepsilon\given X)\biggiven\hat\rho,X=x\bigr\} = o_{\mathcal{P},\mathcal{X}}(1)$.
            \label{ass:rho'-consistency}
            \item 
            $\E_P\bigl\{(\hat\rho'-\rho')^2(\varepsilon\given X)\biggiven X=\cdot\bigr\}$ is uniformly continuous. 
        \end{enumerate}
        \label{ass:score-estimation}
        \item (Moment conditions) 
        There exist values 
        $C_1,C_2,c_1,C_3,c_3,\delta,
        \epsilon>0$ such that  $\sup_{|\tau|\leq\epsilon}\E_P\bigl(\bigl|\rho(\varepsilon+\tau\given X)\bigr|^{2+\delta}\biggiven X=x\bigr) \leq C_1$, $\sup_{|\tau|\leq\epsilon}\E_P\bigl(\rho'(\varepsilon+\tau\given X)^2\given X=x\bigr) \leq C_2$, $\bigl|\E_P\bigl(\rho^2(\varepsilon\given X)\biggiven X=x\bigr)\bigr| \geq c_1$, $\E_P(\varepsilon^{2}\given X=x) \geq c_3$ and  $\E_P(|\varepsilon|^{2+\delta}\given X=x) \leq C_3$ for $P_X$-almost all $x\in\cX^\circ$. 
        \label{ass:exp-bounds}
    \end{enumerate}
\end{assumption}

Although there are many parts to Assumption~\ref{ass:DGM}, the overall restrictions on the data generating mechanism remain relatively mild.  At this point we do not impose constraints on the (relative) H\"older smoothness levels of the different functions involved, though these will appear as part of Assumption~\ref{ass:bandwidth-kernels-and-co} below.  

In addition to these assumptions on the data generating mechanism, we constrain the practitioner-chosen inputs to the outrigger local polynomial estimation algorithm as follows.
\begin{assumption}[Outrigger estimator construction]\label{ass:bandwidth-kernels-and-co}
The inputs to Algorithm~\ref{alg:outrigger} satisfy:
\begin{enumerate}[label=(A2.\arabic*), leftmargin=*, align=left]
    \item (Primary kernel) The kernel $K$ is bounded, supported on $\mathcal{B}_{0}(1)$, and of order $p+1\geq\ceil{\beta_{\cE}\wedge\beta_X}$. 
    We further suppose that $s_1(K) \in \mathbb{R}^{\bar{p} \times \bar{p}}$ is invertible, and 
    $\mu_{\beta^*}(K)\neq0$, where $\beta^* := \beta \wedge (p+1)$.
    \label{ass:kernel}
    \item (Outrigger kernel) 
    The outrigger kernel $\kappa_\lambda$ is a bounded kernel, supported on $\mathcal{B}_0(\lambda)\setminus\mathcal{B}_0(1)$, and of order $\ceil{\beta_\cE\wedge\beta_X}$.  
    We suppose that  $\lambda\mapsto R_2(\kappa_\lambda)$ is strictly decreasing on $(1,\infty)$, with $R_2(\kappa_\lambda) \rightarrow 0$ as $\lambda \rightarrow \infty$, and let $\lambda_0(K) \in (1,\infty)$ be the unique solution to $R_2(\kappa_\lambda)=R_2(K)$. 
    In addition, we assume that  $\sup_{\lambda\geq\lambda_0(K)}\int_{\R^d}|\kappa_\lambda|<\infty$.
    \label{ass:kappa}
    \item (Bandwidth)
    For each $n\in\mathbb{N}$, the bandwidth $h=h_n$ lies in an interval $\cH_n$ and the outrigger parameter $\lambda=\lambda_n$ lies in an interval $\Lambda_n$, where
    \begin{equation}\label{ass:smoothness}
        \sup_{(h,\lambda) \in\cH_n \times \Lambda_n} \biggl(\lambda h \vee \frac{1}{nh^d}\vee \frac{L_\cE(\lambda h)^{\beta_\cE\wedge\beta_X}}{h^{\beta^*}}\biggr)\to0,
    \end{equation}
    as $n\to\infty$, and moreover  $\inf_{n\in\N}\inf\Lambda_n>\lambda_0(K)$.  
    Let $\cH$ and $\Lambda$ denote the set of all such sequences $(h_n)_{n\in\N}$ and $(\lambda_n)_{n\in\N}$ respectively. 
    \label{ass:h}
    \item (Cross-fitting) The number of folds $\mathcal{K}$ used for cross-fitting is deterministic and a bounded function of $n$.
    \label{ass:cross-fitting}
    \item (Identifiability) There exist $0<\eta_1<\eta_2$ such that, for each $k \in [\mathcal{K}]$ and on a sequence of events of probability $1-\supoP(1)$, the estimated score functions $\hat\rho_k$ have the property that there is a zero of the function~\eqref{eq:algorithm-est-eqn} in Algorithm~\ref{alg:outrigger} in the ball $\bigl\{\theta \in \mathbb{R}^{\bar{p}}:\|\theta-f(x_0)\mathrm{e}_1\|_\infty\leq\eta_1\bigr\}$, that is unique on the ball $\bigl\{\theta \in \mathbb{R}^{\bar{p}}:\|\theta-f(x_0)\mathrm{e}_1\|_\infty\leq\eta_2\bigr\}$.
    \label{ass:unique-root-ish}
\end{enumerate}
\end{assumption}
The outrigger kernel can take similar forms to classical kernels supported on $\mathcal{B}_{0}(1)$, except that it should be supported over the outer region $\mathcal{B}_{x_0}(\lambda )\setminus\mathcal{B}_{x_0}(1)$. For example for a second-order outrigger kernel one could take the uniform kernel
$\kappa_\lambda=\frac{v_{d,q}}{\lambda^d-1}\ind_{\mathcal{B}_{x_0}(\lambda)\setminus\mathcal{B}_{x_0}(1)}$ for $\lambda > 1$, where $v_{d,q}
:=\frac{\Gamma(1+d/q)}{\{2\Gamma(1+1/q)\}^{d}}$ denotes the Lebesgue measure of the unit $\ell_q$ ball in $\R^d$. Then indeed $\lambda\mapsto R_2(\kappa_\lambda)=\frac{v_{d,q}}{\lambda^d-1}$ is strictly decreasing, so for any $\lambda>\bigl(1+\frac{v_{d,q}}{R_2(K)}\bigr)^{1/d} =: \lambda_0(K)$ we have that $\frac{R_2(\kappa_\lambda)}{R_2(K)}<1$.  
If in addition $d=1$,  
then $\lambda_0(K)=1+\tfrac{1}{2R_2(K)}$, which for the Epanechnikov kernel $\nu\mapsto\frac{3}{4}\max(1-\nu^2,0)$ evaluates as $11/6$.

The requirement~\eqref{ass:smoothness} is our critical assumption that facilitates full distributional adaptivity; a sufficient condition for the existence of intervals $\mathcal{H}_n$ and $\Lambda_n$ satisfying this property is that  
\begin{equation}\label{eq:suff-smoothness}
    \beta_\cE\wedge\beta_X > \beta^* =\beta\wedge(p+1).
\end{equation}
The left-hand side of~\eqref{eq:suff-smoothness} represents the minimum of the H\"older smoothnesses of the conditional density of the errors $\varepsilon\given X$ and the expectations of the estimator of the conditional score function and its derivative.  Common practical choices of $p$ are $p=0$, corresponding to the local constant estimator and $p=1$, corresponding to the local linear estimator, so in these cases, asking for $\beta_\cE\wedge\beta_X > p+1$ is relatively mild.  Moreover, the requirement is further weakened in cases where the underlying regression function is relatively rough (so that $\beta < p+1$).  We finally mention that we can also find intervals~$\mathcal{H}_n$ and $\Lambda_n$ satisfying~\eqref{ass:smoothness} when either $\varepsilon$ and $X$ are independent or when both $\varepsilon \given X$ is symmetric and $\hat\rho(\varepsilon\given x)$ is antisymmetric, because in both of these cases $L_\cE = 0$.  However, these assumptions are substantially stronger than what is necessary for~\eqref{ass:smoothness}, as outlined above. 

Some intuition regarding Assumption~\ref{ass:unique-root-ish} was provided in Section~\ref{sec:outrigging}: after replacing the estimator~$\hat{\rho}_k$ with its population-level analogue $\rho$ in~\eqref{eq:algorithm-est-eqn}, the estimating equation has asymptotically invertible derivative at $\theta_0 = f(x_0)\mathrm{e}_1$; see~\eqref{eq:Out-LL-2}.  Moreover, under our conditions, the derivative of the sample version in~\eqref{eq:algorithm-est-eqn} is a smooth and uniformly consistent estimator of this population-level analogue, so it is reasonable to expect that with high probability there will be a unique zero of the estimating equation~\eqref{eq:algorithm-est-eqn} in a small neighbourhood of $\theta_0$.

\subsection{Theoretical guarantees}

Theorem~\ref{thm:decomp} below provides an asymptotic decomposition of the pointwise error of our estimator $\hat{f}^{\mathrm{Outrig}}$.  Recall that the conditional variance and conditional Fisher information of the errors are given by $\sigma_P^2(x_0):=\Var_P(\varepsilon\given X=x_0)$ and $i_P(x_0) := \E_P\bigl\{\rho^2(\varepsilon\biggiven X)\given X=x_0\bigr\}$ respectively, and for $\lambda > \lambda_0(K)$ and $x_0 \in \mathcal{X}$, define
\begin{equation}\label{eq:V-lambda}
V_{P}^{(\lambda)}(x_0) := \frac{1}{i_P(x_0)}+\biggl(\sigma_P^2(x_0)-\frac{1}{i_P(x_0)}\biggr)\frac{R_2(\kappa_\lambda)}{R_2(K)}.
\end{equation}
\begin{theorem}\label{thm:decomp}
    Suppose that $\mathcal{P}$ satisfies Assumption~\ref{ass:DGM} and that $\hat{f}^{\mathrm{Outrig}}$ in Algorithm~\ref{alg:outrigger} satisfies Assumption~\ref{ass:bandwidth-kernels-and-co}. Then for each $x_0\in\cX$,
    \begin{equation}\label{eq:decomp}
        \hat{f}^{\mathrm{Outrig}}(x_0)-f(x_0) = B(f,x_0,K,h)h^{\beta^*}+\frac{1}{\sqrt{nh^d}}\biggl\{\frac{R_2(K)V_{P}^{(\lambda)}(x_0)}{p_X(x_0)}\biggr\}^{1/2}Z_n 
         + R_n,
    \end{equation}
    where
    \[
\sup_{P\in\cP}\sup_{x_0\in\cX}\sup_{(h,\lambda)\in\cH_n\times\Lambda_n}\sup_{t\in\R}\big|\Pr_P( Z_{n} \leq t)-\Phi(t)\big| \rightarrow 0 \quad \text{and} \quad R_{n} = \op\bigg(h^{\beta^*}+\frac{1}{\sqrt{nh^d}}\bigg),
\]
    and where the deterministic quantity $B(f,x_0,K,h)$ is defined in Appendix~\ref{appsec:proof-decomp} and satisfies $B(f,x_0,K,h) = O_{\mathcal{P},\mathcal{X},\mathcal{H}}(1)$.
\end{theorem}

We see from Theorem~\ref{thm:decomp} that the error of $\hat{f}^{\mathrm{Outrig}}(x_0)$ decomposes into the sum of a deterministic bias term, a stochastic error term with an asymptotic centred Gaussian distribution and an asymptotically negligible remainder term.  An attractive feature of the decomposition is that neither the bias term $B(f,x_0,K,h)$ nor the variance term $V_{P}^{(\lambda)}(x_0)$ carry any dependence on $\hat\rho$, so all score estimators satisfying~\ref{ass:score-estimation} result in the same asymptotic behaviour for the outrigger estimator.  In fact, the bias term also does not depend on $\rho$, and in particular it is identical to the bias of the standard local polynomial estimator.  The advantage of the outrigger lies in the variance term being necessarily no larger than that of the standard local polynomial estimator; see  Lemma~\ref{lem:score-min} and Figure~\ref{fig:intro}.  
Since $\sigma_P^2(x_0)\geq 1/i_P(x_0)$ by Cauchy--Schwarz (see Lemma~\ref{lem:score-min}), $V_{P}^{(\lambda)}(x_0)$ decreases from $\sigma_P^2(x_0)$ to $1/i_P(x_0)$ as~$\lambda$ increases from $\lambda_0(K)$ to infinity.  If~\eqref{eq:suff-smoothness} holds, then there exists a sequence $(\lambda_n) \in \Lambda$ with $\lambda_n\to\infty$, and hence we can obtain full distributional adaptivity, i.e.~we can replace $V_{P}^{(\lambda)}(x_0)$ in~\eqref{eq:decomp} with
\begin{equation*}
    V_{P}^{(\infty)}(x_0)
    :=
    \lim_{\lambda\to\infty}V_{P}^{(\lambda)}(x_0) 
    = \frac{1}{i_P(x_0)}.
\end{equation*}

The following corollary is an immediate consequence of Theorem~\ref{thm:decomp}.
\begin{corollary}\label{corr:smoothing}
Assume the hypotheses of Theorem~\ref{thm:decomp}.  Consider a sequence of measurable bandwidth functions 
$h_n:\cX\to (0,\infty)$, set $\tau_n:=n^{1/(2\beta^*\hspace{-0.1em}+d)}h_n$, and define a sequence of outrigger parameters $\lambda_n\to\lambda^*\in \bigl(\lambda_0(K),\infty\bigr]$. Then
\begin{equation*}
    \sup_{P\in\cP}\sup_{x_0\in\cX}
    \sup_{t\in\R}\biggl|
    \Pr_P\biggl(\biggl\{\frac{p_X(x_0)\tau_n^d(x_0)}{R_2(K)V_{P}^{(\hspace{-0.01em}\lambda^*\hspace{-0.05em})}\hspace{-0.08em}(x_0)}\biggr\}^{\hspace{-0.1em}\frac{1}{2}}\hspace{-0.1em}\biggl\{n^{\frac{\beta^*}{2\beta^*\hspace{-0.1em}+d}}\bigl(\hat{f}^{\mathrm{Outrig}}(x_0)-f(x_0)\bigr)-B_f(x_0)\tau_n^{\beta^*}\hspace{-0.2em}(x_0)\biggr\}\leq t\biggr)-\Phi(t)
    \biggr|
     \rightarrow 0,
\end{equation*}
where we use the shorthand $B_f(x_0):=B\bigl(f,x_0,K,h_n(x_0)\bigr)$. 
\end{corollary}
Corollary~\ref{corr:smoothing} can be employed for confidence interval construction by taking an undersmoothing bandwidth satisfying $\sup_{x\in\cX}\tau_n(x)\to0$, as for standard local polynomial estimators~\citep{hall, wasserman}, but with our outrigger intervals being asymptotically tighter, except in the case of Gaussian errors, where they are asymptotically equivalent. 

We now proceed to our main results, which describe the senses in which our outrigger local polynomial estimator outperforms the standard local polynomial estimator.  This requires some preliminary notation.  For $\epsilon>0$ and $f \in \cH(\beta,L)$, define the \emph{local H\"older class} 
$\cH_{f,\epsilon}^{\mathrm{loc}}(\beta,L):=\{\tilde{f}\in\cH(\beta,L):\|\tilde{f}-f\|_\infty\leq\epsilon\}$. 
Given a polynomial degree $p\in\N_0$, define the convergence rate $s_{n,h}:=h^{2\beta^*}+\frac{1}{nh^d}$.  We measure the performance of an estimator $\bar{f}$ via its \emph{local worst-case risk}
\[
\mathcal{R}_{n,h,P,x_0,\epsilon,M}(\bar{f}) := 
\sup_{\tilde{f}\in\cH_{f,\epsilon}^{\mathrm{loc}}(\beta,L)}
\E_{(P_X,P_{\varepsilon|X},\tilde{f})}\bigl\{s_{n,h}^{-1}\bigl(\bar{f}(x_0)-\tilde{f}(x_0)\bigr)^2\wedge M\bigr\},
\]
where $M>0$ is a truncation level.  The supremum here is taken over a local function class to avoid a  Hodges-type phenomenon where pointwise asymptotics may disguise true finite-sample performance \citep[e.g.][Section~1.3.2]{samworth-shah}. The truncation at level $M$ may be replaced with the exact mean squared error (i.e.~we may take $M=\infty$) under a number of possible mild modifications of the estimator \citep[e.g.][]{fan-minimax, incontext}.  

We will compare the performance of our outrigger estimator with the standard local polynomial estimator by studying the ratio of the local worst-case risks.  To be more precise, and in order to make the comparison as favourable as possible for the standard estimator, our primary quantity of interest will be a supremum of this ratio over choices of bandwidths and outrigger parameters, points $x_0 \in \mathcal{X}$ at which we seek to estimate the regression function, and all data generating distributions $P = (P_X,P_{\varepsilon|X},f) \in \mathcal{P}$.  Indeed, we define the \emph{asymptotic least-favourable risk ratio} to be
\[
r_{\mathcal{H}} := \sup_{\epsilon>0}\limsup_{M\to\infty}\limsup_{n\to\infty}\sup_{P\in\cP}\sup_{x_0\in\cX}\sup_{(h,\lambda)\in\cH_n\times\Lambda_n}\frac{\mathcal{R}_{n,h,P,x_0,\epsilon,M}(\hat{f}^{\mathrm{Outrig}})}{\mathcal{R}_{n,h,P,x_0,\epsilon,M}(\hat{f}^{\mathrm{LP}})}.
\]
\begin{theorem}\label{thm:do-no-worse}
Consider the standard local polynomial estimator $\hat{f}^{\mathrm{LP}}$ and outrigger estimator $\hat{f}^{\mathrm{Outrig}}$ in Algorithm~\ref{alg:outrigger}, each with the same degree $p\in\mathbb{N}_0$ and kernel $K$. Suppose that Assumptions~\ref{ass:DGM} and~\ref{ass:bandwidth-kernels-and-co} hold.  Then 
\begin{equation*}
r_{\cH}
    \leq 1.
\end{equation*}
\end{theorem}
Theorem~\ref{thm:do-no-worse} demonstrates a very strong sense in which the outrigger estimator is no worse than the standard estimator: despite taking the worst case for the outrigger estimator over all choices of bandwidth~$h$ and outrigger parameter $\lambda$ in $\mathcal{H}_n \times \Lambda_n$, $x_0 \in \mathcal{X}$ and $P \in \mathcal{P}$ in the definition of the asymptotic least-favourable risk ratio, this quantity is at most 1.

Now let $\cH_n'\subseteq \cH_n$ be such that
\[
\sup_{(h,\lambda)\in\cH_n'\times\Lambda_n}n^{1/(2\beta^*\hspace{-0.1em}+d)}h<\infty.
\]
This restriction constrains our bandwidths to be such that our estimator is asymptotically stochastic, i.e.~the deterministic bias term does not asymptotically strictly dominate in the decomposition~\eqref{eq:decomp}. 
To argue that the outrigger estimator represents a strict improvement on standard local polynomial estimation, we consider two versions of the \emph{asymptotic least-favourable risk ratio at $x_0 \in \mathcal{X}$ for a given error distribution $P_{\varepsilon|x_0}$}: first, let 
 \begin{equation*}
    {r}_{\cH'}(P_{\varepsilon|x_0}) := \sup_{\epsilon>0}\limsup_{M\to\infty}\limsup_{n\to\infty}\sup_{P_X\in\cP_X}\sup_{f\in\cH(\beta,L)}\sup_{(h,\lambda)\in \cH_n' \times \Lambda_n}\frac{\mathcal{R}_{n,h,P,x_0,\epsilon,M}(\hat{f}^{\mathrm{Outrig}})}{\mathcal{R}_{n,h,P,x_0,\epsilon,M}(\hat{f}^{\mathrm{LP}})},
\end{equation*}
and second, for a sequence $(\lambda_n) \in \Lambda$ with  $\underbar{\lambda} := \liminf_{n \rightarrow \infty} \lambda_n$, let
\begin{equation*}
    {r}_{\cH'}(P_{\varepsilon|x_0},\underbar{\lambda}) := \sup_{\epsilon>0}\limsup_{M\to\infty}\limsup_{n\to\infty}\sup_{P_X\in\cP_X}\sup_{f\in\cH(\beta,L)}\sup_{h\in \cH_n'}\frac{\mathcal{R}_{n,h,P,x_0,\epsilon,M}(\hat{f}^{\mathrm{Outrig}})}{\mathcal{R}_{n,h,P,x_0,\epsilon,M}(\hat{f}^{\mathrm{LP}})},
\end{equation*}

\begin{theorem}
\label{thm:improvements}
Consider the standard local polynomial estimator $\hat{f}^{\mathrm{LP}}$ and outrigger estimator  $\hat{f}^{\mathrm{Outrig}}$ in Algorithm~\ref{alg:outrigger}, each with the same degree $p\in\mathbb{N}_0$ and kernel $K$. Suppose that Assumptions~\ref{ass:DGM}~and~\ref{ass:bandwidth-kernels-and-co} hold and let $x_0 \in \mathcal{X}$. Then 
\begin{enumerate}[label=(\roman*)]
    \item ${r}_{\cH'}(P_{\varepsilon|x_0}) \leq 1$ with equality if and only if $P_{\varepsilon|x_0}$ is Gaussian;
    \label{thm:improved-ratio-2}
    \item For arbitrarily small $\varsigma\in(0,1)$, there exist a distribution $P^*_{\varepsilon|x_0}$ and $\lambda^* > 0$ such that ${r}_{\cH'}(P^*_{\varepsilon|x_0},\underbar{\lambda})\leq \varsigma$ for all $\underbar{\lambda}\geq\lambda^*$.
    \label{thm:improved-ratio-3}
\end{enumerate}

\end{theorem}

It is immediate from the definitions that 
\[
{r}_{\cH'}(P_{\varepsilon|x_0},\underbar{\lambda}) \leq {r}_{\cH'}(P_{\varepsilon|x_0}) \leq {r}_{\cH}
\]
for all $\underbar{\lambda} > 0$ and error distributions $P_{\varepsilon|x_0}$.  Theorem~\ref{thm:improvements}\ref{thm:improved-ratio-2} reveals that the first form of asymptotic least-favourable risk ratio at $x_0$ for error distribution $P_{\varepsilon|x_0}$ is strictly less than 1 for every non-Gaussian error distribution $P_{\varepsilon|x_0}$.  In combination with Theorem~\ref{thm:do-no-worse}, this provides a sense in which the outrigger estimator strictly dominates the standard local polynomial estimator.  In fact, from Theorem~\ref{thm:improvements}\ref{thm:improved-ratio-3}, we see that the second form of least-favourable risk ratio at $x_0$ can be arbitrarily small.

\subsection{Minimax optimality with constants}\label{sec:minimax}

Global minimax optimality is often considered as a gold standard for statistical procedures~\citep[e.g.][]{tsybakov,samworth-shah}. 
However  
this notion of optimality in terms of a `best-case estimator' over a `worst-case distribution' ignores the possibility that there may exist estimators that perform equally well in the worst case but where one adapts to provide improvements over another in more favourable settings. In our case, where we are interested in adaptivity to the conditional distribution of $\varepsilon$ given $X$, the `worst-case distribution' is (at least almost) that of conditionally Gaussian errors, and as such minimising a least squares data fidelity term as in standard local polynomial estimation naturally yields an estimator that attains the global minimax rate.  However, Theorem~\ref{thm:improvements} demonstrates the potential for arbitrarily large improvements outside this worst-case Gaussian errors scenario.  In this subsection, therefore, we study the instance-optimality of the outrigger estimator for different error distributions.  Specifically, for a truncation level $M \in (0,\infty]$, as well as $x_0 \in \mathcal{X}$, $P_X \in \mathcal{P}_X$ and error distribution $P_{\varepsilon|x_0}$, we define the \emph{worst-case normalised mean squared error} of an estimator $\hat{f}_n$ at $x_0$ by
\[
\mathrm{MSE}_{n,M}(\hat{f}_n):=
\sup_{f\in\cH(\beta,L)}
\E_P\Bigl\{n^{2\beta/(2\beta+d)}\bigl(\hat{f}_n(x_0)-f(x_0)\bigr)^2\,\wedge M\Bigr\},
\]
as well as $\mathrm{MSE}_n(\hat{f}_n) :=\mathrm{MSE}_{n,\infty}(\hat{f}_n)$. 

\begin{theorem}\label{thm:UB}
    Under Assumptions~\ref{ass:DGM},~\ref{ass:kappa} and~\ref{ass:cross-fitting}, there exist a sequence of bandwidths and kernel~$K$ such that for any sequence $(\lambda_n)$ of outrigger parameters with $\lambda_n \to \infty$ and $\lambda_n^{2\beta+d}/n \to0$, and polynomial degree $p \geq \lceil \beta \rceil - 1$, the outrigger estimator from Algorithm~\ref{alg:outrigger}  satisfies
    \begin{equation}\label{eq:LAM-UB}
        \limsup_{M\to\infty}\limsup_{n\to\infty}\mathrm{MSE}_{n,M}\bigl(\hat{f}^{\mathrm{Outrig}}\bigr)
        \leq
        C_{\beta,d}\,\bigg(\frac{L^{d/\beta}}{p_X(x_0) i_P(x_0)}\bigg)^{2\beta/(2\beta+d)},
    \end{equation}
    where $C_{\beta,d} > 0$ depends only on $(\beta,d)$.   
    In the case $\beta\in(0,2]$, we can take
    \begin{equation}
    \label{Eq:Cbetad}
        C_{\beta,1} = 
        \biggl(\frac{(\beta+1)^{2\beta}}{(2\beta)^{2\beta}(2\beta+1)(1\vee\beta)^2}\biggr)^{1/(2\beta+1)}.
    \end{equation}
\end{theorem}
The upper bound in Theorem~\ref{thm:UB} is complemented by the lower bound in Theorem~\ref{thm:LAM-LB} below. 

\begin{theorem}
\label{thm:LAM-LB}
    Under Assumption~\ref{ass:DGM}, for any sequence $(\hat{f}_n)_{n\in\mathbb{N}}$ of Borel measurable estimators,
    \begin{equation*}
\liminf_{n\to\infty} 
\mathrm{MSE}_{n}(\hat{f}_n) \geq c_{\beta,d} \bigg(\frac{L^{d/\beta}}{p_X(x_0) i_P(x_0)}\bigg)^{2\beta/(2\beta+d)}.
    \end{equation*}
    where $c_{\beta,d} > 0$ depends only on $(\beta,d)$.  In fact,
\begin{equation}
\label{Eq:Cratio}
    1 = \lim_{\beta \searrow 0} \sup_{d \in \mathbb{N}} \frac{C_{\beta,d}}{c_{\beta,d}} \leq \sup_{\beta\in(0,1]} \sup_{d\in\N}
    \frac{C_{\beta,d}}{c_{\beta,d}} \leq 1.69.
\end{equation}
\end{theorem}

Since the bounds in Theorems~\ref{thm:UB} and~\ref{thm:LAM-LB} hold for every error distribution $P_{\varepsilon|x_0}$ satisfying the relevant parts of Assumption~\ref{ass:DGM} (and every $P_X \in \mathcal{P}_X$),~\eqref{Eq:Cratio} indicates that, even at the level of constants, the outrigger estimator is almost instance optimal across error distributions, and is asymptotically optimal even up to constants in the low smoothness limit.

\section{Numerical experiments}\label{sec:sims}

\subsection{Implementation details}\label{sec:numerical-implementation}

In all of our numerical experiments, we implement the outrigger local polynomial estimator (and the standard local polynomial estimator) as follows.  We take the local constant $(p=0)$ version of the estimator with the Epanechnikov kernel $K(\nu)=\tfrac{3}{4}\max(1-\nu^2,0)$ as the primary kernel, and the uniform kernel $\kappa_\lambda(\nu)\propto\ind_{\mathcal{B}_{x_0}(\lambda)\setminus\mathcal{B}_{x_0}(1)}(\nu)$ as the outrigger kernel. 

Algorithm~\ref{alg:outrigger} depends on an estimator of the conditional score function.  Our practical choice is motivated by the observation that for $x\in\mathcal{X}^\circ$, the conditional score function $\rho(\cdot\given x)$ minimises the (conditional) score matching objective~\citep{dennis-cox, scorematching}
\begin{equation*}
    \mathcal{L}(\varrho) := \E\bigl\{\varrho^2(\varepsilon\given X)+2\varrho'(\varepsilon\given X)\bigr\}
    ,
\end{equation*}
over measurable functions $\varrho:\R\times\mathcal{X}^\circ\to\R$ satisfying $\max\bigl(\E\bigl\{\varrho(\varepsilon\given X)^2\given X=x\bigr\},\,\E\bigl\{\bigl|\varrho'(\varepsilon\given X)\bigr|\given X=x\bigr\}\bigr)<\infty$. 
For $x_0\in\cX$ and $t>0$, define the  local region $\mathrm{loc}_t(x_0):=\{x\in\cX^\circ:\|x-x_0\|\leq t\}$. The  score matching spline $\hat{\rho}_{\mathrm{SM}}(\cdot\given x)$ for $x\in\mathrm{loc}_t(x_0)$ then minimises the penalised empirical score matching loss
\begin{equation*}
    \hat{\mathcal{L}}_n(\varrho) :=
    \frac{1}{|\{i:X_i\in\mathrm{loc}_t(x_0)\}|}\sum_{i:X_i\in\mathrm{loc}_t(x_0)} \bigl\{\varrho^2(\hat\varepsilon_i\given x) 
    +
    2\varrho'(\hat\varepsilon_i\given x)\bigr\} 
    +
    \eta \int_\R\varrho''(\varepsilon\given x)^2 w_{x_0}(\varepsilon)\,d\varepsilon,
\end{equation*}
where $w_{x_0}:\mathbb{R} \rightarrow [0,\infty)$ is measurable with $\int_\R w_{x_0}<\infty$, and where $\hat{\varepsilon}_1,\ldots,\hat{\varepsilon}_n$ are the residuals from the outrigger estimator in Algorithm~\ref{alg:outrigger}. 
The regularisation parameter $\eta>0$ is chosen by 10-fold cross-validation with the score matching loss. In simulations, and as recommended by~\citet{dennis-cox} and \citet{ng}, we take $w_{x_0}(\cdot):=\ind_{[a_0,b_0]}(\cdot)$ where $a_0:=\min_{i:X_i\in\mathrm{loc}_t(x_0)}\hat\varepsilon_i$ and $b_0:=\max_{i:X_i\in\mathrm{loc}_t(x_0)}\hat\varepsilon_i$. 
Consistency results for the resulting (conditional) score estimator and its derivative follow from~\citet[Corollary 7]{dennis-cox}.

\subsection{Uniform improvements over bandwidths}\label{sec:sim-indeperrors}

We take the data generating mechanism~\eqref{eq:y=f(x)+e} with $n=2000$, $X\sim\text{Unif}[-2,2]$, $f(x)=4\cos(\pi x)$, and $\varepsilon\given X\sim P_\varepsilon$ with five choices of error distribution:
\begin{enumerate}[label=(\roman*)]
    \item Standard Gaussian: $P_\varepsilon=N(0,1)$;
    \item 
    Gaussian scale mixture: $P_\varepsilon=\tfrac{1}{2}N(0,1)+\tfrac{1}{2}N(0,16)$;
    \label{item:scale-mix}
    \item Gaussian location mixture: $P_\varepsilon=\tfrac{1}{2}N\bigl(1,\tfrac{1}{10}\bigr)+\tfrac{1}{2}N\bigl(-1,\tfrac{1}{10}\bigr)$;
    \label{item:loc-mix}
    \item Smoothed exponential: $P_\varepsilon$ is the distribution of $W-1+\tfrac{\sqrt{3}}{10}Z$ where $W\sim\text{Exp}(1)$ and $Z\sim N(0,1)$ are independent;
    \item Cubed Gaussian: $P_\varepsilon$ is the distribution of $Z^3$ where $Z\sim N(0,1)$. 
\end{enumerate}
In the penultimate case, the oracle score function has a large Lipschitz constant of $100/3$, while in the final one it has a discontinuity at the origin, and moreover has infinite Fisher information, so does not satisfy Assumption~\ref{ass:DGM}.  These cases, particularly~(v), should therefore be considered as challenging examples.  We compare the outrigger estimator (Algorithm~\ref{alg:outrigger}) with $\lambda=8$, to the standard 
Nadaraya--Watson ($p=0$) estimator~\eqref{eq:lp}, and the oracle local likelihood estimator that is given access to the true score function~\eqref{eq:lp-ora-conditional}.

In Figure~\ref{fig:indep-errors} we plot the mean squared errors~$\E\bigl\{\bigl(\hat{f}(0)-f(0)\bigr)^2\bigr\}$ of our estimators of $f(0)$, as a function of the bandwidth $h$, averaged over 1000 repetitions.  Table~\ref{tab:indep-errors} presents numerical values for these mean squared errors for the empirically optimal choice of bandwidth.  For all non-Gaussian error distributions, and for all bandwidths, our outrigger local polynomial estimator exhibits improvements over the local polynomial estimator, validating our theoretical findings in Theorem~\ref{thm:improvements}\ref{thm:improved-ratio-2}.  Moreover, in the Gaussian case we see indistinguishable performance compared with the standard local polynomial estimator, again in line with Theorem~\ref{thm:improvements}\ref{thm:improved-ratio-2}.  Even in settings~(iv) and~(v), where the performance of the outrigger and oracle estimators is not identical, the outrigger estimator still comfortably outperforms the standard local polynomial estimator, and in fact improves on the oracle estimator for some choices of bandwidth. 

\begin{figure}[H]
    \centering
    \includegraphics[width=0.99\linewidth]{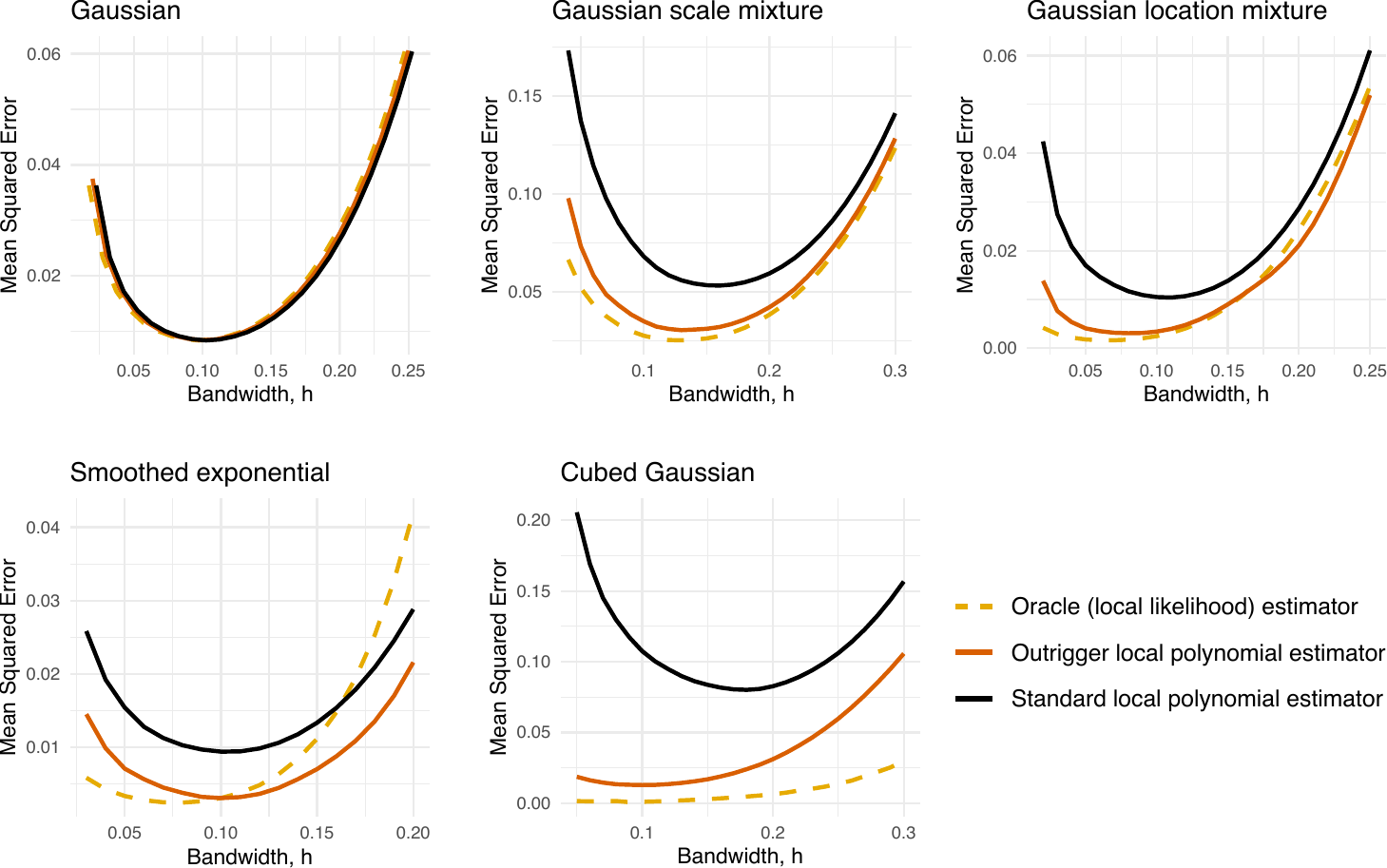}
    \caption{Pointwise mean squared error $\E\bigl\{\bigl(\hat{f}(0)-f(0)\bigr)^2\bigr\}$ in the numerical experiments of Section~\ref{sec:sim-indeperrors} over different bandwidths $h$.}
    \label{fig:indep-errors}
\end{figure}

\begin{table}[ht]
    \centering
    \begin{tabular}{|c|c|c|c|c|}
    \hline
        & Oracle & {\begin{tabular}{@{}c@{}c@{}}Outrigger\\estimator\end{tabular}} & {\begin{tabular}{@{}c@{}c@{}}Standard\\local polynomial\end{tabular}}
        & {\begin{tabular}{@{}c@{}c@{}}Ratio of MSE of outrigger\\to local polynomial\end{tabular}}
        \\
    \hline
        Standard Gaussian 
        & 8.41 & 8.45 & {\bf 8.41} & 100\%
        \\
        Gaussian scale mixture 
        & 25.3 & {\bf 30.5} & 53.3 & 57.3\%
        \\
        Gaussian location mixture
        & 1.64 & {\bf 3.07} & 10.39 & 29.6\%
        \\ 
        Smoothed exponential
        & 2.42 & {\bf 3.04} & 9.38 & 32.4\%
        \\
        Cubed Gaussian 
        & 0.69 & {\bf 17.31} & 80.13 & 21.6\%
        \\
        \hline
    \end{tabular}
    \caption{Pointwise mean squared error $\E\bigl\{\bigl(\hat{f}(0)-f(0)\bigr)^2\bigr\}$ ($\times 10^3$) in the numerical experiments of Section~\ref{sec:sim-indeperrors} for the optimal bandwidth.  The last column presents the ratio of the mean squared errors in the previous two columns.}
    \label{tab:indep-errors}
\end{table}

Figure~\ref{fig:lambda-dep} compares the  theoretical and empirical ratios of the mean squared errors for the two error distributions~\ref{item:scale-mix} and~\ref{item:loc-mix} as the outrigger parameter $\lambda$ is varied. In addition to providing a visualisation of the rate of improvement in these ratios as $\lambda$ increases, we note that in both cases the empirical and theoretical curves match well, providing reassurance that the asymptotic theory of Section~\ref{sec:theory} is reflected in empirical performance. 


\begin{figure}[h]
    \centering
    \includegraphics[width=0.9\linewidth]{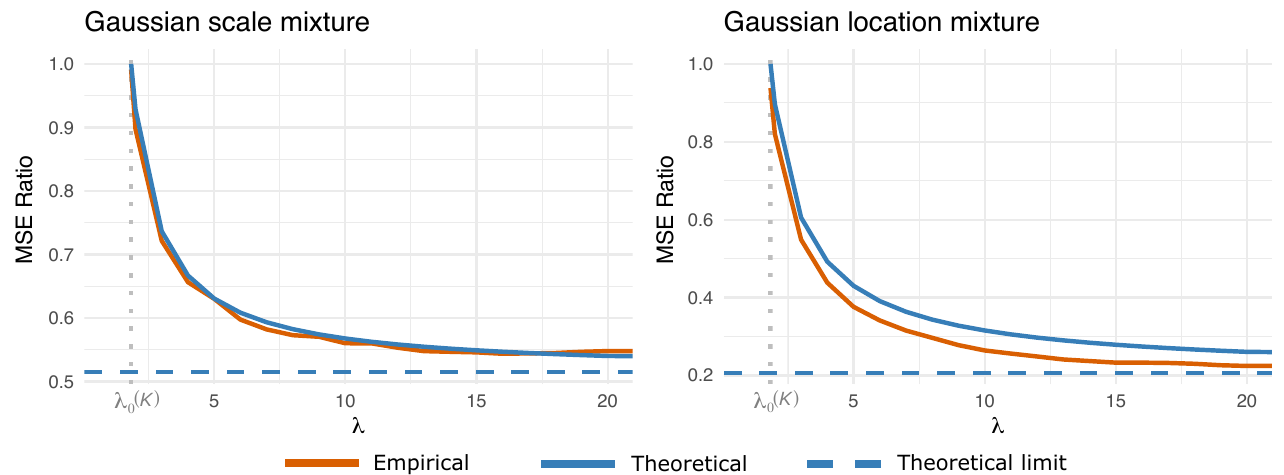}
    \caption{Empirical and theoretical comparison of the ratio of the mean squared error (MSE) of the outrigger estimator with parameter $\lambda\in[\lambda_0(K),20]$ and the standard local polynomial estimator for the experiments of Sections~\ref{sec:sim-indeperrors}\ref{item:scale-mix} and~\ref{sec:sim-indeperrors}\ref{item:loc-mix}. The empirical curve plots $\lambda\mapsto\frac{\mathrm{MSE} \, \hat{f}^{\mathrm{Outrig}}(x_0)}{\mathrm{MSE} \, \hat{f}^{\mathrm{LP}}(x_0)}$, with MSEs estimated over 1000 repetitions. 
    The theoretical MSE ratio anticipated by the bias and variance terms in Theorem~\ref{thm:decomp} is $\lambda\mapsto\bigl({V_{P}^{(\lambda)}(x_0)}/{\sigma_P^2(x_0)}\bigr)^{2/3}$, with theoretical limit $\lim_{\lambda\to\infty}\bigl({V_{P}^{(\lambda)}(x_0)}/{\sigma_P^2(x_0)}\bigr)^{2/3}=\Bigl(\frac{1/i_P(x_0)}{\sigma_P^2(x_0)}\Bigr)^{2/3}$.
    }
    \label{fig:lambda-dep}
\end{figure}

\subsection{Errors not independent of covariates}\label{sec:numerics-nonindep-errors}

In this subsection we study settings where the errors and covariates are dependent.  We consider $n=10^4$ independent copies of data $(X,Y)$ following
\begin{gather*}
    X \sim N(0,1),
    \qquad
    Y = 4\cos(\pi X)+\varepsilon,
\end{gather*}
where $\varepsilon\given X\sim P_{\varepsilon|X}$, with the following choices of error distribution:
\begin{enumerate}[label=(\roman*)]
    \item Gaussian scale mixture: $P_{\varepsilon|X}=\tfrac{\exp(X)}{1+\exp(X)}N(0,1)+\tfrac{1}{1+\exp(X)}N(0,16)$;
    \item Exponential-$t_3$ convolution: $P_{\varepsilon|X}$ is the distribution of $\tfrac{4}{4+X^2}(W-1)+\tfrac{1+X^2}{4+X^2}T$ where $W\sim\mathrm{Exp}(1)$ and $T\sim t_3$, with $W$, $T$ and $X$ independent;
    \label{item:exp-t3}
    \item Power of a Gaussian: $P_{\varepsilon|X}$ is the distribution of $|Z|^{\tfrac{1}{2}+X^2}\sgn(Z)$ where $Z\sim N(0,1)$ and $X$ are independent.
\end{enumerate}
Surface plots of the corresponding conditional score functions are given in Figure~\ref{fig:cond-scores}. The first two conditional scores are relatively well-behaved, while the third is much less so.  Figure~\ref{fig:errors-not-indep-cov} plots mean squared errors as the  outrigger parameter $\lambda\in[2,10]$ varies.  We take the same fixed bandwidth across all estimators of $h=0.11,0.07,0.07$ for the three error distributions respectively (which are the optimal bandwidths for the standard local polynomial estimator in each case) 
and $\mathcal{K}=2$ folds for cross-fitting.  While the fixed choice of bandwidth is not necessarily optimal for the outrigger estimator, we are still guaranteed asymptotic improvements over the local polynomial estimator; the same conclusion holds when taking any deterministic `rule of thumb' bandwidth and any choice of outrigger parameter $\lambda$. As indicated by~\eqref{eq:V-lambda}, the asymptotic mean squared error of the outrigger is comparable to that of the local polynomial estimator when $\lambda=\lambda_0(K)= 11/6$, and decreases in $\lambda$, at least up to a point where it is no longer reasonable to think of $\lambda h\to0$.  The Gaussian scale mixture is impervious to choosing $\lambda$ too large, while for the latter two cases we see the apparent necessity of the condition $\lambda h\to0$ in our asymptotic regime in Assumption~\ref{ass:h}.

\begin{figure}[h]
    \centering
    \includegraphics[width=0.99\linewidth]{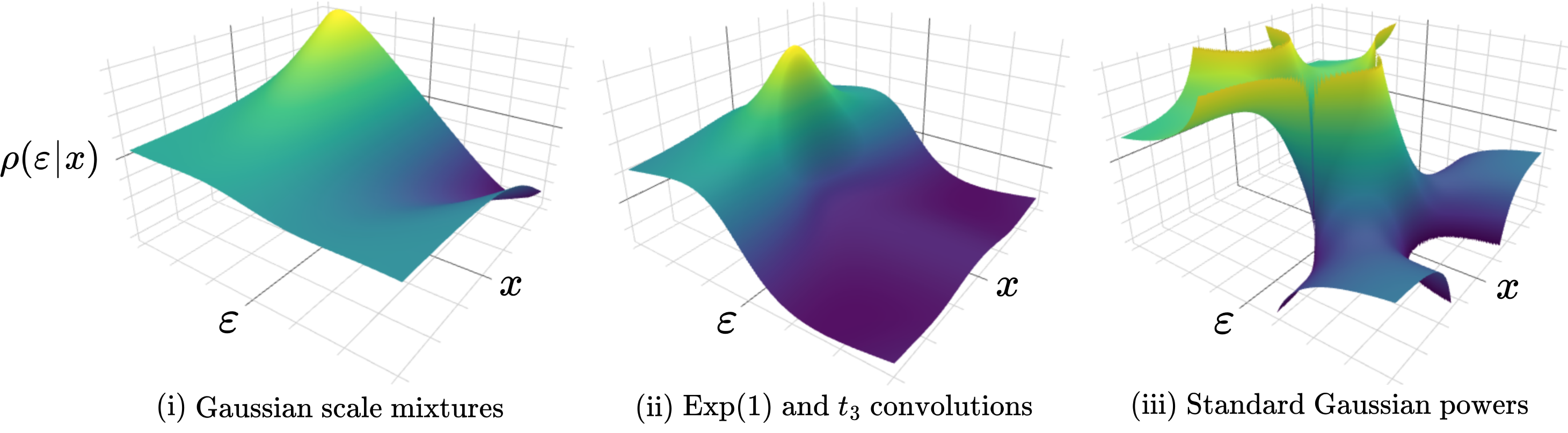}
    \caption{Surface plots of the conditional score functions $\rho(\varepsilon\given x)$ for the three data generating mechanisms in Section~\ref{sec:numerics-nonindep-errors}.}
    \label{fig:cond-scores}
\end{figure}

\begin{figure}[h]
    \centering
    \includegraphics[width=0.99\linewidth]{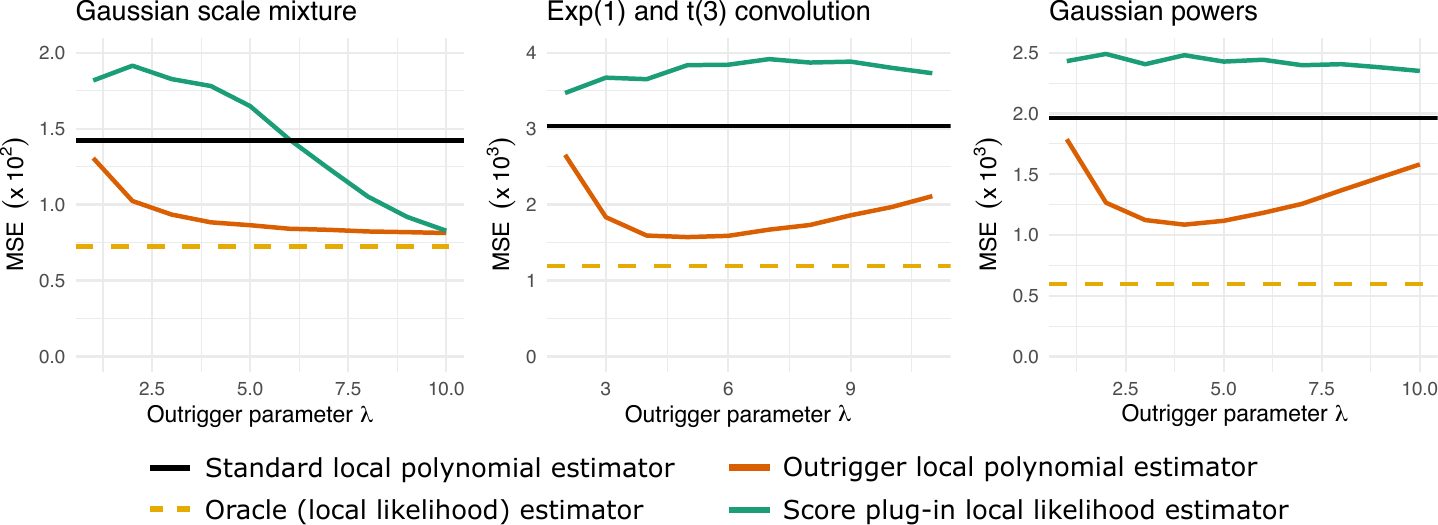}
    \caption{Mean squared errors $\E\bigl\{\bigl(\hat{f}(0)-f(0)\bigr)^2\bigr\}$ for fixed bandwidths $h=0.11,0.07,0.07$ over the three simulations respectively (optimal in the standard local polynomial case) and over different outrigger parameters $\lambda$ (1000 simulations).}
    \label{fig:errors-not-indep-cov}
\end{figure}

\subsection{Real data study}\label{sec:real-data} 

We study a Spotify tracks dataset, which can be found on Kaggle at~\url{www.kaggle.com/datasets/maharshipandya/-spotify-tracks-dataset}. Our interest is in studying the association between a popularity metric (calculated based on the number of recent streams, on a scale from $0$ to $100$), and a positivity metric (calculated based on the sentiment of the track, measured on a scale from $0$ to $1$). More details on these metrics are given in the aforementioned link.
We restrict attention to the $91{,}271$ tracks with popularity score exceeding 10.   
We fit the outrigger estimator with $\lambda=8$ and the standard local polynomial estimator on $100$ randomly selected subsamples of size $10{,}000$; see Figure~\ref{fig:spotify}. 
The same bandwidth was employed for both estimators, chosen by squared error cross-validation when fitting the standard local polynomial estimator.  
Conditional score estimation for the outrigger estimator was carried out using the score matching spline methodology of Section~\ref{sec:numerical-implementation} with $t$ equal to the bandwidth. 
The strong similarity of the empirical mean functions in the left and middle panels of Figure~\ref{fig:spotify} is an illustration of the similar bias of the outrigger and standard local polynomial estimators, but the outrigger estimator has visibly smaller empirical variance across these subsamples.  Further detail on this empirical variance reduction is given in the right panel of Figure~\ref{fig:spotify}; the average MSE ratio, over a uniform distribution with respect to positivity scores, of the outrigger to standard local polynomial is $0.53$. 
Figure~\ref{fig:spotify-score} presents an estimator of the conditional score function at different positivity levels using the score matching splines of Section~\ref{sec:numerical-implementation}.  We see that the conditional error distribution is not symmetric, and nor are the errors independent of the covariates. 

\begin{figure}[h]
    \centering
    \includegraphics[width=1.0\linewidth]{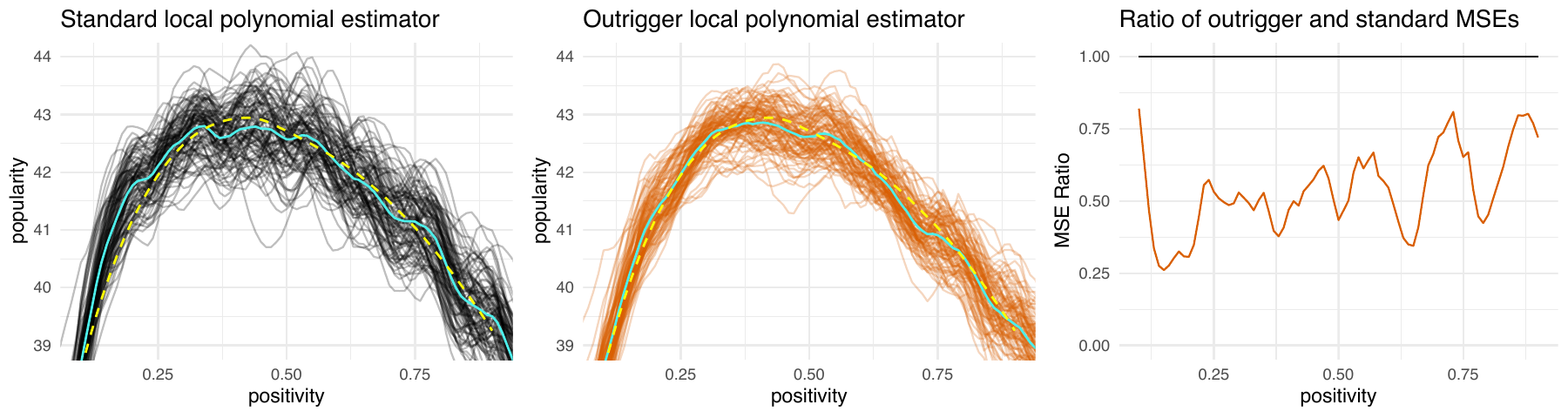}
    \caption{Estimates of the regression function on $100$ subsampled datasets of size $n=10^4$ for the local constant (left) and outrigger local constant estimator (middle).  The empirical mean of the estimators are plotted in cyan, and the standard local quadratic estimator fitted on the entire dataset is plotted in dashed yellow. The right panel presents the estimated MSE with respect to the `semi-oracle' standard local quadratic estimator that uses the entire dataset.}
    \label{fig:spotify}
\end{figure}

\begin{figure}[h]
    \centering
    \includegraphics[width=0.6\linewidth]{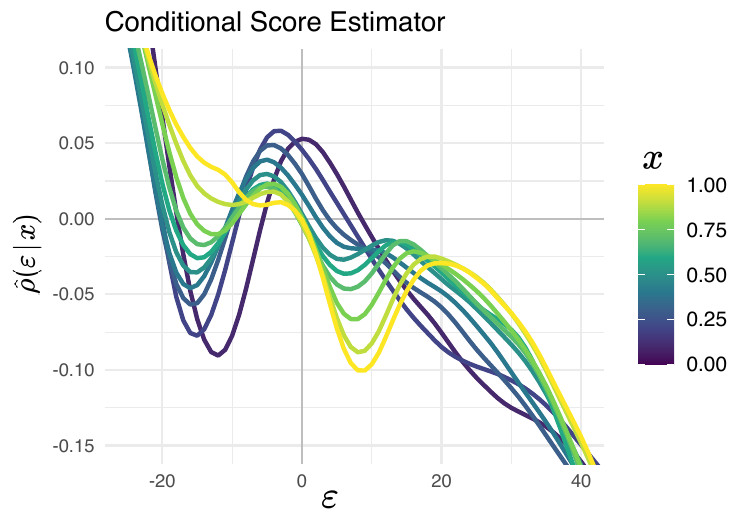}
    \caption{Estimated conditional score functions $\hat{\rho}(\varepsilon\given x)$ of the errors $\varepsilon:=Y-\E(Y\given X=x)$, where~$Y$ represents the Spotify popularity score of a track, $X$ represents its positivity rating and $x\in\{r/10:r\in[10]\}$.}
    \label{fig:spotify-score}
\end{figure}

\section{Extensions}\label{sec:extensions}

Here we discuss two possible extensions of our ideas.

\medskip\noindent{\bf Beyond conditional score estimation: } 
Suppose that the conditional score function $\rho$ does not satisfy Assumption~\ref{ass:DGM}, e.g.~because the conditional density of the errors is not sufficiently smooth.  Theorem~\ref{thm:decomp} extends to allow us to replace $\rho$ and $\hat{\rho}$ with an alternative function $\varrho$ satisfying~\ref{ass:score} and estimator $\hat{\varrho}$ satisfying~\ref{ass:score-estimation}, provided that $\E_P\{\varrho(\varepsilon\given X)\given X\}=0$ and $\bigl|\E_P\bigl(\varrho'(\varepsilon\given X)\biggiven X=x\bigr)\bigr|\geq c_1$ for $P_X$-almost all $x\in\cX^\circ$ and all $P\in\cP$; under~\ref{ass:score} these latter conditions are automatically satisfied by the conditional score, i.e.~when $\varrho = \rho$.  Moreover, we can even relax the requirement that the conditional expectation of $\varrho$ is exactly zero to requiring only that 
\begin{equation}\label{eq:approx-soln}
\sup_{h\in\mathcal{H}_n}\sup_{P\in\cP}\sup_{x\in\cX^\circ}\frac{\bigl|\E_P\bigl\{\varrho(\varepsilon\given X)\biggiven X=x\bigr\}\bigr|}{h^{\beta^*}}\to0.
\end{equation}
\begin{theorem}\label{thm:decomp-ext}
    Suppose that $\mathcal{P}$ satisfies Assumption~\ref{ass:DGM} with $\varrho$ in place of~$\rho$ and $\hat{\varrho}$ in place of $\hat{\rho}$.  Suppose further that $\varrho$ satisfies~\eqref{eq:approx-soln} and that  $\bigl|\E_P\bigl(\varrho'(\varepsilon\given X)\biggiven X=x\bigr)\bigr|\geq c_1$ for $P_X$-almost all $x\in\cX^\circ$ and all $P\in\cP$. Finally suppose that the estimator $\hat{f}^{\mathrm{Outrig}}$ of Algorithm~\ref{alg:outrigger} satisfies Assumption~\ref{ass:bandwidth-kernels-and-co}. Then the decomposition~\eqref{eq:decomp} of Theorem~\ref{thm:decomp} holds for every $x_0 \in \mathcal{X}$ with $V_{P}^{(\lambda)}(x_0)$ replaced with
    \begin{equation*}
        V_P^{(\lambda)}(x_0)[\varrho]
        :=
        \frac{\E_P\bigl\{\varrho^2(\varepsilon\given X)\given X=x_0\bigr\}}{\bigl\{\E_P\bigl(\varrho'(\varepsilon\given X)\given X=x_0\bigr)\bigr\}^2}+\biggl(\sigma_P^2(x_0)-\frac{\E_P\bigl\{\varrho^2(\varepsilon\given X)\given X=x_0\bigr\}}{\bigl\{\E_P\bigl(\varrho'(\varepsilon\given X)\given X=x_0\bigr)\bigr\}^2}\biggr)\frac{R_2(\kappa_\lambda)}{R_2(K)}.
    \end{equation*}
\end{theorem}
Theorem~\ref{thm:decomp} is recovered by recognising that when $\varrho = \rho$,
\begin{equation*}
    \frac{\E_P\bigl\{\varrho^2(\varepsilon\given X)\given X=x_0\bigr\}}{\bigl\{\E_P\bigl(\varrho'(\varepsilon\given X)\given X=x_0\bigr)\bigr\}^2}=\frac{1}{i_P(x_0)};
\end{equation*}
see Lemma~\ref{lem:score-min}.  

As a first example of such an alternative function $\varrho$, suppose we pre-specify integrable functions $M_1,\ldots,M_J:\R\times\mathcal{X}^\circ\to\R$ with $M_1(\varepsilon,x)=\varepsilon$ and moments
\begin{equation*}
    m_j(x):=\E_P\bigl(M_j(\varepsilon,X)\given X=x\bigr),\qquad j\in[J]
\end{equation*} 
to be estimated, except that $m_1(x)=0$.  For arbitrary measurable functions $a_1,\ldots,a_J:\mathcal{X}^\circ \rightarrow \mathbb{R}$, we can introduce the \emph{pseudo-score}
\begin{equation}
    \varrho(\varepsilon\given x)
    :=
    \sum_{j=1}^J a_j(x)\bigl\{M_j(\varepsilon,x)-m_j(x)\bigr\}.
\end{equation}
A natural estimator would then be
\begin{equation}
    \hat{\varrho}(\varepsilon\given x)
    :=
    \sum_{j=1}^J \hat{a}_j(x)\bigl\{M_j(\varepsilon,x)-\hat{m}_j(x)\bigr\},
\end{equation}
where $\hat{m}_j$ is estimated via a nonparametric regression of $M_j(\varepsilon,X)$ on $X$, and $\hat{a}_j$ is an appropriately chosen weight function. 
The functions $M_1,\ldots,M_J$ can act as a set of candidate (potentially misspecified) score functions that allow the practitioner to gain partial distributional adaptivity within this `candidate score basis'.  The weight functions $\hat{a}_j$ may be chosen via a data-adaptive procedure~\citep{rose}. 

As a second example, in $M$-estimation, minimising a convex loss function is convenient computationally.  Regarding the loss as a negative log-likelihood, the loss is convex if and only if the error density is log-concave, or equivalently, if and only if its score function is antitonic (decreasing).  \citet{asm} characterise the optimal projection of the score function onto the class of antitonic functions in terms of minimising $\varrho \mapsto V_{P}^{\infty}(x_0)[\varrho]$.  This antitonic projection can be incorporated within outrigger regression; when the conditional density of $\varepsilon$ given $X$ is log-concave, this projection comes at no cost in terms of asymptotic variance of the resulting estimator, while in general the cost is quantifiable.

\medskip\noindent{\bf Beyond conditional mean estimation: } 
Our outrigger proposal extends naturally beyond the estimation of the conditional mean of the response, e.g.~to distributionally adaptive quantile regression.  For $\tau\in(0,1)$, the $\tau$th level \emph{conditional quantile} of the response $Y$ is
\begin{equation*}
    q_\tau(x):=\inf\bigl\{y \in \mathbb{R}:\Pr_P(Y\leq y\given X=x)\geq\tau\bigr\}.
\end{equation*}
In fact, we may write
\begin{equation}\label{eq:qantile-reg}
    Y = q_\tau(X)+\varepsilon_\tau, 
    \qquad \text{where }
    \E_P\bigl(\tau-\ind_{(-\infty,0]}(\varepsilon_\tau)\given X\bigr)=0.
\end{equation}
Under the analogous smoothness assumptions on $\varepsilon_\tau$ (taking the place of $\varepsilon$ in Assumption~\ref{ass:DGM}) and using the moment condition of~\eqref{eq:qantile-reg} for the pilot estimation, our outrigger proposal may be adapted to achieve distributional adaptivity in quantile regression. 

\section{Discussion}

The minimisation of a (penalised) least squares objective remains both foundational and common practice in statistics.  Despite known optimality results in the special case of Gaussian errors, one of the main findings of this work is that, in the context of nonparametric regression, considerable gains are achievable outside this setting.  In fact, our new outrigger local polynomial estimator adapts to an unknown, potentially non-Gaussian, conditional error distribution.  The method, which requires no additional structural assumptions on the error distribution, yields uniform improvements on standard local polynomial estimators, and comes within small constant factors of minimax lower bounds.  We believe the work opens the door to an exploration of the phenomenon of distributional adaptivity in related statistical problems.

\paragraph*{Funding.} EHY and RJS were supported by European Research Council Advanced Grant 101019498.

\bibliography{scorelp.bib}
\newpage
\appendix 

\section{Proof of Theorem~\ref{thm:decomp}}\label{appsec:proof-decomp}


We introduce the following additional notation for the proof of Theorem~\ref{thm:decomp}. 
For $\alpha = (\alpha_1,\ldots,\alpha_d)^\top \in\N_0^d$, define the \emph{multi-index factorial} $\alpha!:=\prod_{r=1}^d\alpha_r!$. 
Given a symmetric matrix $M \in \R^{m\times m}$, we write its minimum eigenvalue as $\Lambda_{\min}(M)$. Recall that our cross-fitted estimator (Algorithm~\ref{alg:outrigger}) is over $\mathcal{K}\geq2$ folds, for which we define the number of observations per fold as $n_k:=|\cI_k| = \lfloor n/\mathcal{K} \rfloor$ for $k\in[\mathcal{K}]$.  It will be convenient to introduce the shorthands $\iota_h(\cdot):=K_h(\cdot)Q_h(\cdot)$ and $\pi_h^{(t)}(\cdot):=K_h^t(\cdot)Q_h(\cdot)Q_h(\cdot)^\top$, for $t\in\{1,2\}$ and $p_0:=p_X(x_0)$.  For $m \in [\bar{p}]$, we write $\mathrm{e}_m\in \mathbb{R}^{\bar{p}}$ for the $m$th standard basis vector, and define the shorthands $\beta_0^*:=\ceil{\beta^*}-1$ and 
\begin{gather}
    \zeta_{\rho^2}(x):=\E_P\{\rho^2(\varepsilon\given X)\given X=x\},
    \qquad
    \zeta_{\rho'}(x):=\E_P\{\rho'(\varepsilon\given X)\given X=x\},
    \\
    \zeta_{(\rho')^2}(x):=\E_P\{(\rho'(\varepsilon\given X))^2\given X=x\},
    \\
    \cE_{\hat\rho,k}(x):=\E_P\big\{(\hat\rho_k-\rho)(\varepsilon\given X)\biggiven X=x, \hat\rho_k\big\},
    \quad
    \cE_{\hat\rho',k}(x):=\E_P\big\{(\hat\rho_k'-\rho')(\varepsilon\given X)\biggiven X=x, \hat\rho_k\big\},
    \\
    \cE_{\hat\rho,k,\text{sq}}(x) := \E_P\big\{(\hat\rho_k-\rho)^2(\varepsilon\given X)\biggiven X=x, \hat\rho_k\big\},
    \quad
    \cE_{\hat\rho',k,\text{sq}}(x) := \E_P\big\{(\hat\rho_k'-\rho')^2(\varepsilon\given X) \biggiven X=x, \hat\rho_k\big\},
    \label{eq:cE-rho'-abs}
\end{gather}
for $x \in \mathbb{R}^d$ and $k \in [\mathcal{K}]$.
For a Borel measurable function $\kappa:\R^d\to\R$ and $t\geq1$, it will help to define
\begin{equation*}
    R_t(\kappa):=\int_{\R^d}|\kappa(\nu)|^t \, d\nu.
\end{equation*}


\begin{proof}[Proof of Theorem~\ref{thm:decomp}]
    Fix $x_0 \in \mathcal{X}$.  Define coefficients $\theta_0=\bigl(f(x_0),0,\ldots,0\bigr)^\top \in\R^{\bar{p}}$, the local error term $e_0(x):=f(x)-f(x_0)$, the population level orthogonaliser
    \begin{equation*}
        \mu(x_0) := \frac{\E_P\big(\iota_h(X-x_0)\big)}{\E_P\big(\kappa_{h,\lambda}(X-x_0)\big)},
    \end{equation*}
    and, for $k \in [\mathcal{K}]$, 
    \begin{equation*}
        c_k(x_0) := \frac{\E_P\big\{\kappa_{h,\lambda}(X-x_0)\big(f(X)-\hat{f}^{\mathrm{LP}}_k(X)\big)\biggiven\hat{f}_k^{\mathrm{LP}}\big\}}{\E_P\big(\kappa_{h,\lambda}(X-x_0)\big)}.
    \end{equation*}
    Let $a_n:=\inf\cH_n$ and $b_n:=\sup\cH_n$. Note that by Assumption~\ref{ass:h} these sequences satisfy $0 < a_n\leq b_n$, $b_n\to0$ and $na_n^d\to \infty$.
    
    For $k\in[\mathcal{K}]$, define functions on $\R^d$ given by
    \begin{equation*}
        \hat\varphi_k(\cdot):=\iota_h(\cdot-x_0)-\hat\mu_k(x_0)\kappa_{h,\lambda}(\cdot-x_0),
    \end{equation*}
    where $\hat{\mu}_k$ is defined in Algorithm~\ref{alg:outrigger}, and the analogous population level version
    \begin{equation*}
        \varphi(\cdot):=\iota_h(\cdot-x_0)-\mu(x_0)\kappa_{h,\lambda}(\cdot-x_0).
    \end{equation*}
Finally, let 
\begin{align*}
\varrho_i := \ind_{\innerr}(X_i)Q_h(X_i-x_0)^\top(\hat\theta-\theta_0)
        &-\ind_{\innerr}(X_i)e_0(X_i) \\
        &+\bigl(\hat{f}^{\mathrm{LP}}_k(X_i)-f(X_i)+\hat{c}_k(x_0)\bigr)\ind_{\outerr}(X_i)
        \end{align*}
        for $k \in [\mathcal{K}]$ and $i \in \mathcal{I}_k$.  Then, by a Taylor expansion with the mean value form of the remainder, there exists $\boldsymbol{\tau}=(\tau_i)_{i\in[n]} \in[0,1]^n$ such that $\hat\theta$ in Algorithm~\ref{alg:outrigger} satisfies
\begin{align}\label{eq:est-eqn}
      0 &=  \frac{1}{n}\sum_{k=1}^{\mathcal{K}}\sum_{i\in\cI_k}\hat\varphi_k(X_i)\hat\rho_k\Big(Y_i-\ind_{\innerr}(X_i)Q_h(X_i \!-\! x_0)^\top\hat\theta-\big(\hat{f}^{\mathrm{LP}}_k(X_i)+\hat{c}_k(x_0)\big)\ind_{\outerr}(X_i)\Biggiven X_i\Big) \nonumber \\
      &= \frac{1}{n}\sum_{k=1}^{\mathcal{K}}\sum_{i\in\cI_k}\hat\varphi_k(X_i)\hat\rho_k(\varepsilon_i - \varrho_i\given X_i) \nonumber \\
      &= \frac{1}{n}\sum_{k=1}^{\mathcal{K}}\sum_{i\in\cI_k}\hat\varphi_k(X_i)\hat\rho_k(\varepsilon_i\given X_i) - \frac{1}{n}\sum_{k=1}^{\mathcal{K}}\sum_{i\in\cI_k}\hat\varphi_k(X_i)\hat\rho_k'(\varepsilon_i-\tau_i\varrho_i\given X_i)\varrho_i.
    \end{align}
    Rearranging, we obtain 
    \begin{align}
    \label{eq:gam-gam}
        \hat{\Omega}(\boldsymbol{\tau})(\hat\theta-\theta_0)
        =
        \hat\phi_{\RN{1}} + \hat\phi_{\RN{2}}(\boldsymbol{\tau}) + \hat\phi_{\RN{3}}(\boldsymbol{\tau}) + \hat\phi_{\RN{4}}(\boldsymbol{\tau}),
    \end{align}
    where
    \begin{align}
        \hat\Omega(\boldsymbol{\tau}) &:= \frac{1}{n}\sum_{k=1}^{\mathcal{K}}\sum_{i\in\cI_k}{\pi_h}(X_i-x_0)\hat\rho_k'(\varepsilon_i-\tau_i\varrho_i\given X_i),
        \\
        \hat\phi_{\RN{1}} &:= \frac{1}{n}\sum_{k=1}^{\mathcal{K}}\sum_{i\in\cI_k}\hat\varphi_k(X_i)\hat\rho_k(\varepsilon_i\given X_i),
        \label{eq:phi-1}
        \\
        \hat\phi_{\RN{2}}(\boldsymbol{\tau}) &:= \frac{1}{n}\sum_{k=1}^{\mathcal{K}}\sum_{i\in\cI_k}\iota_h(X_i-x_0)e_0(X_i)\hat\rho_k'(\varepsilon_i-\tau_i\varrho_i\given X_i),
        \label{eq:phi-2}
        \\
        \hat\phi_{\RN{3}}(\boldsymbol{\tau}) &:= \frac{1}{n}\sum_{k=1}^{\mathcal{K}}\sum_{i\in\cI_k}\bigl\{\hat\rho_k'(\varepsilon_i-\tau_i\varrho_i\given X_i)-\hat\rho_k'(\varepsilon_i\given X_i)\bigr\}\kappa_{h,\lambda}(X_i-x_0)\big(\hat{f}^{\mathrm{LP}}_k(X_i)-f(X_i)+\hat{c}_k(x_0)\big)\hat{\mu}_k(x_0)
        \label{eq:phi-3}
        \\
        \hat\phi_{\RN{4}} &:= \frac{1}{n}\sum_{k=1}^{\mathcal{K}}\sum_{i\in\cI_k}\hat\rho_k'(\varepsilon_i\given X_i)\kappa_{h,\lambda}(X_i-x_0)\big(\hat{f}^{\mathrm{LP}}_k(X_i)-f(X_i)+\hat{c}_k(x_0)\big)\hat\mu_k(x_0),
        \label{eq:phi-4}
    \end{align}
    and where in the above rearrangement we used the identities
    \begin{align*}
        \hat\varphi_k(\cdot)\ind_{\innerr}(\cdot) &=\iota_h(\cdot-x_0), \\
        \hat\varphi_k(\cdot)\ind_{\outerr}(\cdot)&=-\hat\mu_k(x_0)\kappa_{h,\lambda}(\cdot-x_0),\\
        \iota_h(\cdot)Q_h(\cdot)^\top &={\pi_h}(\cdot).
    \end{align*}
    By Lemma~\ref{lem:omega2}, $\hat\Omega(\boldsymbol{\tau})$ is invertible with probability $1-\supoP(1)$. Also by Lemma~\ref{lem:omega2},
    \begin{align}
        \hat{f}(x_0)-f(x_0)
        &= 
        \mathrm{e}_1^\top(\hat\theta-\theta_0)
        \notag
        \\
        &= \mathrm{e}_1^\top\hat\Omega(\boldsymbol{\tau})^{-1}\big( \hat\phi_{\RN{1}} + \hat\phi_{\RN{2}}(\boldsymbol{\tau}) + \hat\phi_{\RN{3}}(\boldsymbol{\tau}) + \hat\phi_{\RN{4}}(\boldsymbol{\tau}) \big)
        \notag
        \\
        &= \mathrm{e}_1^\top\hat\Omega(\boldsymbol{\tau})^{-1}\hat\phi_{\RN{2}}(\boldsymbol{\tau}) + (1+\op(1))\,\mathrm{e}_1^\top\Omega^{-1}\big(\hat\phi_{\RN{1}}+\hat\phi_{\RN{3}}(\boldsymbol{\tau})+\hat\phi_{\RN{4}}(\boldsymbol{\tau})\big)
        ,
        \label{eq:first-expansion}
    \end{align}
    where~$\Omega:=p_0s_1(K)\zeta_{\rho'}(x_0)$. 
    Then each error term may be dealt with in turn:
    \begin{enumerate}[label=(\roman*)]
        \item By Lemma~\ref{lem:phi-1},
        \begin{equation*}
            \mathrm{e}_1^\top\Omega^{-1}\hat\phi_{\RN{1}}
            =
            \mathrm{e}_1^\top\Omega^{-1}\phi_{\RN{1}}+\op\bigg(\frac{1}{\sqrt{nh^d}}+h^{\beta^*}\bigg)
            ,\qquad \phi_{\RN{1}} := \frac{1}{n}\sum_{i=1}^n\varphi(X_i)\rho(\varepsilon_i\given X_i).
        \end{equation*}
        \item By Lemma~\ref{lem:usual-bias},
        \begin{equation*}
            \mathrm{e}_1^\top\hat\Omega(\boldsymbol{\tau})^{-1}\hat\phi_{\RN{2}}(\boldsymbol{\tau}) = B(f,x_0,K,h)h^{\beta^*} + \op\big(h^{\beta^*}\big),
        \end{equation*}
        where 
        \begin{equation}\label{eq:B-def-beta01}
            B(f,x_0,K,h) := h^{-\beta^*}\int_{\mathcal{B}_0(1)}K(\nu)\bigl\{f(x_0+h\nu)-f(x_0)\bigr\}\,d\nu
        \end{equation}
        for $\beta_0^* = 0$ and
        \begin{align}
            &B(f,x_0,K,h):=h^{\beta_0^*-\beta^*} \sum_{\substack{\alpha\in\N_0^d:\\\|\alpha\|_1=\beta_0^*}}\frac{\beta_0^*}{\alpha!}\int_{\mathcal{B}_{0}(1)} K(\nu)\nu^\alpha  
            \notag
            \\
            &\hspace{6cm}\cdot\int_0^1(1-t)^{\beta_0^*-1}\bigl\{\partial^\alpha f(x_0+th\nu)-\partial^\alpha f(x_0)\bigr\} \,dt\,d\nu
            \label{eq:B-def}
        \end{align}
        otherwise.
        \item By Lemma~\ref{lem:phi-3},
        \begin{equation*}
            \mathrm{e}_1^\top\Omega^{-1}\hat\phi_{\RN{3}}(\boldsymbol{\tau}) = \op\bigg(\frac{1}{\sqrt{nh^d}}+h^{\beta^*}\bigg).
        \end{equation*}
        \item By Lemma~\ref{lem:phi-4},
        \begin{equation*}
            \mathrm{e}_1^\top\Omega^{-1}\hat\phi_{\RN{4}}
        =
        \frac{1}{n}\sum_{i=1}^n\kappa_{h,\lambda}(X_i-x_0)\varepsilon_i
        +
        \op\bigg(\frac{1}{\sqrt{nh^d}}+h^{\beta^*}\bigg).
        \end{equation*}
    \end{enumerate}
    Combining all of the above with~\eqref{eq:first-expansion},
    \begin{gather}
        \hat{f}(x_0)-f(x_0) = B(f,x_0,K,h)h^{\beta^*}+\frac{1}{n}\sum_{i=1}^nZ_{n,i}
        +\op\bigg(\frac{1}{\sqrt{nh^d}}+h^{\beta^*}\bigg),
        \label{eq:Zi-decomp}
        \\
        Z_{n,i} := \frac{1}{p_0\zeta_{\rho'}(x_0)}\mathrm{e}_1^\top\varphi(X_i)\rho(\varepsilon_i\given X_i)
         +
        \frac{1}{p_0}\kappa_{h,\lambda}(X_i-x_0)\varepsilon_i.
        \label{eq:Zi}
    \end{gather}
    By Lemma~\ref{lem:clt},
    \begin{equation}
        \lim_{n\to\infty}\sup_{P\in\cP}\sup_{x_0\in\cX}\sup_{h\in\cH_n}\sup_{t\in\R}\bigg| \Pr_P\biggl(\bigg(\frac{1}{nh^d}\bigg\{\frac{R_2(K)V_{P,x_0,\lambda}^{\mathrm{finite}}(\rho)}{p_X(x_0)}\bigg\}\bigg)^{-1/2}\bigg(\frac{1}{n}\sum_{i=1}^nZ_{n,i}\bigg)\leq t
        \biggr)- \Phi(t) \bigg| = 0,
    \end{equation}
    where
    \begin{equation*}
        V_{P,x_0,\lambda}^{\mathrm{finite}}(\rho) := V_P(\rho) + \frac{R_2(\kappa_\lambda)}{R_2(K)}\E_P\bigg\{\biggl(\frac{\rho(\varepsilon\given X)}{\zeta_{\rho'}(x_0)}-\varepsilon\bigg)^2\bigggiven X=x_0\biggr\}.
    \end{equation*}
    Finally, as $\lambda\to\infty$, $R_2(\kappa_\lambda)\to0$, thus 
    $V_{P,x_0,\lambda}^{\mathrm{finite}}(\rho) = V_P(\rho)+\op(1)$. 
    \end{proof}

    \begin{lemma}\label{lem:phi-1}
        Adopting the setup of Theorem~\ref{thm:decomp}, and with $\hat\phi_{\RN{1}}$ as defined in~\eqref{eq:phi-1},
        \begin{equation*}
            \mathrm{e}_1^\top\Omega^{-1}\hat\phi_{\RN{1}}
            =
            \mathrm{e}_1^\top\Omega^{-1}\phi_{\RN{1}}+\op\bigg(\frac{1}{\sqrt{nh^d}}+h^{\beta^*}\bigg),
        \end{equation*}
        where $\phi_{\RN{1}} := \frac{1}{n}\sum_{i=1}^n\varphi(X_i)\rho(\varepsilon_i\given X_i)$.
    \end{lemma}
    \begin{proof}
        We decompose
        \begin{align}
            \hat\phi_{\RN{1}} &= \frac{1}{n}\sum_{i=1}^n
            \varphi(X_i)\rho(\varepsilon_i\given X_i) + \sum_{k=1}^{\mathcal{K}}\frac{1}{n}\sum_{i\in\cI_k}\big(\hat\varphi_k(X_i)-\varphi(X_i)\big)\rho(\varepsilon_i\given X_i)
            \notag
            \\
            &\hspace{1cm} +
            \sum_{k=1}^{\mathcal{K}}\frac{1}{n}\sum_{i\in\cI_k}\varphi(X_i)(\hat\rho_k-\rho)(\varepsilon_i\given X_i)
            +
            \sum_{k=1}^{\mathcal{K}}\frac{1}{n}\sum_{i\in\cI_k}\big(\hat\varphi_k(X_i)-\varphi(X_i)\big)(\hat\rho_k-\rho)(\varepsilon_i\given X_i)
            \notag
            \\
            &=\underbrace{\frac{1}{n}\sum_{i=1}^n
            \varphi(X_i)\rho(\varepsilon_i\given X_i)}_{=:\phi_{\RN{1}}} - \sum_{k=1}^{\mathcal{K}}\underbrace{\big(\hat\mu_k(x_0)-\mu(x_0)\big)\frac{1}{n}\sum_{i\in\cI_k}\kappa_{h,\lambda}(X_i-x_0)\rho(\varepsilon_i\given X_i)}_{=:\hat\phi_{\RN{1}.\mathrm{i}.k}}
            \notag
            \\
            &\hspace{0.5cm}
            +
            \sum_{k=1}^{\mathcal{K}}\underbrace{\frac{1}{n}\sum_{i\in\cI_k}\varphi(X_i)(\hat\rho_k\!-\!\rho)(\varepsilon_i\given X_i)}_{=:\hat\phi_{\RN{1}.\mathrm{ii}.k}}
            \!-\!
            \sum_{k=1}^{\mathcal{K}}\underbrace{\big(\hat\mu_k(x_0)\!-\!\mu(x_0)\big)\frac{1}{n}\sum_{i\in\cI_k}\kappa_{h,\lambda}(X_i\!-\!x_0)(\hat\rho_k\!-\!\rho)(\varepsilon_i\given X_i)}_{=:\hat\phi_{\RN{1}.\mathrm{iii}.k}}.
            \label{eq:phi-1-decomp}
        \end{align}
        We show each of the terms $\hat{\phi}_{\RN{1}.\mathrm{i}.k}, \hat{\phi}_{\RN{1}.\mathrm{ii}.k}, \hat{\phi}_{\RN{1}.\mathrm{iii}.k}$ is $\supoP(1)$ in turn. Define $\eta:=\Omega^{-1}\mathrm{e}_1 \in \mathbb{R}^{\bar{p}}$.
        
        \medskip\noindent{\bf Term $\boldsymbol{\hat\phi_{\RN{1}.\mathrm{i}.k}}$:} 
        For $n$ large enough that $\mathcal{X} + b_n\mathcal{B}_0(1) \subseteq \mathcal{X}^\circ$,
        \begin{align}
                        \E_P\bigg\{\bigg(\frac{1}{n}\sum_{i\in\cI_k}&\kappa_{h,\lambda}(X_i -x_0)\rho(\varepsilon_i\given X_i)\bigg)^2\bigg\}
                        \notag
            \\
            &\leq
            \E_P\bigg\{\frac{1}{n^2}\sum_{i\in\cI_k}\kappa_{h,\lambda}^2(X_i-x_0)\rho^2(\varepsilon_i\given X_i)\bigg\}
            +
            \biggl(\underbrace{\E_P\bigg\{\frac{1}{n}\sum_{i\in\cI_k}\kappa_{h,\lambda}(X_i-x_0)\rho(\varepsilon_i\given X_i)\bigg\}}_{=0\text{ as }\E_P\{\rho(\varepsilon\given X)\given X\}=0}\biggr)^2
            \label{eq:thmdecompext-1}
            \\
            &\leq\frac{1}{n}\E_P\big(\kappa_{h,\lambda}^2(X-x_0)\rho^2(\varepsilon\given X)\big)
            \notag
            \\
            &=
            \frac{1}{nh^d}\int_{\mathcal{B}_{0}(1)} \kappa_{\lambda}^2(\nu) \zeta_{\rho^2}(x_0+h\nu)p_X(x_0+h\nu) \, d\nu
            \notag
            \\
            &\leq \frac{C_1^{2/(2+\delta)}C_X}{nh^d}R_2(\kappa_\lambda) = \supOP\bigg(\frac{R_2(\kappa_\lambda)}{nh^d}\bigg)
            .
            \notag
        \end{align}
        Thus $\frac{1}{n}\sum_{i\in\cI_k}\kappa_{h,\lambda}(X_i-x_0)\rho(\varepsilon_i\given X_i)=\supOP\big(\frac{1}{\sqrt{nh^d}}R_2^{1/2}(\kappa_\lambda)\big)$, and so together with Lemma~\ref{lem:mu(x0)},
        \begin{equation*}
            \eta^\top\hat\phi_{\RN{1}.\mathrm{i}.k} = \supOP\bigg(\frac{R_2^{1/2}(\kappa_\lambda)}{\sqrt{nh^d}}\bigg)\supOP\bigg(\frac{1}{\sqrt{nh^d}}\bigg)
            = \supOP\bigg(\frac{R_2^{1/2}(\kappa_\lambda)}{nh^d}\bigg)
            =\op\bigg(\frac{1}{\sqrt{nh^d}}\bigg).
        \end{equation*}
    
    \medskip\noindent{\bf Term $\boldsymbol{\hat\phi_{\RN{1}.\mathrm{ii}.k}}$:} 
    We have
        \begin{align}
            \E_P\big\{&(\eta^\top\hat\phi_{\RN{1}.\mathrm{ii}.k})^2\biggiven\hat\rho_k\big\}
            \notag
            \\
            &=
            \frac{1}{n^2}\sum_{i\in\cI_k}\E_P\big\{\bigl(\eta^\top\varphi(X_i)\bigr)^2(\hat\rho_k-\rho)^2(\varepsilon_i\given X_i)\biggiven\hat\rho_k\big\}
            \notag
            \\
            &\qquad\qquad+
            \frac{1}{n^2}\sum_{i\in\cI_k}\sum_{j\in\cI_k\setminus\{i\}}\E_P\big\{\eta^\top\varphi(X_i)(\hat\rho_k-\rho)(\varepsilon_i\given X_i)\biggiven\hat\rho_k\big\}\E_P\big\{\eta^\top\varphi(X_j)(\hat\rho_k-\rho)(\varepsilon_j\given X_j)\biggiven\hat\rho_k\big\}
            \notag
            \\
            &\leq
            \frac{1}{n}\E_P\big\{\bigl(\eta^\top\varphi(X)\bigr)^2(\hat\rho_k-\rho)^2(\varepsilon\given X)\biggiven\hat\rho_k\big\}
            +
            \Big(\E_P\big\{\eta^\top\varphi(X)(\hat\rho_k-\rho)(\varepsilon\given X)\biggiven\hat\rho_k\big\}\Big)^2
            \notag
            \\
            &=
            \frac{1}{n}\E_P\big\{\bigl(\eta^\top\varphi(X)\bigr)^2\,\cE_{\hat\rho,k,\text{sq}}(X)\biggiven\hat\rho_k\big\}
            +
            \Big(\E_P\big\{\eta^\top\varphi(X)\,\cE_{\hat\rho,k}(X)\biggiven\hat\rho_k\big\}\Big)^2.
            \label{eq:phi-1-2-decomp}
        \end{align}
        For the first term in~\eqref{eq:phi-1-2-decomp}, 
        \begin{align*}
            \E_P\bigl\{\bigl(\eta^\top\varphi(X)\bigr)^2\bigr\}
            &\leq
            \|\eta\|_2^2\,\E_P\bigl(\|\varphi(X)\|_2^2\bigr)
            \\
            &=
            \|\eta\|_2^2\,\E_P\Bigl(K_h^2(X-x_0)\|Q_h(X-x_0)\|_2^2+\|\mu(x_0)\|_2^2\,\kappa_{h,\lambda}^2(X-x_0)\Bigr)
            \\
            &\leq \frac{1}{c_X^2 c_1^2 \Lambda_{\min}^2\bigl(s_1(K)\bigr)}\Bigl(e^d\,\E_P\bigl(K_h^2(X-x_0)\bigr) + \|\mu(x_0)\|_2^2\,\E_P\bigl(\kappa_{h,\lambda}^2(X-x_0)\bigr)\Bigr)
            \\
            &=\supOP\bigl(h^{-d}\bigr)
            ,
        \end{align*}
        where we applied Lemma~\ref{lem:Q} in the penultimate line. 
        Thus
        \begin{align}
            \E_P\big\{\bigl(\eta^\top\varphi(X)\bigr)^2\,\cE_{\hat\rho,k,\text{sq}}(X)\biggiven\hat\rho_k\big\}
            &\leq
            \E_P\bigl\{(\eta^\top\varphi(X))^2\bigr\}\sup_{x\in\cX+\mathcal{B}_0(b_n)}\cE_{\hat\rho,k,\mathrm{sq}}(x)
            =\op\big(h^{-d}\big).
            \label{eq:phi-1-2-var}
        \end{align}
        For the second term in~\eqref{eq:phi-1-2-decomp}, define $\bar\cE_{\hat\rho,k}(x):=\cE_{\hat\rho,k}(x)-\cE_{\hat\rho,k}(x_0)$. 
        We remark that in the special case where $L_\cE=0$ (which holds if $\varepsilon\indep X$), we simply obtain $\bar\cE_{\hat\rho,k}(x)=0$ and so $\cE_{\hat\rho,k}(\cdot)$ is constant, which together with $\E_P\{\varphi(X)\}=0$ means that the second term in~\eqref{eq:phi-1-2-decomp} is zero. 
        In general, however, we accommodate  $\varepsilon\not\hspace{-0.1em}\indep X$ (i.e.~$L_\cE > 0$). 
        
        Observing that $\E_P\{\varphi(X)\}=0$ by definition of $\mu(\cdot)$, we have
        \begin{align*}
            \E_P\big\{\eta^\top\varphi(X)&\cE_{\hat\rho,k}(X)\biggiven\hat\rho_k\big\}
            \\
            &=\E_P\big\{\eta^\top\varphi(X)\bar\cE_{\hat\rho,k}(X)\biggiven\hat\rho_k\big\}
            \\
            &=
            \eta^\top\E_P\bigl\{\iota_h(X-x_0)\bar\cE_{\hat\rho,k}(X)\biggiven\hat\rho_k\bigr\}-\eta^\top\mu(x_0)\E_P\bigl\{\kappa_{h,\lambda}(X-x_0)\bar\cE_{\hat\rho,k}(X)\biggiven\hat\rho_k\bigr\}
            \\
            &=
            \frac{1}{p_0\zeta_{\rho'}(x_0)}\int_{\mathcal{B}_{0}(1)} K(\nu)\bar\cE_{\hat\rho,k}(x_0+h\nu)p_X(x_0+h\nu) \, d\nu \\
            &\hspace{3cm}
            - \eta^\top\mu(x_0)\int_{\mathcal{B}_{0}(\lambda)\setminus\mathcal{B}_{0}(1)}\kappa_{\lambda}(\nu)\bar\cE_{\hat\rho,k}(x_0+h\nu)p_X(x_0+h\nu) \, d\nu
            \\
            &=
            \frac{1}{p_0\zeta_{\rho'}(x_0)}\int_{\mathcal{B}_{0}(1)} K(\nu)\bar\cE_{\hat\rho,k}(x_0+h\nu)p_X(x_0+h\nu)d\nu \\
            &\hspace{3cm}
            - \eta^\top\mu(x_0)\int_{\mathcal{B}_{0}(\lambda)\setminus\mathcal{B}_{0}(1)}\kappa_{\lambda}(\nu)\bar\cE_{\hat\rho,k}(x_0+h\nu)p_X(x_0+h\nu) \, d\nu,
        \end{align*}
        where the third equality makes use of the fact that $s_1(K)\mathrm{e}_1=\mathrm{e}_1$ for our order $p+1$ kernel, so that 
        \begin{equation*}
            \eta^\top Q(\nu)
            =
            \frac{\mathrm{e}_1^\top s_1(K)^{-1}Q(\nu)}{p_0\zeta_{\rho'}(x_0)}
            =
            \frac{\mathrm{e}_1^\top Q(\nu)}{p_0\zeta_{\rho'}(x_0)}
            =
            \frac{1}{p_0\zeta_{\rho'}(x_0)}.
        \end{equation*}

        Define $\bar\beta:=\beta_X\wedge\beta_\cE$ and  $\bar\beta_0:=\ceil{\bar\beta}-1$. Then $\bar\cE_{\hat\rho,k}p_X \in \cH(\bar\beta,\bar{L})$ on $\mathcal{X}^\circ$ for some $\bar{L}\in(0,\infty)$ that depends only on $L_X,L_\cE,\beta_X,\beta_\cE$. 
        Moreover, $\bar\cE_{\hat\rho,k}(x_0)p_X(x_0)=0$.  By a Taylor expansion, for any $x\in\R^d$, there exists $t_x\in[0,1]$ such that
        \begin{align*}
            \bar\cE_{\hat\rho,k}(x)p_X(x)
            &=
            \sum_{\substack{\alpha\in\N_0^d:\\1\leq\|\alpha\|_1\leq\bar\beta_0}}\frac{1}{\alpha!}\big(\partial^\alpha(\bar\cE_{\hat\rho,k}p_X)(x_0)\big)(x-x_0)^{\alpha}
            \\
            &\hspace{1cm}
            +
            \sum_{\substack{\alpha\in\N_0^d:\\\|\alpha\|_1=\bar\beta_0}}\frac{1}{\alpha!}
            \Big(\partial^\alpha\bigl(\bar\cE_{\hat\rho,k}p_X)\bigl(x_0+t_x(x-x_0)\bigr)-\partial^\alpha(\bar\cE p_X)(x_0)\Bigr)(x-x_0)^{\alpha}.
        \end{align*}
        Then since $K$ is a kernel of order $p+1\geq\ceil{\bar{\beta}}$, and making use of~\eqref{eq:multinom-expansion},
        \begin{align*}
            \bigg|\int_{\mathcal{B}_{0}(1)} &K(\nu)\bar\cE_{\hat\rho,k}(x_0+h\nu)p_X(x_0+h\nu) \, d\nu\bigg|
            \\
            &= 
h^{\bar{\beta}_0}\biggl|\sum_{\substack{\alpha\in\N_0^d:\\\|\alpha\|_1=\bar\beta_0}}\frac{1}{\alpha!}
            \int_{\mathcal{B}_{0}(1)} K(\nu)\big\{\partial^\alpha(\bar\cE_{\hat\rho,k}p_X)(x_0+t_{x_0+h\nu}h\nu)-\partial^\alpha(\bar\cE_{\hat\rho,k}p_X)(x_0)\big\}
            \nu^{\alpha} \, d\nu\biggr|
            \\
            &\leq \bar{L}\,h^{\bar{\beta}_0}\sum_{\substack{\alpha\in\N_0^d:\\\|\alpha\|_1=\bar\beta_0}}\frac{1}{\alpha!}
            \int_{\mathcal{B}_{0}(1)} |K(\nu)\nu^\alpha|\cdot\|h\nu\|^{\bar{\beta}-\bar{\beta}_0} \,
            d\nu 
            \leq \frac{\bar{L}R_1(K)d^{\bar{\beta}_0}}{\bar{\beta}_0!}h^{\bar{\beta}}
            = \supOP\big(h^{\bar{\beta}}\big).
        \end{align*}
        Furthermore, by similar arguments, and since 
        $\kappa_{\lambda}$ is a kernel of order $\ceil{\bar\beta}$, 
        it follows that
        \begin{align*}
            \bigg|\int_{\mathcal{B}_{0}(\lambda)\setminus\mathcal{B}_{0}(1)} \kappa_{\lambda}(\nu)&\bar\cE_{\hat\rho,k}(x_0+h\nu)p_X(x_0+h\nu) \, d\nu\bigg|
            \\
            &\leq 
\bar{L}\,h^{\bar{\beta}}\sum_{\substack{\alpha\in\N_0^d:\\\|\alpha\|_1=\bar\beta_0}}\frac{1}{\alpha!}
            \int_{\mathcal{B}_{0}(\lambda)\setminus\mathcal{B}_{0}(1)} |\kappa_{\lambda}(\nu)\nu^\alpha|\cdot\|\nu\|^{\bar{\beta}-\bar{\beta}_0} \, 
            d\nu \leq
            \frac{\bar{L}R_1(\kappa)d^{\bar{\beta}_0}}{\bar{\beta}_0!}(\lambda h)^{\bar{\beta}}.
        \end{align*}
        Therefore
        \begin{equation*}
            \E_P\big\{\eta^\top\varphi(X)\,\cE_{\hat\rho,k}(X)\biggiven\hat\rho_k\big\} = \supOP\big((\lambda h)^{\bar{\beta}}\big).
        \end{equation*}
        Combining this with~\eqref{eq:phi-1-2-decomp}~and~\eqref{eq:phi-1-2-var},
        \begin{equation*}
            \E_P\big\{(\eta^\top\hat\phi_{\RN{1}.\mathrm{ii}.k})^2\biggiven\hat\rho_k\big\} 
            = \op\bigg(\frac{1}{nh^d}\bigg) + \supOP\big((\lambda h)^{2\bar{\beta}}\big)
            =
            \op\bigg(\frac{1}{nh^d} + h^{2\beta^*}\bigg).
        \end{equation*}

    \medskip\noindent{\bf Term $\boldsymbol{\hat\phi_{\RN{1}.\mathrm{iii}.k}}$:} First,
    \begin{align*}
            \bigg|\E_P\bigg\{\frac{1}{n}\sum_{i\in\cI_k}&\kappa_{h,\lambda}(X_i-x_0)(\hat\rho_k-\rho)(\varepsilon_i\given X_i)\bigggiven\hat\rho_k,(X_i)_{i\in\cI_k}\bigg\}\bigg|
            \\
            &=\bigg|\frac{1}{n}\sum_{i\in\cI_k}\kappa_{h,\lambda}(X_i-x_0)\cE_{\hat\rho,k}(X_i)\bigg|
            \\
            &\leq
            \frac{1}{n}\sum_{i\in\cI_k} \bigl|\kappa_{h,\lambda}(X_i-x_0)\bigr|\cdot|\cE_{\hat\rho,k}(x_0)|
            +
            \frac{1}{n}\sum_{i\in\cI_k}|\kappa_{h,\lambda}(X_i-x_0)|\cdot\big|\cE_{\hat\rho,k}(X_i)-\cE_{\hat\rho,k}(x_0)\big|
            \\
            &\leq
            \frac{1}{n}\sum_{i\in\cI_k} \bigl|\kappa_{h,\lambda}(X_i-x_0)\bigr|\cdot|\cE_{\hat\rho,k}(x_0)|
            +
            L_{\cE}\bigg(\frac{1}{n}\sum_{i\in\cI_k}|\kappa_{h,\lambda}(X_i-x_0)|\cdot\|X_i-x_0\|^{\beta_{\cE} \wedge 1}\bigg)
            \\
            &\leq 
            \bigl\{
            |\cE_{\hat\rho,k}(x_0)| + L_{\cE}(\lambda h)^{\beta_{\cE} \wedge 1}
            \bigr\}
            \frac{1}{n}\sum_{i\in\cI_k} \bigl|\kappa_{h,\lambda}(X_i-x_0)\bigr|
            .
        \end{align*}
        Now for $n$ large enough that $\mathcal{X} + \lambda b_n\mathcal{B}_0(1) \subseteq \mathcal{X}^\circ$,
        \begin{equation*}
            \E_P\bigg\{\frac{1}{n}\sum_{i\in\cI_k}\bigl|\kappa_{h,\lambda}(X_i-x_0)\bigr|\bigg\} \leq \E_P\bigl\{|\kappa_{h,\lambda}(X-x_0)|\bigr\} = \int_{\mathcal{B}_{0}(\lambda)\setminus\mathcal{B}_{0}(1)} |\kappa_{\lambda}(\nu)|p_X(x_0+h\nu) \, d\nu
            \leq
            C_XR_1(\kappa_\lambda),
        \end{equation*}
        and so $\frac{1}{n}\sum_{i\in\cI_k}\kappa_{h,\lambda}(X_i-x_0)=\supOP\big(R_1(\kappa_\lambda)\big)$. 
        Thus, together with Assumption~\ref{ass:DGM},
        \begin{equation*}
            \bigg|\E_P\bigg\{\frac{1}{n}\sum_{i\in\cI_k}\kappa_{h,\lambda}(X_i-x_0)(\hat\rho_k-\rho)(\varepsilon_i\given X_i)\bigggiven\hat\rho_k,(X_i)_{i\in\cI_k}\bigg\}\bigg| 
            =
            \op\bigl(R_1(\kappa_\lambda)\bigr),
        \end{equation*}
        and so
        \begin{equation*}
            \frac{1}{n}\sum_{i\in\cI_k}\kappa_{h,\lambda}(X_i-x_0)(\hat\rho_k-\rho)(\varepsilon_i\given X_i)=\op\bigl(R_1(\kappa_\lambda)\bigr).
        \end{equation*}
        Finally, then, we conclude by Lemma~\ref{lem:mu(x0)} that
        \begin{equation*}
            \hat\phi_{\RN{1}.\mathrm{iii}.k}
            =\supOP\bigg(\frac{1}{\sqrt{nh^d}}\bigg)\op\big(R_1(\kappa_\lambda)\big)
            =\op\bigg(\frac{R_1(\kappa_\lambda)}{\sqrt{nh^d}}\bigg)
            =\op\bigg(\frac{1}{\sqrt{nh^d}}\bigg),
        \end{equation*}
as required.
    \end{proof}

    \begin{lemma}\label{lem:mu(x0)}
        Adopting the setup of Theorem~\ref{thm:decomp}, for each $k \in [\mathcal{K}]$, 
        \begin{equation*}
            \hat\mu_k(x_0) - \mu(x_0) = \supOP\bigg(\frac{1}{\sqrt{nh^d}}\bigg).
        \end{equation*}
    \end{lemma}
    \begin{proof}
        We consider the numerator and denominator term in $\hat\mu_k(x_0)$ separately. For the numerator, by Lemma~\ref{lem:Q},
        \begin{multline*}
            \E_P\bigg\{\bigg\|\frac{1}{|\cI_k^c|}\sum_{i\in\cI_k^c}\iota_h(X_i-x_0)-\E_P\bigl\{\iota_h(X-x_0)\bigr\}\bigg\|_2^2\bigg\}
            \leq
            \frac{1}{|\cI_k^c|}\E_P\big\{\|\iota_h(X-x_0)\|_2^2\big\}
            \\
            =
            \frac{1}{|\cI_k^c|h^d}\int_{\mathcal{B}_0(1)} K^2(\nu)\|Q(\nu)\|_2^2\,p_X(x_0+h\nu) \, d\nu
            \leq
            \frac{C_Xe^dR_2(K)}{|\cI_k^c|h^d}
            =
            \supOP\bigg(\frac{1}{nh^d}\bigg),
        \end{multline*}
        and so
        \begin{equation*}
            \frac{1}{|\cI_k^c|}\sum_{i\in\cI_k^c}\iota_h(X_i-x_0) = \E_P\bigl\{\iota_h(X-x_0)\bigr\} + \supOP\bigg(\frac{1}{\sqrt{nh^d}}\bigg).
        \end{equation*}
        For the denominator,
        \begin{multline*}
            \E_P\bigg\{\bigg(\frac{1}{|\cI_k^c|}\sum_{i\in\cI_k^c}\kappa_{h,\lambda}(X_i-x_0)-\E_P\bigl\{\kappa_{h,\lambda}(X-x_0)\bigr\}\bigg)^2\bigg\}
            \leq
            \frac{1}{|\cI_k^c|}\E_P\big\{\kappa_{h,\lambda}^2(X-x_0)\big\}
            \\
            \leq
            \frac{1}{|\cI_k^c|h^d}\int_{\mathcal{B}_0(\lambda)\setminus\mathcal{B}_0(1)}\kappa_{\lambda}^2(\nu)p_X(x_0+h\nu) \, d\nu
            \leq
            \frac{C_X}{|\cI_k^c|h^d}R_2(\kappa_\lambda)
            =
            \supOP\bigg(\frac{R_2(\kappa_\lambda)}{nh^d}\bigg),
        \end{multline*}
        so
        \begin{equation}\label{eq:mu-or-c-den}
            \frac{1}{|\cI_k^c|}\sum_{i\in\cI_k^c}\kappa_{h,\lambda}(X_i-x_0) = \E_P\bigl\{\kappa_{h,\lambda}(X-x_0)\bigr\}+\supOP\bigg(\frac{R_2^{1/2}(\kappa_\lambda)}{\sqrt{nh^d}}\bigg).
        \end{equation}
        Therefore, using the fact that $\sup_{\lambda \geq \lambda_0(K)} R_2(\kappa_\lambda) = R_2(\kappa_{\lambda_0(K)}) = R_2(K) < \infty$, we have
        \[
            \hat\mu_k(x_0)=\frac{\E_P\{\iota_h(X-x_0)\}\big\{1+\supOP\big((nh^d)^{-1/2}\big)\big\}}{\E_P\{\kappa_{h,\lambda}(X-x_0)\}\big\{1+\supOP\big((nh^d)^{-1/2}R_2^{1/2}(\kappa_\lambda)\big)\big\}}=\mu(x_0)\big\{1+\supOP\big((nh^d)^{-1/2}\big)\big\}.
        \]
        Finally, as
        \begin{equation}\label{eq:mu(x_0)bound}
            \mu(x_0)=\frac{\int_{\mathcal{B}_0(1)} K(\nu)Q(\nu) \, d\nu}{\int_{\mathcal{B}_0(\lambda)\setminus\mathcal{B}_0(1)}\kappa_{\lambda}(\nu) \, d\nu}\big\{1+\op(1)\big\}
            =
            u(K)\big\{1+\op(1)\big\},
        \end{equation}
        we have
        \begin{equation*}
            \hat\mu_k(x_0)-\mu(x_0)=\supOP\bigg(\frac{1}{\sqrt{nh^d}}\bigg),
        \end{equation*}
        as required.
    \end{proof}

    \begin{lemma}\label{lem:c(x0)}
        Adopting the setup of Theorem~\ref{thm:decomp}, for each $k \in [\mathcal{K}]$, 
        \begin{equation*}
            \hat{c}_k(x_0) - c_k(x_0) = \supOP\bigg(\frac{1}{\sqrt{nh^d}}R_2^{1/2}(\kappa_\lambda)\bigg).
        \end{equation*}
    \end{lemma}
    \begin{proof}
        We consider the numerator and denominator in $\hat{c}_k(x_0)$ separately. 
        For the numerator,
        \begin{align*}
            &\quad\;
            \E_P\bigg\{\bigg(\frac{1}{|\cI_k|}\sum_{i\in\cI_k}\kappa_{h,\lambda}(X_i-x_0)\big\{Y_i-\hat{f}^{\mathrm{LP}}_k(X_i)\big\}
            -
            \E_P\Big(\kappa_{h,\lambda}(X-x_0)\big\{Y-\hat{f}^{\mathrm{LP}}_k(X)\big\}\Biggiven\hat{f}^{\mathrm{LP}}_k\Big)\bigg)^2\bigggiven \hat{f}^{\mathrm{LP}}_k\bigg\}
            \\
            &\leq
            \frac{1}{|\cI_k|}\E_P\Big\{\kappa_{h,\lambda}^2(X-x_0)\big\{Y-\hat{f}^{\mathrm{LP}}_k(X)\big\}^2\Biggiven\hat{f}^{\mathrm{LP}}_k\Big\}
            \\
            &=
            \frac{1}{|\cI_k|}\bigg\{
            \E_P\big(\kappa_{h,\lambda}^2(X-x_0)\E_P(\varepsilon^2\given X)\big)
            +
            \E_P\big(\kappa_{h,\lambda}^2(X-x_0)\big\{\hat{f}^{\mathrm{LP}}_k(X)-f(X)\big\}^2\biggiven\hat{f}^{\mathrm{LP}}_k\big)
            \bigg\}
            \\
            &\leq \frac{1}{|\cI_k|}\bigg\{\bigl(C_3+\supoP(1)\bigr)\E_P\bigl\{\kappa_{h,\lambda}^2(X-x_0)\bigr\} 
            + \E_P\big(\kappa_{h,\lambda}^2(X-x_0)\big\{\hat{f}^{\mathrm{LP}}_k(X)-f(X)\big\}^2\biggiven\hat{f}^{\mathrm{LP}}_k\big)
            \bigg\}
            \\
            &=
            \supOP\bigg(\frac{R_2(\kappa_\lambda)}{nh^d}\bigg),
        \end{align*}
        where the final line follows using the fact that 
        \begin{align*}
            \E_P\Bigl\{\kappa_{h,\lambda}^2(X-x_0)\bigl\{\hat{f}^{\mathrm{LP}}_k(X)-f(X)\bigr\}^2
            \Bigr\} &\leq
            \E_P\big\{\kappa_{h,\lambda}^2(X-x_0)\big\}\sup_{x\in\cX}\E_P\Big\{\bigl(\hat{f}^{\mathrm{LP}}_k(x)-f(x)\bigr)^2\Big\},
        \end{align*}
        and Lemma~\ref{lem:LP-rate}. 
        Therefore
        \begin{equation*}
            \frac{1}{|\cI_k|}\sum_{i\in\cI_k}\kappa_{h,\lambda}(X_i-x_0)\big\{Y_i-\hat{f}^{\mathrm{LP}}_k(X_i)\big\}
            =
            \E_P\Big(\kappa_{h,\lambda}(X-x_0)\big\{Y-\hat{f}^{\mathrm{LP}}_k(X)\big\}\Biggiven\hat{f}^{\mathrm{LP}}_k\Big)
            +
            \supOP\bigg(\frac{R_2^{1/2}(\kappa_\lambda)}{\sqrt{nh^d}}\bigg).
        \end{equation*}
        The denominator is the same as the denominator in~$\hat\mu_k(x_0)$; see~\eqref{eq:mu-or-c-den}. 
        Thus,
        \begin{equation*}
            \hat{c}_k(x_0)=c_k(x_0)\bigg\{1+\supOP\bigg(\frac{1}{\sqrt{nh^d}}R_2^{1/2}(\kappa_\lambda)\bigg)\bigg\}.
        \end{equation*}
        Now,
        \begin{align}
            \E_P\big(|c_k(x_0)|\big) &\leq \frac{\E_P\big\{|\kappa_{h,\lambda}(X-x_0)|\cdot\bigl|\hat{f}^{\mathrm{LP}}_k(X)-f(X)\bigr|\big\}}{|\E_P\{\kappa_{h,\lambda}(X-x_0)\}|}
            \notag
            \\
            &\leq
            \frac{\E_P\{|\kappa_{h,\lambda}(X-x_0)|\}}{|\E_P\{\kappa_{h,\lambda}(X-x_0)\}|}\cdot\sup_{x\in\mathcal{B}_{x_0}(\lambda h)}\E_P\bigl\{\bigl|\hat{f}^{\mathrm{LP}}_k(x)-f(x)\bigr|\bigr\}
            \notag
            \\
            &=\supOP\bigg(\frac{1}{\sqrt{nh^d}}+h^{\beta^*}\bigg)
            ,
            \label{eq:c_k}
        \end{align}
       where we use Jensen's inequality and Lemma~\ref{lem:LP-rate}. 
        Thus 
        \begin{equation*}
            \hat{c}_k(x_0)-c_k(x_0) = \supOP\bigg(\frac{1}{\sqrt{nh^d}}R_2^{1/2}(\kappa_\lambda)\bigg),
        \end{equation*}
        as required.
    \end{proof}

    \begin{lemma}\label{lem:usual-bias}
        Adopting the setup of Theorem~\ref{thm:decomp}, and with $\hat\phi_{\RN{2}}(\boldsymbol{\tau})$ defined in~\eqref{eq:phi-2} and $B(f,x_0,K,h)$ defined as in~\eqref{eq:B-def},
        \begin{equation*}
\mathrm{e}_1^\top\hat\Omega(\boldsymbol{\tau})^{-1}\hat\phi_{\RN{2}}(\boldsymbol{\tau}) = \bigl(1+\op(1)\bigr)B(f,x_0,K,h)h^{\beta^*}.
        \end{equation*}
    \end{lemma}
    \begin{proof}
    For all $\alpha\in\N_0^d$ with $\|\alpha\|_1\leq \beta_0^*$, and writing $m_\alpha\in[\bar{p}]$ as the index such that $\{Q(u)\}_{m_\alpha}=\frac{1}{\alpha!}u^\alpha$,
    \begin{align*}
        &\quad
        \bigg(\frac{1}{n}\sum_{k=1}^{\mathcal{K}}\sum_{i\in\cI_k}{\pi_h}(X_i-x_0)\hat\rho_k'(\varepsilon_i-\tau_i\varrho_i\given X_i)\bigg)^{-1}\bigg(\frac{1}{n}\sum_{k=1}^{\mathcal{K}}\sum_{i\in\cI_k}\iota_h(X_i-x_0)(X_i-x_0)^\alpha\hat\rho_k'(\varepsilon_i-\tau_i\varrho_i\given X_i)\bigg)
        \\
        &=\alpha!h^{\|\alpha\|_1}\bigg(\frac{1}{n}\sum_{k=1}^{\mathcal{K}}\sum_{i\in\cI_k}{\pi_h}(X_i-x_0)\hat\rho_k'(\varepsilon_i-\tau_i\varrho_i\given X_i)\bigg)^{-1}\bigg(\frac{1}{n}\sum_{k=1}^{\mathcal{K}}\sum_{i\in\cI_k}{\pi_h}(X_i-x_0)\hat\rho_k'(\varepsilon_i-\tau_i\varrho_i\given X_i)\bigg)\mathrm{e}_{m_\alpha}
        \\
        &=\alpha!h^{\|\alpha\|_1}\mathrm{e}_{m_\alpha},
    \end{align*}
    where we use the fact that $u^\alpha=\alpha!h^{\|\alpha\|_1}Q_h(u)^\top\mathrm{e}_{m_\alpha}$. As the components of $Q$ are in lexicographical order, for any $\alpha\in\N_0^d$ with $1\leq\|\alpha\|_1\leq\beta_0^*$, 
    \begin{equation}\label{eq:destroy}
        \mathrm{e}_1^\top \bigg(\frac{1}{n}\sum_{k=1}^{\mathcal{K}}\sum_{i\in\cI_k}{\pi_h}(X_i-x_0)\hat\rho_k'(\varepsilon_i-\tau_i\varrho_i\given X_i)\bigg)^{-1}\bigg(\frac{1}{n}\sum_{k=1}^{\mathcal{K}}\sum_{i\in\cI_k}\iota_h(X_i-x_0)(X_i-x_0)^\alpha\hat\rho_k'(\varepsilon_i-\tau_i\varrho_i\given X_i)\bigg)=0.
    \end{equation}
    By taking a Taylor expansion with an integral form of remainder, and noticing that $e_0(x_0)=0$, we obtain the decomposition
    \begin{align*}
        e_0(X_i)
        &=
        \sum_{\substack{\alpha\in\N_0^d:\\1\leq\|\alpha\|_1\leq\beta_0^*}}\frac{1}{\alpha!}\,\partial^\alpha f(x_0)(X_i-x_0)^\alpha
        + e^*_0(X_i)
        \end{align*}
        where
        \begin{align}\label{eq:e0*}
        e_0^*(X_i) &:= \sum_{\substack{\alpha\in\N_0^d:\\\|\alpha\|_1=\beta_0^*}}\frac{\beta_0^*}{\alpha!}\bigg(\int_0^1 (1-t)^{\beta_0^*-1}\bigl\{\partial^\alpha f\bigl(x_0+t(X_i-x_0)\bigr)-\partial^\alpha f(x_0)\bigr\} \, dt\bigg)(X_i-x_0)^{\alpha}.
    \end{align}
    In combination with~\eqref{eq:destroy}, this yields
    \begin{equation}\label{eq:de-staring}
        \mathrm{e}_1^\top\hat\Omega(\boldsymbol{\tau})^{-1}\hat\phi_{\RN{2}}(\boldsymbol{\tau}) = \mathrm{e}_1^\top\hat\Omega(\boldsymbol{\tau})^{-1}\hat\phi_{\RN{2}}^*(\boldsymbol{\tau}),
    \end{equation}
    where 
    \begin{equation*}
        \hat\phi_{\RN{2}}^*(\boldsymbol{\tau}) := \frac{1}{n}\sum_{k=1}^{\mathcal{K}}\sum_{i\in\cI_k}\iota_h(X_i-x_0)e_0^*(X_i)\hat\rho_k'(\varepsilon_i-\tau_i\varrho_i\given X_i).
    \end{equation*}
    Now, by Lemma~\ref{lem:omega2},
    \begin{equation*}
        \mathrm{e}_1^\top\hat\Omega(\boldsymbol{\tau})^{-1}\hat\phi_{\RN{2}}^*(\boldsymbol{\tau})
        =
        \bigl(1+\op(1)\bigr)\mathrm{e}_1^\top\Omega^{-1}\hat\phi_{\RN{2}}^*(\boldsymbol{\tau}),
    \end{equation*}
    and
    \begin{align}
        \mathrm{e}_1^\top\Omega^{-1}\hat\phi_{\RN{2}}^*(\boldsymbol{\tau})
        &=
        \underbrace{\frac{1}{n}\sum_{i=1}^n\big(\mathrm{e}_1^\top\Omega^{-1}\iota_h(X_i-x_0)\big)e_0^*(X_i)\rho'(\varepsilon_i\given X_i)}_{=:\phi_{\RN{2}.\mathrm{i}}^*}
        \notag
        \\
        &\qquad\qquad
        +
        \sum_{k=1}^{\mathcal{K}}\underbrace{\frac{1}{n}\sum_{i\in\cI_k}\big(\mathrm{e}_1^\top\Omega^{-1}\iota_h(X_i-x_0)\big)e_0^*(X_i)\big\{\hat\rho_k'(\varepsilon_i\given X_i)-\rho'(\varepsilon_i\given X_i)\big\}}_{=:\hat\phi_{\RN{2}.\mathrm{ii}.k}}
        \notag
        \\
        &\qquad\qquad
        +
        \sum_{k=1}^{\mathcal{K}}\underbrace{\frac{1}{n}\sum_{i\in\cI_k}\big(\mathrm{e}_1^\top\Omega^{-1}\iota_h(X_i-x_0)\big)e_0^*(X_i)\big\{\hat\rho_k'(\varepsilon_i-\tau_i\varrho_i\given X_i)-\hat\rho_k'(\varepsilon_i\given X_i)\big\}}_{=:\hat\phi_{\RN{2}.\mathrm{iii}.k}}
        .
        \label{eq:phi2-decomp}
    \end{align}
    We proceed by treating each term in the above decomposition separately.

    \medskip\noindent{\bf Term $\boldsymbol{\phi_{\RN{2}.\mathrm{i}}^*}$:} 
    Recalling the definition of $C_{e^*}$ from Lemma~\ref{lem:e-bound}, we have
    \begin{align*}
        \Var_P(\phi_{\RN{2}.\mathrm{i}}^*)
        &=
        \frac{1}{n}\Var_P\Big(\big(\mathrm{e}_1^\top\Omega^{-1}\iota_h(X-x_0)\big)e_0^*(X)\rho'(\varepsilon\given X)\Big)
        \\
        &\leq
        \frac{1}{n}\E_P\Big(\big(\mathrm{e}_1^\top\Omega^{-1}\iota_h(X-x_0)\big)^2\bigl\{e_0^*(X)\bigr\}^2\bigl\{\rho'(\varepsilon\given X)\bigr\}^2\Big)
        \\
        &\leq \frac{1}{n}C_{e^*}^2(L,d,\beta^*)h^{2\beta^*}\E_P\Big(\big(\mathrm{e}_1^\top\Omega^{-1}\iota_h(X-x_0)\big)^2\zeta_{(\rho')^2}(X)\Big)
        \\
        &\leq 
        \frac{C_{e^*}^2(L,d,\beta^*)h^{2\beta^*}}{p_0^2 \zeta_{\rho'}^2(x_0) nh^d} 
        \int_{\mathcal{B}_0(1)} K^2(\nu) \zeta_{(\rho')^2}(x_0+h\nu) p_X(x_0+h\nu) \, d\nu
        \\
        &\leq 
        \frac{C_XC_2R_2(K)C_{e^*}^2(L,d,\beta^*)h^{2\beta^*}}{c_1^2p_0^2 \zeta_{\rho'}^2(x_0) nh^d}
        \bigl(1+\supoP(1)\bigr) = \supOP\bigg(\frac{h^{2\beta^*}}{nh^d}\bigg),
    \end{align*}
    by Lemmas~\ref{lem:e-bound}~and~\ref{lem:Q}.  
    Therefore, by Chebychev's inequality,
    \begin{align*}
        \phi_{\RN{2}.\mathrm{i}}^*
        &=
        \E_P\big\{\big(\mathrm{e}_1^\top\Omega^{-1}\iota_h(X-x_0)\big)e_0^*(X)\rho'(\varepsilon\given X)\big\}
        +
        \supOP\bigg(\frac{h^{\beta^*}}{\sqrt{nh^d}}\bigg)
        \\
        &=
        \frac{1}{p_0\zeta_{\rho'}(x_0)}\mathrm{e}_1^\top s_1(K)^{-1}\E_P\big\{\iota_h(X-x_0) e_0^*(X)\rho'(\varepsilon\given X)\big\}
        +
        \supOP\bigg(\frac{h^{\beta^*}}{\sqrt{nh^d}}\bigg)
        \\
        &=
        \frac{1}{p_0}\E_P\bigl\{K_h(X-x_0) e_0^*(X)\bigr\}
        +
        \op\big(h^{\beta^*}\big)
        ,
    \end{align*}
    where the final line follows by uniform continuity of~$\zeta_{\rho'}$.

    \medskip\noindent{\bf Term $\boldsymbol{\phi_{\RN{2}.\mathrm{ii}.k}^*}$:} 
    Recalling the definition of~$\cE_{\hat\rho,k,\mathrm{abs}}$ from~\eqref{eq:cE-rho'-abs}, we have
    \begin{align*}
\E_P\bigl(|\hat\phi_{\RN{2}.\mathrm{ii}.k}&|\biggiven\hat\rho_k\bigr)
        \leq
        \E_P\Big\{\big|\mathrm{e}_1^\top\Omega^{-1}\iota_h(X-x_0)\big|\cdot|e_0^*(X)|\cdot \bigl|(\hat\rho_k'-\rho')(\varepsilon\given X)\bigr|\Biggiven\hat\rho_k\Big\}
        \\
        &\leq
        \frac{C_{e^*}(L,d,\beta^*)h^{\beta^*}}{p_0\zeta_{\rho'}(x_0)}\E_P\{|K_h(X-x_0)|\cE_{\hat\rho,k,\text{sq}}^{1/2}(X)\}
        \\
        &\leq
        \frac{C_X R_1(K) C_{e^*}(L,d,\beta^*)h^{\beta^*}}{p_0\zeta_{\rho'}(x_0)}(1+\supoP(1))\sup_{x\in\cX}\E_P\big\{(\hat\rho_k'-\rho')^2(\varepsilon\given X)\biggiven\hat\rho_k,X=x\big\}^{1/2}
        \\
        &=\op\big(h^{\beta^*}\big),
    \end{align*}
    by Lemma~\ref{lem:e-bound} and Assumption~\ref{ass:DGM}, so $\hat{\phi}_{\RN{2}.\mathrm{ii}.k}=\op\big(h^{\beta^*}\big)$.

    \medskip\noindent{\bf Term $\boldsymbol{\phi_{\RN{2}.\mathrm{iii}.k}^*}$:} 
    As $\hat\rho_k'$ is uniformly continuous by Assumption~\ref{ass:DGM}, the function $\omega(\Delta):=\sup\{|\hat\rho_k'(x)-\hat\rho_k'(x')|:x,x'\in\R^d,\,\|x-x'\|\leq\Delta\}$ satisfies $\lim_{\Delta\searrow0}\omega(\Delta)=0$.  Hence
    \begin{align}
        |\hat\phi_{\RN{2}.\mathrm{iii}.k}|
        &=\bigg|\frac{1}{n}\sum_{i\in\cI_k}\big(\mathrm{e}_1^\top\Omega^{-1}\iota_h(X_i-x_0)\big)e_0^*(X_i)\big\{\hat\rho_k'(\varepsilon_i-\tau_i\varrho_i\given X_i)-\hat\rho_k'(\varepsilon_i\given X_i)\big\}\bigg|
        \notag
        \\
        &\leq
        \frac{1}{n}\sum_{i\in\cI_k}\big|\mathrm{e}_1^\top\Omega^{-1}\iota_h(X_i-x_0)\big|\cdot |e_0^*(X_i)| \cdot \bigl|\hat\rho_k'(\varepsilon_i-\tau_i\varrho_i\given X_i)-\hat\rho_k'(\varepsilon_i\given X_i)\bigr|
        \notag
        \\
        &\leq
        \frac{ C_{e^*}(L,d,\beta^*)h^{\beta^*}}{p_0\zeta_{\rho'}(x_0)}\cdot\frac{1}{n}\sum_{i\in\cI_k}|K_h(X_i-x_0)|\cdot\omega\bigl(|\tau_i\varrho_i|\bigr) 
        \notag
        \\
        &\leq
        \frac{ C_{e^*}(L,d,\beta^*)h^{\beta^*}}{p_0\zeta_{\rho'}(x_0)} \omega\Bigl(\max_{i:X_i\in\mathcal{B}_{x_0}(h)}|\varrho_i|\Bigr)\cdot\frac{1}{n}\sum_{i\in\cI_k}|K_h(X_i-x_0)|.
        \label{eq:modulus-continuity}
    \end{align}
    Now, $\frac{1}{n}\sum_{i\in\cI_k}|K_h(X_i-x_0)|=\supOP(1)$, and
    \begin{align}
        \max_{i:X_i\in\mathcal{B}_{x_0}(h)}|\varrho_i| 
        &=
        \max_{i:X_i\in\mathcal{B}_{x_0}(h)}\bigl|Q_h(X_i-x_0)^\top(\hat\theta-\theta_0)-e_0(X_i)\bigr|
        \notag
        \\
        &\leq
        e^{d/2} \|\hat\theta-\theta_0\|_2 + C_e(L,d,\beta^*,h)h^{1\vee\beta^*}
        =
        \op(1),
        \label{eq:modulus-continuity-op(1)}
    \end{align}
    by Cauchy--Schwarz and Lemmas~\ref{lem:e-bound},~\ref{lem:Q},~and~\ref{lem:gamma0-consistency}. It follows that $\hat\phi_{\RN{2}.\mathrm{iii}.k}=\supoP\big(h^{\beta^*}\big)$.

    \medskip
    We conclude from~\eqref{eq:de-staring} and~\eqref{eq:phi2-decomp} that
    \begin{equation*}
        \mathrm{e}_1^\top\hat\Omega(\boldsymbol{\tau})^{-1}\hat\phi_{\RN{2}}(\boldsymbol{\tau})=\frac{1}{p_0}\E_P\bigl\{K_h(X-x_0) e_0^*(X)\bigr\} + \op\big(h^{\beta^*}\big).
    \end{equation*}
    Finally, applying Lemma~\ref{lem:B(f)} gives the result.
    \end{proof}
    
    \begin{lemma}\label{lem:phi-3}
        Adopting the setup of Theorem~\ref{thm:decomp}, and with $\hat\phi_{\RN{3}}$ as defined in~\eqref{eq:phi-3},
        \begin{equation*}
            \mathrm{e}_1^\top\Omega^{-1}\hat\phi_{\RN{3}}(\boldsymbol{\tau}) = \op\bigg(\frac{1}{\sqrt{nh^d}}+h^{\beta^*}\bigg).
        \end{equation*}
    \end{lemma}
    \begin{proof}
    We decompose
    \begin{equation*}
        \hat\phi_{\RN{3}}(\boldsymbol{\tau}) = \sum_{k=1}^{\mathcal{K}}\underbrace{\frac{1}{n}\sum_{i\in\cI_k}\bigl\{\hat\rho_k'(\varepsilon_i-\tau_i\varrho_i\given X_i)-\hat\rho_k'(\varepsilon_i)\bigr\}\kappa_{h,\lambda}(X_i-x_0)\big(\hat{f}^{\mathrm{LP}}_k(X_i)-f(X_i)+\hat{c}_k(x_0)\big)\hat{\mu}_k(x_0)}_{=:\hat\phi_{\RN{3}.k}}.
    \end{equation*}

    \medskip\noindent{\bf Term $\boldsymbol{\hat\phi_{\RN{3}.k}}$:} 
        For $i\in\cI_k$,
        \begin{equation*}
            \varrho_i\kappa_{h,\lambda}(X_i) = \big(\hat{f}^{\mathrm{LP}}_k(X_i)-f(X_i)+\hat{c}_k(x_0)\big)\kappa_{h,\lambda}(X_i).
        \end{equation*}
        For the modulus of continuity $\omega$ of $\hat\rho_k'$ as in~\eqref{eq:modulus-continuity} it follows that
        \begin{align*}
            \|\hat\phi_{\RN{3}.k}\|_2 
            &\leq
            \omega\Bigl(\max_{i:X_i\in\mathcal{B}_{x_0}(\lambda h)}|\varrho_i|\Bigr)
            \bigg(\frac{1}{n}\sum_{i\in\cI_k}\big|\hat{f}^{\mathrm{LP}}_k(X_i)-f(X_i)
            +\hat{c}_k(x_0)\big|\cdot|\kappa_{h,\lambda}(X_i)|\bigg)\|\hat\mu_k(x_0)\|_2
            \\
            &\leq \omega\Bigl(\max_{i:X_i\in\mathcal{B}_{x_0}(\lambda h)}|\varrho_i|\Bigr)\bigg\{\bigg(\frac{1}{n}\sum_{i\in\cI_k}\big|\hat{f}^{\mathrm{LP}}_k(X_i)-f(X_i)\big|\cdot|\kappa_{h,\lambda}(X_i)|\bigg)
            \\
            &\hspace{6cm}
            +
            \bigg(\frac{1}{n}\sum_{i\in\cI_k}|\kappa_{h,\lambda}(X_i)|\bigg)|\hat{c}_k(x_0)|
            \bigg\}
            \|\hat\mu_k(x_0)\|_2.
        \end{align*}
        By Lemma~\ref{lem:mu(x0)},~\eqref{eq:mu(x_0)bound} and Lemma~\ref{lem:Q}, and then Lemma~\ref{lem:c(x0)}~and~\eqref{eq:c_k},
        \begin{equation*}
            \|\hat\mu_k(x_0)\|_2 = \supOP\big(R_1(K)\big),
            \qquad
            |\hat{c}_k(x_0)|=\supOP\bigg(\bigg\{\frac{1}{\sqrt{nh^d}}+h^{\beta^*}\bigg\}R_1(\kappa_\lambda)\bigg).
        \end{equation*}
        Further, $\frac{1}{n}\sum_{i\in\cI_k}|\kappa_{h,\lambda}(X_i)|=\supOP\big(R_1(\kappa_\lambda)\big)$ and for $n$ large enough that $\mathcal{X} + \mathcal{B}_0(\lambda b_n) \subseteq \mathcal{X}^\circ$,
        \begin{align*}
            \E_P\bigg\{\frac{1}{n}\sum_{i\in\cI_k}\big|\hat{f}^{\mathrm{LP}}_k(X_i)-f(X_i)\big|\cdot|\kappa_{h,\lambda}(X_i)|\bigg\}
            &\leq
            \E_P\Big\{\big|\hat{f}^{\mathrm{LP}}_k(X)-f(X)\big|\cdot|\kappa_{h,\lambda}(X)|\Big\}
            \\
            &\leq
            \E_P\{|\kappa_{h,\lambda}(X)|\}
            \sup_{x\in\cX+\mathcal{B}_0(\lambda b_n)}\E_P\Big\{\big|\hat{f}^{\mathrm{LP}}_k(x)-f(x)\big|\Big\}
            \\
            &\leq
            \E_P\{|\kappa_{h,\lambda}(X)|\}
            \bigg(\sup_{x\in\cX^\circ}\E_P\Big\{\bigl(\hat{f}^{\mathrm{LP}}_k(x)-f(x)\bigr)^2\Big\}\bigg)^{1/2}
            \\
            &=\supOP\bigg(\bigg\{\frac{1}{\sqrt{nh^d}}+h^{\beta^*}\bigg\}R_1(\kappa_\lambda)\bigg),
        \end{align*}
        where we applied Lemma~\ref{lem:LP-rate} in the final line. 
        Thus
        \[
        \frac{1}{n}\sum_{i\in\cI_k}\big|\hat{f}^{\mathrm{LP}}_k(X_i)-f(X_i)\big|\cdot|\kappa_{h,\lambda}(X_i)|= \supOP\bigg(\bigg\{\frac{1}{\sqrt{nh^d}}+h^{\beta^*}\bigg\}R_1(\kappa_\lambda)\bigg).
        \]
        By Assumption~\ref{ass:DGM}~\eqref{eq:modulus-continuity-op(1)} we also have $\omega\bigl(\max_{i:X_i\in\mathcal{B}_{x_0}(h)}|\varrho_i|\bigr)=\op(1)$, 
        so $\hat\phi_{\RN{3}.k}=\op\big(\frac{1}{\sqrt{nh^d}}+h^{\beta^*}\big)$, as required.
\end{proof}

    \begin{lemma}\label{lem:phi-4}
        Adopting the setup of Theorem~\ref{thm:decomp},
    \begin{equation*}
        \mathrm{e}_1^\top\Omega^{-1}\hat\phi_{\RN{4}} = \bigl(1+\op(1)\bigr)\cdot\frac{1}{n}\sum_{i=1}^n\kappa_{h,\lambda}(X_i-x_0)\varepsilon_i.
        \end{equation*}
    \end{lemma}
    \begin{proof}
        We decompose
        \begin{align*}
            \hat\phi_{\RN{4}} &= \sum_{k=1}^{\mathcal{K}}\underbrace{\frac{1}{n}\sum_{i\in\cI_k}\hat\rho_k'(\varepsilon_i\given X_i)\kappa_{h,\lambda}(X_i-x_0)\big(\hat{f}^{\mathrm{LP}}_k(X_i)-f(X_i)+c_k(x_0)\big)\hat\mu_k(x_0)}_{=:\hat\phi_{\RN{4}.\mathrm{i}.k}}
            \\
            &\qquad\qquad +
            \sum_{k=1}^{\mathcal{K}}\underbrace{\frac{1}{n}\sum_{i\in\cI_k}\kappa_{h,\lambda}(X_i-x_0)\big(\hat{c}_k(x_0)-c_k(x_0)\big)\big\{\hat\rho_k'(\varepsilon_i\given X_i)\hat\mu_k(x_0)-\rho'(\varepsilon_i\given X_i)\mu(x_0)\big\}}_{=:\hat\phi_{\RN{4}.\mathrm{ii}.k}}
            \\
            &\qquad\qquad+
            \sum_{k=1}^{\mathcal{K}}\underbrace{\frac{1}{n}\sum_{i\in\cI_k}\rho'(\varepsilon_i\given X_i)\kappa_{h,\lambda}(X_i-x_0)\big(\hat{c}_k(x_0)-c_k(x_0)\big)\mu(x_0)}_{=:\hat\phi_{\RN{4}.\mathrm{iii}.k}}.
        \end{align*}

        \medskip\noindent{\bf Term $\boldsymbol{\hat\phi_{\RN{4}.\mathrm{i}.k}}$:} 
        By construction of $c_k(x_0)$, for any $i\in\cI_k$, 
        \begin{equation*}
            \E_P\big\{\kappa_{h,\lambda}(X_i-x_0)\big(\hat{f}^{\mathrm{LP}}_k(X_i)-f(X_i)+c_k(x_0)\big)\biggiven (X_{i'},Y_{i'})_{i'\in\cI_k^c}\big\}
            =0.
        \end{equation*}
        Hence
        \begin{align*}
            \E_P&\bigg\{\bigg\|\frac{1}{n}\sum_{i\in\cI_k}\hat\rho_k'(\varepsilon_i\given X_i)\kappa_{h,\lambda}(X_i-x_0)\big(\hat{f}^{\mathrm{LP}}_k(X_i)-f(X_i)+c_k(x_0)\big)\hat\mu_k(x_0)\bigg\|_2^2\bigggiven(X_{i'},Y_{i'})_{i'\in\cI_k^c}\bigg\}
            \\
            &\leq \frac{1}{n}\E_P\Big\{\E\big(\bigl\{\hat\rho_k'(\varepsilon\given X)\bigr\}^2\biggiven \hat\rho_k, X\big)\kappa_{h,\lambda}^2(X-x_0)\big(\hat{f}^{\mathrm{LP}}_k(X)-f(X)+c_k(x_0)\big)^2\Biggiven(X_{i'},Y_{i'})_{i'\in\cI_k^c}\Big\}
            \|\hat\mu_k(x_0)\|_2^2
            \\
            &\leq
            \frac{1}{n}\bigl(2C_2+\op(1)\bigr)\E_P\Big\{\kappa_{h,\lambda}^2(X-x_0)\big(\hat{f}^{\mathrm{LP}}_k(X)-f(X)+c_k(x_0)\big)^2\Biggiven\hat{f}^{\mathrm{LP}}_k\Big\}
            \|\hat\mu_k(x_0)\|_2^2
            \\
            &\leq
            \frac{4C_2+\op(1)}{n}\biggl(\E_P\Big\{\kappa_{h,\lambda}^2(X - x_0)\big(\hat{f}^{\mathrm{LP}}_k(X) - f(X)\big)^2\Biggiven\hat{f}^{\mathrm{LP}}_k\Big\} \\
            &\hspace{8cm}+\E_P\bigl\{\kappa_{h,\lambda}^2(X - x_0)\bigr\}c_k^2(x_0)\biggr)
            \|\mu(x_0)\|_2^2,
        \end{align*}
        where the final line follows by Lemma~\ref{lem:mu(x0)}. 
        Further, for $n$ large enough that $\mathcal{X} + \mathcal{B}_0(\lambda b_n) \subseteq \mathcal{X}^\circ$,
        \begin{align*}
            \E_P\Big\{\kappa_{h,\lambda}^2(X-x_0)\big(\hat{f}^{\mathrm{LP}}_k(X)-f(X)&\big)^2\Big\} = \E_P\big\{\kappa_{h,\lambda}^2(X-x_0)\big\}\sup_{x\in\cX+\mathcal{B}_0(\lambda b_n)}\E_P\Big\{\big(\hat{f}^{\mathrm{LP}}_k(x)-f(x)\big)^2\Big\}
            \\
            &\leq
            \frac{1}{h^d}C_X R_2(\kappa_\lambda) \sup_{x\in\cX^\circ}\E_P\Big\{\big(\hat{f}^{\mathrm{LP}}_k(x)-f(x)\big)^2\Big\} =\op\bigg(\frac{R_2(\kappa_\lambda)}{h^d}\bigg),
        \end{align*}
        and so $\E_P\big\{\kappa_{h,\lambda}^2(X-x_0)(\hat{f}^{\mathrm{LP}}_k(X)-f(X))^2\biggiven\hat{f}^{\mathrm{LP}}_k\big\}=\op\big(h^{-d}R_2(\kappa_\lambda)\big)$.  Moreover, $\E_P\bigl\{\kappa_{h,\lambda}^2(X-x_0)\bigr\}=\supOP\bigl(h^{-d}R_2(\kappa_\lambda)\bigr)$, and  by~\eqref{eq:c_k} and Markov's inequality, $c_k(x_0)=\supOP\bigl(\frac{1}{\sqrt{nh^d}}+h^{\beta^*}\bigr)$. Therefore,
        \begin{align*}
            \E_P\bigg\{\bigg\|\frac{1}{n}\sum_{i\in\cI_k} & \hat\rho_k'(\varepsilon_i\given X_i)\kappa_{h,\lambda}(X_i-x_0)\big(\hat{f}^{\mathrm{LP}}_k(X_i)-f(X_i)+c_k(x_0)\big)\hat\mu_k(x_0)\bigg\|_2^2\bigggiven(X_i,Y_i)_{i\in\cI_k^c}\bigg\}
            \\
            &=
            \frac{1}{n}\biggl\{\op\bigg(\frac{R_2(\kappa_\lambda)}{h^d}\bigg)
            +
            \supOP\biggl(\frac{R_2(\kappa_\lambda)}{h^d}\biggr)\supOP\biggl(\frac{1}{nh^d}+h^{2\beta^*}\biggr)\biggr\}\supOP(1)
            \\
            &=
            \op\bigg(\frac{1}{nh^d}\bigg),
        \end{align*}
        and so $\hat\phi_{\RN{4}.\mathrm{i}.k}=\op\big(\frac{1}{\sqrt{nh^d}}\big)$.

        \medskip\noindent{\bf Term $\boldsymbol{\hat\phi_{\RN{4}.\mathrm{ii}.k}}$:} 
        Applying Lemmas~\ref{lem:mu(x0)}~and~\ref{lem:c(x0)},
        \begin{align*}
            \hat\phi_{\RN{4}.\mathrm{ii}.k} &= \bigg(\underbrace{\frac{1}{n}\sum_{i\in\cI_k}\rho'(\varepsilon_i\given X_i)\kappa_{h,\lambda}(X_i-x_0)}_{=\supOP(1)}\bigg)
            \underbrace{\big(\hat{c}_k(x_0)-c_k(x_0)\big)}_{=\supOP\big(\frac{1}{nh^d}R_2^{1/2}(\kappa_\lambda)\big)}\underbrace{\big(\hat\mu_k(x_0)-\mu(x_0)\big)}_{=\supOP\big(\frac{1}{\sqrt{nh^d}}\big)}
            \\
            &\qquad
            +
            \bigg(\underbrace{\frac{1}{n}\sum_{i\in\cI_k}\big\{\hat\rho_k'(\varepsilon_i\given X_i)-\rho'(\varepsilon_i\given X_i)\big\}\kappa_{h,\lambda}(X_i-x_0)}_{=\op(R_1(\kappa_\lambda))}\bigg)
            \underbrace{\big(\hat{c}_k(x_0)-c_k(x_0)\big)}_{=\supOP\big(\frac{1}{nh^d}R_2^{1/2}(\kappa_\lambda)\big)}\underbrace{\hat\mu_k(x_0)}_{=\supOP(1)}
            \\
            &=\op\bigg(\frac{1}{nh^d}R_2^{1/2}(\kappa_\lambda)\bigg).
        \end{align*}

        \medskip\noindent{\bf Term $\boldsymbol{\hat\phi_{\RN{4}.\mathrm{iii}.k}}$:} 
        We have, for $n$ large enough that $\mathcal{X} + \lambda b_n\mathcal{B}_0(1) \subseteq \mathcal{X}^\circ$,
        \begin{equation*}
            \Var_P\bigg(\frac{1}{n}\sum_{i\in\cI_k}\rho'(\varepsilon_i\given X_i)\kappa_{h,\lambda}(X_i-x_0)\bigg) 
            = \frac{|\cI_k|}{n^2}\,\E_P\big\{\zeta_{(\rho')^2}(X)\kappa_{h,\lambda}^2(X-x_0)\big\}
            \leq \frac{C_2C_X}{nh^d}R_2(\kappa_\lambda),
        \end{equation*}
        and
        \begin{equation*}
            \big|\E_P\bigl\{\zeta_{\rho'}(X)\kappa_{h,\lambda}(X-x_0)\bigr\}-p_0\zeta_{\rho'}(x_0)\big|
            =\supoP(1),
        \end{equation*}
        by dominated convergence, using the fact that $\E_P\bigl\{\bigl|\zeta_{\rho'}(X)\kappa_{h,\lambda}(X-x_0)\bigr|\bigr\} 
            \leq C_2^{1/2}C_XR_1(\kappa_{\lambda})$.
        Thus
        \begin{equation*}
            \frac{1}{n}\sum_{i\in\cI_k}\rho'(\varepsilon_i\given X_i)\kappa_{h,\lambda}(X_i-x_0)
            =
            \bigl(1+\op(1)\bigr)\frac{|\cI_k|}{n}p_0\zeta_{\rho'}(x_0).
        \end{equation*}
        Moreover, as $\mu(x_0)=u(K)\bigl(1+\op(1)\bigr)$,
        \begin{equation}\label{eq:phi-4-simplif}
            \hat\phi_{\RN{4}.\mathrm{iii}}
            :=             \sum_{k=1}^{\mathcal{K}}\hat\phi_{\RN{4}.\mathrm{iii}.k}
            =
            \bigl(1+\op(1)\bigr)u(K)\zeta_{\rho'}(x_0)p_0\sum_{k=1}^{\mathcal{K}}\frac{|\cI_k|}{n}\big(\hat{c}_k(x_0)-c_k(x_0)\big).
        \end{equation}
        We further decompose
        \begin{equation}\label{eq:c-exp}
            \hat{c}_k(x_0)-c_k(x_0) = {c}_k^{(1)}(x_0)
            +
            c_k^{(2)}(x_0)
            +
            \bigl\{{c}_k^{(3)}(x_0)-c_k(x_0)\bigr\} ,
        \end{equation}
        where
        \begin{align*}
            {c}_k^{(1)}(x_0)&:=\frac{1}{\frac{1}{|\cI_k|}\sum_{i\in\cI_k}\kappa_{h,\lambda}(X_i-x_0)}\bigg(\frac{1}{|\cI_k|}\sum_{i\in\cI_k}\kappa_{h,\lambda}(X_i-x_0)\big(Y_i-f(X_i)\big)\bigg),
            \\
            {c}_k^{(2)}(x_0) &:= 
            \biggl\{\frac{1}{\frac{1}{|\cI_k|}\sum_{i\in\cI_k}\kappa_{h,\lambda}(X_i\!-\!x_0)}
            \!-\!
            \frac{1}{\E_P\bigl(\kappa_{h,\lambda}(X\!-\!x_0)\bigr)}
            \biggr\}\biggl(\frac{1}{|\cI_k|}\sum_{i\in\cI_k}\kappa_{h,\lambda}(X_i\!-\!x_0)\bigl(f(X_i) \!-\!\hat{f}^{\mathrm{LP}}_k(X_i)\bigr)\biggr),
            \\
            {c}_k^{(3)}(x_0) &:= 
            \frac{1}{\E_P\bigl(\kappa_{h,\lambda}(X-x_0)\bigr)}\bigg(\frac{1}{|\cI_k|}\sum_{i\in\cI_k}\kappa_{h,\lambda}(X_i-x_0)\big(f(X_i)-\hat{f}^{\mathrm{LP}}_k(X_i)\big)\bigg)
            .
        \end{align*}
        Now,
        \begin{align}
            {c}_k^{(1)}(x_0)&=
            \bigg(\frac{1}{|\cI_k|}\sum_{i\in\cI_k}\kappa_{h,\lambda}(X_i-x_0)\bigg)^{-1}\bigg(\frac{1}{|\cI_k|}\sum_{i\in\cI_k}\kappa_{h,\lambda}(X_i-x_0)\varepsilon_i\bigg)
            \notag
            \\
            &= \bigl(1+\op(1)\bigr)\bigg(\frac{1}{|\cI_k|}\sum_{i\in\cI_k}\kappa_{h,\lambda}(X_i-x_0)\varepsilon_i\bigg).
            \label{eq:c-pt1}
        \end{align}
        Next, since
        \begin{equation*}
            \frac{1}{|\cI_k|}\sum_{i\in\cI_k}\kappa_{h,\lambda}(X_i-x_0) = \E_P\bigl(\kappa_{h,\lambda}(X-x_0)\bigr)
            +
            \supOP\biggl(\frac{R_2^{1/2}(\kappa_\lambda)}{\sqrt{nh^d}}\biggr),
        \end{equation*}
        and $\E_P\bigl(\kappa_{h,\lambda}(X-x_0)\bigr)\geq c_X$, we have that
        \begin{equation*}
            \frac{1}{\frac{1}{|\cI_k|}\sum_{i\in\cI_k}\kappa_{h,\lambda}(X_i-x_0)} - \frac{1}{\E_P\bigl(\kappa_{h,\lambda}(X-x_0)\bigr)}
            =
            \supOP\biggl(\frac{R_2^{1/2}(\kappa_\lambda)}{\sqrt{nh^d}}\biggr).
        \end{equation*}
        Moreover,
        \begin{align*}
            \E_P\biggl(\frac{1}{|\cI_k|}\biggl|\sum_{i\in\cI_k} &\kappa_{h,\lambda}(X_i-x_0)\big(f(X_i)-\hat{f}^{\mathrm{LP}}_k(X_i)\big)\biggr|\biggr)
            \\
            &\leq
            \E_P\bigl(\bigl|\kappa_{h,\lambda}(X-x_0)\bigl(f(X)-\hat{f}_k^{\mathrm{LP}}(X)\bigr)\bigr|\bigr)
            \\
            &\leq
            \E_P\bigl(|\kappa_{h,\lambda}(X-x_0)|\bigr)\bigl(1+\supoP(1)\bigr)\sup_{x\in\cX^\circ}\Bigl(\E_P\bigl\{f(x)-\hat{f}_k^{\mathrm{LP}}(x)\bigr)\bigr\}^2\Bigr)
            \\
            &= \supOP\biggl(\frac{1}{\sqrt{nh^d}}+h^{\beta^*}\biggr),
        \end{align*}
        by Lemma~\ref{lem:LP-rate}. 
        Thus
        \begin{equation}\label{eq:c-pt2}
            c_k^{(2)}(x_0) = \supoP\biggl(\biggl\{\frac{1}{\sqrt{nh^d}}+h^{\beta^*}\biggr\}R_2^{1/2}(\kappa_\lambda)\biggr).
        \end{equation}
        Finally, for the third term, $\E_P\bigl\{{c}_k^{(3)}(x_0)-c_k(x_0)\biggiven\hat{f}^{\mathrm{LP}}\bigr\}=0$, and
        \begin{align*}
            {c}_k^{(3)}(x_0)-c_k(x_0) &= \frac{\frac{1}{|\cI_k|}\sum_{i\in\cI_k}\kappa_{h,\lambda}(X_i-x_0)\big\{\hat{f}^{\mathrm{LP}}_k(X_i)-f(X_i)-\E_P\bigl(\hat{f}^{\mathrm{LP}}_k(X_i)-f(X_i)\given\hat{f}^{\mathrm{LP}}_k\bigr)\big\}}{\E_P\bigl(\kappa_{h,\lambda}(X-x_0)\bigr)}
            =:\frac{c_N}{c_D}.
        \end{align*}
        Now
        \begin{equation*}
            c_D=\bigl(1+\op(1)\bigr)p_0,
        \end{equation*}
        and
        \begin{align*}
            \E_P\big(c_N^2\big) &\leq
            \frac{1}{|\cI_k|}\E_P\Big\{\kappa_{h,\lambda}^2(X-x_0)\big(\hat{f}^{\mathrm{LP}}_k(X)-f(X)\big)^2\Big\}
            \\
            &\leq
            \frac{1}{|\cI_k|}\E_P\big(\kappa_{h,\lambda}^2(X-x_0)\big)
            \bigl(1+\supoP(1)\bigr)
            \,\sup_{x\in\cX^\circ}\E_P\Big(\big\{\hat{f}^{\mathrm{LP}}_k(x)-f(x)\big\}^2\Big)
            \\
            &=\supOP\bigg(\frac{R_2(\kappa_\lambda)}{nh^d}\bigg\{\frac{1}{nh^d}+h^{2\beta^*}\bigg\}\bigg)
        \end{align*}
        by Lemma~\ref{lem:LP-rate}.  Hence
        \begin{equation}\label{eq:c-pt3}
            {c}_k^{(3)}(x_0)-c_k(x_0)=\supoP\biggl(\biggl\{\frac{1}{\sqrt{nh^d}}+h^{\beta^*}\biggr\}R_2^{1/2}(\kappa_\lambda)\biggr).
        \end{equation}
        Combining~\eqref{eq:c-pt1},~\eqref{eq:c-pt2}, and~\eqref{eq:c-pt3}, we conclude that
        \begin{equation*}
            \hat\phi_{\RN{4}.\mathrm{iii}} = (1+\op(1))u(K)\zeta_{\rho'}(x_0)p_0\cdot\frac{1}{n}\sum_{i=1}^n\kappa_{h,\lambda}(X_i-x_0)\varepsilon_i,
        \end{equation*}
        as required. The result then follows as $\zeta_{\rho'}(x_0)p_0\,\mathrm{e}_1^\top\Omega^{-1}u(K)=\mathrm{e}_1^\top s_1(K)^{-1}u(K)=1$.
    \end{proof}

    \begin{lemma}\label{lem:omega2}
        Adopting the setup of Theorem~\ref{thm:decomp},
        \begin{equation*}
            \hat\Omega(\boldsymbol{\tau}) := \frac{1}{n}\sum_{k=1}^{\mathcal{K}}\sum_{i\in\cI_k}{\pi_h}(X_i-x_0)\hat\rho_k'(\varepsilon_i-\tau_i\varrho_i\given X_i) 
            = \Omega+\op(1),
        \end{equation*}
        where
        \begin{equation}\label{eq:Omega}
            \Omega := p_0\,s_1(K)\zeta_{\rho'}(x_0).
        \end{equation}
    \end{lemma}
    \begin{proof}
        We decompose
        \begin{equation}\label{eq:Omega-decomp}
            \hat\Omega(\boldsymbol{\tau})-\Omega = \sum_{k=1}^\mathcal{K}\xi_{\RN{1}.k}+\sum_{k=1}^\mathcal{K}\xi_{\RN{2}.k}+\xi_{\RN{3}}+\xi_{\RN{4}},
        \end{equation}
        where
        \begin{align*}
            \xi_{\RN{1}.k} &:= \frac{1}{n}\sum_{i\in\cI_k}{\pi_h}(X_i-x_0)\big\{\hat\rho_k'(\varepsilon_i-\tau_i\varrho_i\given X_i)-\hat\rho_k'(\varepsilon_i\given X_i)\big\},
            \\
            \xi_{\RN{2}.k} &:= \frac{1}{n}\sum_{i\in\cI_k}{\pi_h}(X_i-x_0)\big\{\hat\rho_k'(\varepsilon_i\given X_i)-\rho'(\varepsilon_i\given X_i)\big\},
            \\
            \xi_{\RN{3}} &:= \frac{1}{n}\sum_{i=1}^n\Big({\pi_h}(X_i-x_0)\rho'(\varepsilon_i\given X_i)-\E_P\big\{{\pi_h}(X-x_0)\rho'(\varepsilon\given X)\big\}\Big),
            \\
            \xi_{\RN{4}} &:= \E_P\big\{{\pi_h}(X-x_0)\rho'(\varepsilon\given X)\big\} - p_0\,s_1(K)\zeta_{\rho'}(x_0).
        \end{align*}

        \medskip\noindent{\bf Term $\boldsymbol{\xi_{\RN{1}.k}}$:} 
        As in~\eqref{eq:modulus-continuity}, and adopting the notation of the same modulus of continuity $\omega$,
        \begin{align*}
\|\xi_{\RN{1}.k}\|_{\mathrm{F}} = \tr^{1/2}\big(\xi_{\RN{1}.k}\xi_{\RN{1}.k}^\top\big)
            &\leq
            e^d \bigg(\frac{1}{n}\sum_{i\in\cI_k}|K_h(X_i-x_0)|\cdot\big|\hat\rho_k'(\varepsilon_i-\tau_i\varrho_i\given X_i)-\hat\rho_k'(\varepsilon_i\given X_i)\big|\bigg)
            \\
            &\leq
            e^d \,\omega\Bigl(\max_{i:X_i\in\innerr}|\varrho_i|\Bigr)\frac{1}{n}\sum_{i\in\cI_k}|K_h(X_i-x_0)|
            =\op(1).
        \end{align*}

        \medskip\noindent{\bf Term $\boldsymbol{\xi_{\RN{2}.k}}$:} 
        Since $\cE_{\hat\rho',k,\mathrm{sq}}$ is uniformly continuous, the modulus of continuity $\bar{\omega}(\Delta):=\sup\{|\cE_{\hat\rho',k,\mathrm{sq}}(x)-\cE_{\hat\rho',k,\mathrm{sq}}(x')|:x,x'\in\R^d, \|x-x'\|\leq\Delta\}$ satisfies $\lim_{\Delta\searrow0}\bar{\omega}(\Delta)=0$. Hence
        \begin{align*}
            \E_P\bigl(\tr(\xi_{\RN{2}.k}\xi_{\RN{2}.k}^\top)\biggiven (X_i)_{i\in\cI_k},\hat\rho_k\bigr) 
            &=
            \frac{1}{n^2}\sum_{i\in\cI_k}\sum_{j\in\cI_k}K_h(X_i-x_0)K_h(X_j-x_0)\bigl(Q_h(X_i-x_0)^\top Q_h(X_j-x_0)\bigr)^2
            \\
            &\hspace{2.5cm}\cdot
            \E_P\bigl\{(\hat\rho_k'-\rho')(\varepsilon_i\given X_i)(\hat\rho_k'-\rho')(\varepsilon_j\given X_j)\biggiven X_i,X_j,\hat\rho_k\bigr\}
            \\
            &\leq
            e^{2d}\biggl(\frac{1}{n^2}\sum_{i\in\cI_k}\sum_{j\in\cI_k}|K_h(X_i-x_0)K_h(X_j-x_0)| 
            \\
            &\hspace{2.5cm}\cdot
            \big|\E_P\bigl\{(\hat\rho_k'-\rho')(\varepsilon_i\given X_i)(\hat\rho_k'-\rho')(\varepsilon_j\given X_j)\biggiven X_i,X_j,\hat\rho_k\bigr\}\big|\bigg)
            \\
            &\leq
            e^{2d}\biggl(\frac{1}{n}\sum_{i\in\cI_k}|K_h(X_i-x_0)| \cdot \cE_{\hat\rho',k,\mathrm{sq}}^{1/2}(X_i) \biggr)^2
            \\
            &\leq
            2e^{2d}|\cE_{\hat\rho',k,\mathrm{sq}}(x_0)|\biggl(\frac{1}{n}\sum_{i\in\cI_k}|K_h(X_i-x_0)| \biggr)^2
            \\
            &\hspace{1cm}
            +
            2e^{2d}\biggl(\frac{1}{n}\sum_{i\in\cI_k}|K_h(X_i-x_0)|\cdot |\cE_{\hat\rho',k,\mathrm{sq}}(X_i)-\cE_{\hat\rho',k,\mathrm{sq}}(x_0)|^{1/2} \biggr)^2
            \\
            &\leq
            2e^{2d}\bigl(| \cE_{\hat\rho',k,\mathrm{sq}}(x_0)| + \bar{\omega}(h) \bigr)\biggl(\frac{1}{n}\sum_{i\in\cI_k}|K_h(X_i-x_0)| \biggr)^2 =\supoP(1),
        \end{align*}
        by Assumption~\ref{ass:DGM}. Therefore by e.g.~\citet[Lemma S6]{lundborg}, it follows that $\xi_{\RN{2}.k}\xi_{\RN{2}.k}^\top=\supoP(1)$.
        
        \medskip\noindent{\bf Term $\boldsymbol{\xi_{\RN{3}}}$:} 
        For $n$ large enough that $\mathcal{X} + b_n\mathcal{B}_0(1) \subseteq \mathcal{X}^\circ$, 
        \begin{align*}
            \E_P\bigl(\tr(\xi_{\RN{3}}\xi_{\RN{3}}^\top)\bigr)
            &\leq
            \frac{1}{n}\E_P\Big\{\tr\big({\pi_h}(X-x_0){\pi_h}(X-x_0)^\top\big) \bigl\{\rho'(\varepsilon\given X)\bigr\}^2\Big\}
            \\
            &=\frac{1}{n}\E_P\Big\{K_h^2(X-x_0)\|Q_h(X-x_0)\|_2^4\,\E_P\big(\bigl\{\rho'(\varepsilon\given X)\bigr\}^2\biggiven X\big)\Big\}
            \\
            &\leq
            \frac{1}{n} C_2e^{2d} \E_P\bigl\{K_h^2(X-x_0)\bigr\}
            \leq C_2C_Xe^{2d} R_2(K) \cdot \frac{1}{nh^d}.
        \end{align*}
        Therefore $\xi_{\RN{3}}=\supOP\big(\frac{1}{\sqrt{nh^d}}\big)=\op(1)$.

        \medskip\noindent{\bf Term $\boldsymbol{\xi_{\RN{4}}}$:} For $n$ large enough that $\mathcal{X} + b_n\mathcal{B}_0(1) \subseteq \mathcal{X}^\circ$, 
        \begin{align*}
            \xi_{\RN{4}}
            &=
            \int_{\mathcal{B}_0(1)} K(\nu)Q(\nu)Q(\nu)^\top \big\{\zeta_{\rho'}(x_0+h\nu)p_X(x_0+h\nu)-\zeta_{\rho'}(x_0)p_X(x_0)\big\} \, d\nu.
        \end{align*}
        Thus, by the continuity of $\zeta_{\rho'}$ and $p_X$ at $x_0$, and since
        \begin{equation*}
            \biggl\|\int_{\mathcal{B}_0(1)}K(\nu)Q(\nu)Q(\nu)^\top\zeta_{\rho'}(x_0+h\nu)p_X(x_0+h\nu) \, d\nu\biggr\|_{\mathrm{F}}^2
            \leq
            C_2^2C_X^2e^{2d} R_1^2(K),
        \end{equation*}
        it follows by the dominated convergence theorem that $\xi_{\RN{4}}=\op(1)$.

        We conclude that  $\hat\Omega(\boldsymbol{\tau})=\Omega+\op(1)$, as required.       
    \end{proof}

    \begin{lemma}\label{lem:B(f)}
    We adopt the setup of Theorem~\ref{thm:decomp}.  
    \begin{enumerate}[label=(\roman*)]
        \label{lem:bias1}
        \item We have
        \[
        \frac{1}{p_0}\E_P\bigl\{K_h(X-x_0) e_0^*(X)\bigr\}
        =
        h^{\beta^*}\,B(f,x_0,K,h)
        + \supoP\bigl(h^{\beta^*}\bigr),
        \]
        where
        \begin{equation*}
        B(f,x_0,K,h):=h^{\beta_0^*-\beta^*} \sum_{\substack{\alpha\in\N_0^d:\\\|\alpha\|_1=\beta_0^*}}\frac{1}{\alpha!}\int_{\mathcal{B}_{0}(1)} K(\nu)\nu^\alpha  g_{\alpha,\beta_0^*,x_0,h}(\nu) \,d\nu
        \end{equation*}
        and
        \[
        g_{\alpha,\beta_0^*,x_0,h}(\nu) := \beta_0^* \int_0^1(1-t)^{\beta_0^*-1}\bigl\{\partial^\alpha f(x_0+th\nu)-\partial^\alpha f(x_0)\bigr\} \,dt.
        \]
        \label{lem:bias2}
        \vspace{-1.5em}
        \item There exists $C_B(K,d,\beta_0^*,L) > 0$ such that 
        \[        \sup_{f\in\cH(\beta,L)} \sup_{x_0\in\cX}\sup_{h\in\cH_n}|B(f,x_0,K,h)| \leq C_B(K,d,\beta_0^*,L),
        \]
        and hence
        \[
        \sup_{f\in\cH(\beta,L)}\sup_{x_0\in\cX}\sup_{h\in\cH_n}\frac{1}{p_0}\E_P\bigl\{K_h(X-x_0) e_0^*(X)\bigr\} = O\big(h^{\beta^*}\big). \label{lem:bias3}
        \]
    \end{enumerate}
\end{lemma}
\begin{proof}
    \emph{(i)} Recalling the definition of $e_0^*$ from~\eqref{eq:e0*}, we have 
    \begin{align*}
         \frac{1}{p_0}\E_P\bigl\{K_h(X&-x_0) e_0^*(X)\bigr\}
        =
        \frac{h^{\beta_0^*}}{p_0}\sum_{\substack{\alpha\in\N_0^d: \\ \|\alpha\|_1=\beta_0^*}}\frac{1}{\alpha!}\int_{\mathcal{B}_{0}(1)} K(\nu)Q(\nu) \nu^\alpha \,
        g_{\alpha,\beta_0^*,x_0,h}(\nu) \,p_X(x_0+h\nu) \,d\nu.
    \end{align*}
    Now, for $n$ large enough that $\mathcal{X} + b_n\mathcal{B}_0(1) \subseteq \mathcal{X}^\circ$, 
    \begin{align*}
        \bigg|&\frac{h^{\beta_0^*}}{p_0}\sum_{\substack{\alpha\in\N_0^d:\\\|\alpha\|_1=\beta_0^*}}\frac{1}{\alpha!}\int_{\mathcal{B}_0(1)}K(\nu)\nu^\alpha g_{\alpha,\beta_0^*,x_0,h}(\nu)\bigl\{p_X(x_0+h\nu)-p_X(x_0)\bigr\} \,d\nu\bigg|
        \\
        &\leq
        \frac{h^{\beta_0^*}}{p_0}\sum_{\substack{\alpha\in\N_0^d:\\\|\alpha\|_1=\beta_0^*}}\frac{\beta_0^*}{\alpha!}\int_{\mathcal{B}_{0}(1)} |K(\nu)| \cdot \big|\nu^\alpha\big| \int_0^1 \bigl|\partial^\alpha f(x_0+th\nu)-\partial^\alpha f(x_0)\bigr| \, dt
        \cdot\bigl|p_X(x_0+h\nu)-p_X(x_0)\bigr| \,d\nu 
        \\
        &\leq
        \frac{L_X L h^{\beta^*+(\beta_X\wedge 1)}}{p_0} \sum_{\substack{\alpha\in\N_0^d: \\ \|\alpha\|_1=\beta_0^*}}\frac{\beta_0^*}{\alpha!} \int_{\mathcal{B}_{0}(1)} |K(\nu)|\|\nu\|^{\beta^*-\beta_0^*+(\beta_X\wedge1)} \, d\nu 
        \\
        &\leq
        \frac{L_X L d^{\beta_0^*}R_1(K)}{c_X(\beta_0^*-1)!} \, h^{\beta^*+(\beta_X\wedge 1)} 
        =\supoP\big(h^{\beta^*}\big),
    \end{align*}
    which establishes~\emph{(i)}. 
    
    \medskip
    \noindent\emph{(ii)} 
    By identical arguments to those employed in~\emph{(i)},
    \begin{align}
\sup_{f\in\cH(\beta,L)}\sup_{x_0\in\cX}\sup_{h\in\cH_n}\Bigg|h^{\beta_0^*-\beta^*} \sum_{\substack{\alpha\in\N_0^d: \\ \|\alpha\|_1=\beta_0^*}}\frac{1}{\alpha!}\int_{\mathcal{B}_{0}(1)} &K(\nu) \nu^\alpha g_{\alpha,\beta_0^*,x_0,h}(\nu) \,d\nu\Bigg|
\notag
        \\
        &\leq 
        \frac{Ld^{\beta_0^*}R_1(K)}{(\beta_0^*-1)!}
        =: C_B(K,d,\beta_0^*,L),
        \label{eq:CB}
    \end{align}
    so~\emph{(ii)} follows.
    \end{proof}

\begin{lemma}\label{lem:Q}
    Let $B_q:=\{\nu\in\R^d:\|\nu\|_q\leq1\}$.  For $\nu \in B_q$, 
    \[
    \|Q(\nu)\|_1 \vee \|Q(\nu)\|_2^2 \leq e^d.
    \] 
\end{lemma}
\begin{proof}
    We have $|\nu_r|\leq \|\nu\|_q^q\leq 1$ for every $r\in[d]$, and recall that the multinomial theorem states that
    \begin{equation}\label{eq:multinom-expansion}
    \biggl(\sum_{r=1}^d x_r\biggr)^\ell = \sum_{\alpha\in\N_0^d:\|\alpha\|_1 = \ell} \frac{\ell!}{\alpha!} x^\alpha
    \end{equation}
    for $x = (x_1,\ldots,x_d)^\top \in \mathbb{R}^d$ and $\ell \in \mathbb{N}_0$.  Thus, by H\"older's inequality,
    \begin{align*}
        \|Q(\nu)\|_1 
        =
        \sum_{\alpha\in\N_0^d:\|\alpha\|_1\leq p}\frac{1}{\alpha!}\prod_{r=1}^d|\nu_r|^{\alpha_r}
        \leq
        \sum_{\ell=0}^p \sum_{\alpha\in\N_0^d:\|\alpha\|_1 = \ell}\frac{1}{\alpha!}
        =
        \sum_{\ell=0}^{ p}\frac{d^\ell}{\ell!}
        \leq
        e^d.
    \end{align*}
    Moreover, 
    \begin{align*}
        \|Q(\nu)\|_2^2
        =
        \sum_{\alpha\in\N_0^d:\|\alpha\|_1\leq p}\frac{1}{(\alpha!)^2}\prod_{r=1}^d|\nu_r|^{2\alpha_r}
        \leq
        \sum_{\alpha\in\N_0^d:\|\alpha\|_1\leq p}\frac{1}{\alpha!}
        \leq
        e^d,
    \end{align*}
    as required.
\end{proof}
 
    \begin{lemma}\label{lem:e-bound}
        For $x\in\R^d$ with $\|x-x_0\|\leq h$, we have
        \begin{align*}
            &|e_0^*(x)|\leq \frac{Ld^{\beta_0^*}}{\beta_0^*!}h^{\beta^*} =: C_{e^*}(L,d,\beta^*)h^{\beta^*},
            \\
            &|e_0(x)|\leq L(e^{hd}-1)+\frac{Ld^{\beta_0^*}}{\beta_0^*!}h^{\beta^*}
             =: C_e(L,d,\beta^*,h)h^{1\vee\beta^*}.
        \end{align*}
    \end{lemma}
\begin{proof}
        By analogous arguments to Lemma~\ref{lem:Q}, 
        \begin{equation*}
            |e_0^*(x)| \leq
            L\sum_{\substack{\alpha\in\N_0^d: \\ \|\alpha\|_1=\beta_0^*}}\frac{\beta_0^*}{\alpha!}\|x-x_0\|^{\beta^*-\beta_0^*} |(x-x_0)^\alpha|\int_0^1(1-t)^{\beta_0^*-1}\,dt
            \leq Lh^{\beta^*}\sum_{\substack{\alpha\in\N_0^d: \\ \|\alpha\|_1=\beta_0^*}}\frac{1}{\alpha!}
            = \frac{Ld^{\beta_0^*}}{\beta_0^*!}h^{\beta^*}.
        \end{equation*} 
        Further, 
        \begin{align*}
            |e_0(x)| \leq \sum_{\substack{\alpha\in\N_0^d: \\ 1\leq\|\alpha\|_1\leq\beta_0^*}} \! \! \frac{1}{\alpha!}\,\bigl|\partial^\alpha f(x_0)(x-x_0)^\alpha\bigr| + |e_0^*(x)|
            &\leq
            L\!\sum_{\substack{\alpha\in\N_0^d: \\ \|\alpha\|_1\geq1}}\frac{h^{\|\alpha\|_1}}{\alpha!} + |e_0^*(x)|
            = 
            L(e^{hd}-1) + |e_0^*(x)|,
        \end{align*}
        as required.
\end{proof}

\begin{lemma}\label{lem:clt}
    Recall the triangular array $Z_{n,1},\ldots,Z_{n,n}$ defined in the proof of Theorem~\ref{thm:decomp} (see~\eqref{eq:Zi}), define $\sigma_P^2 \equiv \sigma_{P,n}^2 := \mathrm{Var}_P(Z_{n,1})$.  Then
    \begin{equation}\label{eq:lyapunov}
\sup_{P\in\cP}\sup_{x_0\in\cX}\sup_{h\in\cH_n}\sup_{t\in\R}\bigg| \Pr_P\biggl(\frac{1}{n^{1/2}\sigma_P}\sum_{i=1}^nZ_{n,i} \leq t\biggr) - \Phi(t) \bigg| \rightarrow 0
    \end{equation}
as $n \rightarrow \infty$.
\end{lemma}

\begin{proof}
    For $x_0 \in \mathcal{X}$, define 
    \begin{equation*}
        a:=\frac{1}{p_0\zeta_{\rho'}(x_0)}\,\mathrm{e}_1 \in \mathbb{R}^{\bar{p}},
        \qquad
        b := 
        \frac{1}{p_0} \in \mathbb{R}
    \end{equation*}
    so that 
    \begin{align}
        Z_{n,i} &:= a^\top \varphi(X_i)\rho(\varepsilon_i\given X_i)+b\,\kappa_{h,\lambda}(X_i-x_0)\varepsilon_i \notag
        \\
        &\phantom{:}=a^\top\iota_h(X_i-x_0)\rho(\varepsilon_i\given X_i)
        +\big\{-a^\top\mu(x_0)\rho(\varepsilon_i\given X_i)+b\varepsilon_i\big\}\kappa_{h,\lambda}(X_i-x_0)
        \label{eq:Z_ni}
    \end{align}
    for $i \in [n]$.  Note $\E_P(Z_{n,i})=0$, and since $\iota_h(\cdot)\kappa_{h,\lambda}(\cdot)=0$,
    \begin{align*}
        h^d\sigma_P^2 &=
        h^d\E_P\big\{\bigl(a^\top\iota_h(X-x_0)\bigr)^2\rho^2(\varepsilon\given X)\big\}
        +
        h^d\E_P\big\{\big(a^\top\mu(x_0)\rho(\varepsilon\given X)-b\varepsilon\big)^2\kappa_{h,\lambda}^2(X-x_0)\big\}
        \\
        &=
        \big(a^\top s_2(K) a\big)\,p_0\zeta_{\rho^2}(x_0)
        +
        \E_P\big\{\big(a^\top\mu(x_0)\rho(\varepsilon\given X)-b\varepsilon\big)^2\biggiven X=x_0\big\}p_0R_2(\kappa_\lambda)
        + \op(1)
        \\
        &=
        \frac{R_2(K)}{p_0}\cdot\frac{\zeta_{\rho^2}(x_0)}{\zeta_{\rho'}^2(x_0)}
        +
        \frac{R_2(\kappa_\lambda)}{p_0}\,\E_P\bigg\{\bigg(\frac{\rho(\varepsilon\given X)}{\zeta_{\rho'}(x_0)}-\varepsilon\bigg)^2\bigggiven X=x_0 \bigg\}
        +\op(1),
    \end{align*}
    where we have used the fact that $\zeta_{\rho'}(x_0)\big(a^\top\mu(x_0)\big)=b+\op(1)$. 
    Thus
    \begin{equation}\label{eq:VarZi}
        \sigma_P^2=\frac{1}{p_0h^d}\bigg\{R_2(K)V_P(\rho)+
        \E_P\bigg\{
        \bigg(\frac{\rho(\varepsilon\given X)}{\zeta_{\rho'}(x_0)}-\varepsilon\bigg)^2\bigggiven X=x_0\bigg\}
        R_2(\kappa_\lambda)
        \bigg\}
        + \op\bigg(\frac{1}{h^d}\bigg).
    \end{equation}
    By Lemma~\ref{lem:unif-norm}, it suffices to verify the uniform Lyapunov condition
    \begin{equation}\label{eq:lyapunov-cond}
        \frac{\E_P\big(|Z_{n,1}|^{2+\delta}\big)}{n^{\delta/2}\sigma_P^{2+\delta}}
        = \op(1).
    \end{equation}
    Note first that
    \begin{equation*}
        \inf_{x_0\in\cX}\inf_{P\in\cP}\frac{R_2(K)V_P(\rho)}{p_0}\geq
        \frac{c_1R_2(K)}{C_2C_X},
    \end{equation*}
    and so
    \begin{equation}\label{eq:lyapunov-var-LB}
        \frac{1}{\sigma_P^2}
        =\supOP(h^d).
    \end{equation}
    For the $(2+\delta)$th absolute moment upper bound, first note the inequalities $|a^\top\iota_h(\cdot)|\leq K(\cdot)/(c_Xc_1)$ and $|b|\leq 1/c_X$. Thus by Lemma~\ref{lem:Q},
    \begin{align}
        \E_P\big(&|Z_{n,1}|^{2+\delta}\big)
        \notag
        \\
        &\leq
        2^{1+\delta}\E_P\Big\{|a^\top\iota_h(X-x_0)|^{2+\delta}|\rho(\varepsilon\given X)|^{2+\delta}
        + \big|a^\top\mu(x_0)\rho(\varepsilon\given X)-b\varepsilon\big|^{2+\delta}|\kappa_{h,\lambda}(X-x_0)|^{2+\delta}
        \Big\}
        \notag
        \\
        &\leq 2^{1+\delta}\E_P\Big\{|a^\top\iota_h(X-x_0)|^{2+\delta} 
        |\rho(\varepsilon \given X)|^{2+\delta}
        \notag
        \\
        &\qquad\qquad
        +2^{1+\delta}\big(\|a\|_2^{2+\delta}\|\mu(x_0)\|_2^{2+\delta}|\rho(\varepsilon\given X)|^{2+\delta}+|b|^{2+\delta}|\varepsilon|^{2+\delta}\big)|\kappa_{h,\lambda}(X-x_0)|^{2+\delta}
        \Big\}
        \notag
        \\
        &\leq
        2^{1+\delta}\biggl[\biggl(\frac{1}{c_Xc_1}\biggr)^{2+\delta}C_XC_1R_{2+\delta}(K)
        +2^{1+\delta}\biggl\{\biggl(\frac{e^{d/2}}{c_Xc_1\Lambda_{\min}\bigl(s_1(K)\bigr)}\biggr)^{2+\delta}C_XC_1\|u(K)\|_2^{2+\delta}
        \notag
        \\
        &\qquad\qquad
        +
         \frac{C_3|\mathrm{e}_1^\top s_1(K)^{-1}u(K)|^{2+\delta}}{c_X^{2+\delta}}
        \biggr\}
        R_{2+\delta}(\kappa_\lambda) + \supoP(1)
        \biggr]\cdot\frac{1}{h^{d(1+\delta)}}
        \notag
        \\
        &=\supOP\bigg(\frac{1}{h^{d(1+\delta)}}\bigg).
        \label{eq:lyapunov-2+d-UB}
    \end{align}
    Combining~\eqref{eq:lyapunov-var-LB}~and~\eqref{eq:lyapunov-2+d-UB}, the uniform Lyapunov condition~\eqref{eq:lyapunov-cond} holds, and this completes the proof.
\end{proof}

\begin{lemma}\label{lem:unif-norm}
    Let $(\psi_{P,n,i})_{P\in\cP,n\in\N,i\in[n]}$ be a triangular array of real-valued random functions  $\psi_{P,n,i}:\cX\times\cH_n\to\R$ such that for any $P\in\cP$ and $n\in\N$, we have that $(\psi_{P,n,i})_{i\in[n]}$ are independent,  $\E_P\bigl(\psi_{P,n,i}(x_0,h)\bigr)=0$ for all $i\in[n]$ and $(P,x_0,h)\in\cP\times\cX\times\cH_n$ and there exists $\delta > 0$ such that the uniform Lyapunov condition
    \begin{equation*}
\sup_{P\in\cP}\sup_{x_0\in\cX}\sup_{h\in\cH_n}\frac{\E_P\big(|\psi_P(x_0,h)|^{2+\delta}\big)}{n^{\delta/2}\Var_P^{1+\delta/2}\psi_P(x_0,h)} \rightarrow 0
    \end{equation*}
    holds. For each $n\in\N$, let $S_{P,n}:=\frac{1}{\sqrt{n}}\sum_{i=1}^n\psi_{P,n,i}$. Then
    \begin{equation*}
\sup_{P\in\cP}\sup_{x_0\in\cX}\sup_{h\in\cH_n}\sup_{t\in\R}\biggl|\Pr_P\biggl(\frac{S_{P,n}(x_0,h)}{\Var_P^{1/2} S_{P,n}(x_0,h)}\leq t\biggr)-\Phi(t)\biggr| \rightarrow 0.
    \end{equation*}
\end{lemma}
\begin{proof}
This follows by identical arguments to~\citet[Lemma~17]{sandwich-boosting} 
    alongside a uniform version of Slutsky's lemma~\citep[e.g.][Lemma 20]{shahpeters}.
\end{proof}

\begin{lemma}\label{lem:gamma0-consistency}
    Adopting the setup of Theorem~\ref{thm:decomp}, the measurable sequence $\hat{\theta} = (\hat{\theta}_n)$ from Algorithm~\ref{alg:outrigger} satisfies
    \[
    \hat\theta-\theta_0=\op(1).
    \]
\end{lemma}

\begin{proof}
Define $\Psi:\mathbb{R}^{\bar{p}} \rightarrow \mathbb{R}^{\bar{p}}$ and $\hat\Psi_n:\mathbb{R}^{\bar{p}} \rightarrow \mathbb{R}^{\bar{p}}$ by
    \begin{equation*}
        \Psi(\theta) := p_X(x_0)\int_{\mathcal{B}_0(1)} K(\nu)Q(\nu)\,\E_P\bigl\{\rho\bigl(\varepsilon+Q(\nu)^\top(\theta_0-\theta)\given X\bigr)\biggiven X=x_0\bigr\}\,d\nu,
    \end{equation*}
    and
    \begin{multline*}
        \hat\Psi_n({\theta}) 
        :=
        \frac{1}{n}\sum_{k=1}^{\mathcal{K}}\sum_{i\in\cI_k}\hat\varphi_k(X_i)\hat\rho_k\Bigl(\varepsilon_i+\bigl(e_0(X_i)+Q_h(X_i-x_0)^\top(\theta_0-\theta)\bigr)\ind_{\innerr}(X_i) 
        \\
        - \bigl(\hat{f}_k^{\mathrm{LP}}(X_i)-f(X_i)+\hat{c}_k(x_0)\bigr)\ind_{\outerr}(X_i) \Biggiven X_i\Bigr).
    \end{multline*}
Taking $\eta > 0$ from Assumption~\ref{ass:DGM}, let $\hat\theta^*$ denote the unique root of $\hat{\Psi}_n$ in $\{\theta\in\R^{\bar{p}}:\|\theta-\theta_0\|_\infty \leq \eta\}$. 
Since $s_1(K)$ is invertible, there exists an orthonormal basis $(t_m)_{m\in[\bar{p}]}$ for $\mathbb{R}^{\bar{p}}$ with $t_m^\top s_1(K)t_m\neq0$ for all $m\in[\bar{p}]$. 
For $m\in[\bar{p}]$, define $\psi_m:[0,\infty) \rightarrow \mathbb{R}$ by
    \begin{equation*}
\psi_m(\lambda):=t_m^\top\Psi(\theta_0+\lambda t_m).
    \end{equation*}
    Then $\psi_m(0)=0$ and, using Assumption~\ref{ass:exp-bounds}, we may differentiate under the integral sign to obtain
    \begin{equation*}
        |\psi_m'(0)| = p_X(x_0)\bigl|t_m^\top s_1(K)t_m\E_P\bigl(\rho'(\varepsilon\given X)\given X=x_0\bigr)\bigr| \geq c_X c_1 |t_m^\top s_1(K) t_m| > 0.
    \end{equation*}
    Therefore, there exists $\lambda^*>0$ such that for all $\lambda\in(0,\lambda^*]$ and $m\in[\bar{p}]$,
    \begin{equation*}
        \bigl(t_m^\top\Psi(\theta_0+\lambda t_m)\bigr)
        \bigl(t_m^\top\Psi(\theta_0-\lambda t_m)\bigr) < 0.
    \end{equation*} 
    Take arbitrary $\epsilon,\eta'\in(0,\eta)$. By the Poincaré--Miranda theorem~\citep[Theorem 17.1.1]{miranda}, if $\|\hat\theta^*-\theta_0\|_\infty>\epsilon$ then there exists $m\in[\bar{p}]$ such that
    \begin{equation*}
        \bigl(t_m^\top\hat\Psi_n(\theta_0+\epsilon t_m)\bigr)\bigl( t_m^\top\hat\Psi_n(\theta_0-\epsilon t_m)\bigr)>0.
    \end{equation*}
    Taking a union bound over these events, for any $\epsilon, \eta_1 >0$, and defining $\epsilon':=\epsilon\wedge\lambda^*\wedge \eta'$,
    \begin{align*}
        \Pr_P&(\|\hat\theta^* -\theta_0\|_\infty>\epsilon)
        \leq \Pr_P(\|\hat\theta^*-\theta_0\|_\infty>\epsilon')
        \\
        &\leq
        \sum_{m=1}^{\bar{p}}\Bigl\{
        \Pr_P\bigl( t_m^\top\hat\Psi_n(\theta_0+\epsilon'\, t_m)>0,\; t_m^\top\hat\Psi_n(\theta_0-\epsilon'\, t_m)>0\bigr)
        \\
        &\qquad
        +
        \Pr_P\bigl( t_m^\top\hat\Psi_n(\theta_0+\epsilon'\, t_m)<0,\; t_m^\top\hat\Psi_n(\theta_0-\epsilon'\, t_m)<0\bigr)
        \Bigr\}
        \\
        &\leq 
        \sum_{m=1}^{\bar{p}}\Bigl\{
            \Pr_P\bigl(\bigl| t_m^\top\bigl\{\hat\Psi_n(\theta_0+\epsilon'\, t_m)-\Psi(\theta_0+\epsilon'\, t_m)\bigr\}\bigr|>\eta_1\bigr) +
            \ind_{\{t_m^\top\Psi(\theta_0+\epsilon'\, t_m)>-\eta_1\}} \ind_{\{t_m^\top\Psi(\theta_0-\epsilon'\, t_m)>-\eta_1\}}
            \\
            &\qquad
            +
            \Pr_P\bigl(\bigl| t_m^\top\bigl\{\hat\Psi_n(\theta_0-\epsilon'\, t_m)-\Psi(\theta_0-\epsilon'\, t_m)\bigr\}\bigr|>\eta_1\bigr)
            +
            \ind_{\{t_m^\top\Psi(\theta_0+\epsilon'\, t_m) <\eta_1\}} \ind_{\{t_m^\top\Psi(\theta_0-\epsilon'\, t_m)<\eta_1\}}\Bigr\}
            .
    \end{align*}
    Similarly to the proof of Theorem~\ref{thm:decomp}, for $\theta \in \mathbb{R}^{\bar{p}}$ with $\|\theta-\theta_0\|_\infty\leq \eta'$,
    \begin{align*}
        \hat\Psi_n(\theta)
        &=\frac{1}{n}\sum_{i=1}^n \varphi(X_i)\rho\bigl(\varepsilon_i+\bigl\{e_0(X_i)+Q_h(X_i-x_0)^\top(\theta_0-\theta)\bigr\}\ind_{\innerr}(X_i)\biggiven X_i\bigr) + \op(1)
        \\
        &=\frac{1}{n}\sum_{i=1}^n \varphi(X_i)\rho\bigl(\varepsilon_i+Q_h(X_i-x_0)^\top(\theta_0-\theta)\ind_{\innerr}(X_i)\biggiven X_i\bigr) + \op(1)
        \\
        &=
        \E_P\bigl\{K_h(X-x_0)\rho\bigl(\varepsilon+Q_h(X-x_0)^\top(\theta_0-\theta)\biggiven X\bigr)\biggiven X=x_0\bigr\} + \op(1)
        \\
        &=
        \Psi(\theta) + \op(1).
    \end{align*}
Hence the unique root $\hat\theta^*$ of $\hat{\Psi}_n$ in $\{\theta \in \R^{\bar{p}}:\|\theta -\theta_0\|_\infty \leq \eta_2\}$ satisfies $\hat\theta^*=\theta_0+\supoP(1)$. 
    It remains to show that $\hat\theta=\hat\theta^*$ with probability at least $1-\supoP(1)$. 
    Recall from Algorithm~\ref{alg:outrigger} that $\hat\theta^{\mathrm{LP}}$ is the unique vector of coefficients that solves the local polynomial estimating equations, and let $\hat\theta^c$ be a root on the outer region~$\{\theta\in\R^{\bar{p}}:\|\theta-\theta_0\|_\infty > \eta_2\}$ if at least one exists. Then 
    \begin{equation*}
        \|\hat\theta^c-\hat\theta^{\mathrm{LP}}\|_\infty
        \geq
        \|\hat\theta^c-\theta_0\|_\infty 
        -
        \|\hat\theta^{\mathrm{LP}}-\theta_0\|_\infty
        > \eta_2 - \supoP(1).
    \end{equation*}
    On the other hand, by assumption
    \begin{equation*}
        \|\hat\theta-\hat\theta^{\mathrm{LP}}\|_\infty \leq \eta_1,
    \end{equation*}
 so on a sequence of events of probability at least $1-\supoP(1)$, we have $\hat\theta^*=\hat\theta$, as required.
\end{proof}

\begin{lemma}\label{lem:LP-rate}
        Under Assumptions~\ref{ass:DGM} and~\ref{ass:bandwidth-kernels-and-co}, define a standard local polynomial estimator $\hat{f}^{\mathrm{LP}}$ by 
        \[
        \hat{f}^{\mathrm{LP}}(x):=\biggl(\sum_{i=1}^nK_h(X_i-x)Q_h(X_i-x)Q_h(X_i-x)^\top\biggr)^{-1}\sum_{i=1}^nK_h(X_i-x)Q_h(X_i-x)Y_i.
        \]
        Then 
        \begin{equation*}
            \E_P\Bigl\{\bigl(\hat{f}^{\mathrm{LP}}(x_0)-f(x_0)\bigr)^2\Bigr\}
            = \supOP\biggl(h^{2\beta^*}+\frac{1}{nh^d}\biggr).
        \end{equation*}
    \end{lemma}
    \begin{proof}
        Algorithm~\ref{alg:outrigger} can trivially be adapted to construct the simpler standard local polynomial estimator $\hat{f}^{\mathrm{LP}}$ by removing the for loop over $k\in[\mathcal{K}]$ in entirety, and instead taking $\hat{f}^{\mathrm{LP}}_k=\hat\mu_k=\hat{c}_k\equiv0$ for all $k\in[\mathcal{K}]$ and also taking $\kappa_\lambda\equiv0$ and $\hat\rho_k=\rho=\mathrm{id}$ for all $k\in[\mathcal{K}]$. In the decomposition~\eqref{eq:first-expansion},     $\mathrm{e}_1^\top\Omega^{-1}\hat\phi_{\RN{3}}(\boldsymbol{\tau})
        =
        {e}_1^\top\Omega^{-1}\hat\phi_{\RN{4}}(\boldsymbol{\tau})
        =0$. Further, in the proof of Lemma~\ref{lem:phi-1},  $\hat\phi_{\RN{1}.\mathrm{i}.k}=\hat\phi_{\RN{1}.\mathrm{ii}.k}=\hat\phi_{\RN{1}.\mathrm{iii}.k}=0$, and in the proof of Lemma~\ref{lem:usual-bias},  $\hat\phi_{\RN{2}.\mathrm{iii}.k}=0$ for all $k\in[\mathcal{K}]$. The remaining three terms $\phi_{\RN{1}}$, $\phi_{\RN{2}.\mathrm{i}}^*$, and $\hat\phi_{\RN{2}.\mathrm{ii}.k}$ are dealt with as in Lemmas~\ref{lem:phi-1}~and~\ref{lem:usual-bias}.
    \end{proof}

\begin{proof}[Proof of Theorem~\ref{thm:decomp-ext}]
    Theorem~\ref{thm:decomp-ext} follows by identical arguments to the proof of Theorem~\ref{thm:decomp}, with the exception of two small amendments that we now describe. In~\eqref{eq:thmdecompext-1}, by similar arguments we have
    \begin{equation*}
        \E_P\bigg\{\bigg(\frac{1}{n}\sum_{i\in\cI_k}\kappa_{h,\lambda}(X_i -x_0)\varrho(\varepsilon_i\given X_i)\bigg)^2\bigg\}
        =
        \supoP\biggl(\frac{1}{nh^d}+h^{2\beta^*}\biggr),
    \end{equation*}
    and so writing $\hat\phi_{\RN{1}.\mathrm{i}.k}[\varrho]$ for the quantity $\hat\phi_{\RN{1}.\mathrm{i}.k}$ in~\eqref{eq:phi-1-decomp} with $\rho$ replaced with $\varrho$, we find that $\eta^\top\hat\phi_{\RN{1}.\mathrm{i}.k}[\varrho]=\supoP\bigl(\frac{1}{\sqrt{nh^d}}+h^{\beta^*}\bigr)$. Second, writing $Z_{n,1}[\varrho]$ for the quantity $Z_{n,1}$ in~\eqref{eq:Z_ni} with $\rho$ replaced with $\varrho$, by~\eqref{eq:lyapunov-2+d-UB} and~\eqref{eq:approx-soln},
    \begin{equation*}
        \frac{\E_P\big(|Z_{n,1}[\varrho]|^{2+\delta}\big)}{n^{\delta/2}\sigma_P^{2+\delta}}
        =
        \supOP\biggl(\frac{1}{n^{\delta/2}}\biggl\{1+h^{(2\beta+d)(1+\delta/2)}\biggr\}\biggr) = \supoP(1).
    \end{equation*}
    The result follows.
\end{proof}

\section{Existence of higher order kernels}

From a methodological standpoint, it would be common to consider local polynomial estimators (e.g.~locally linear polynomials) with a second order kernel $K$.  However, from a theoretical perspective, our results allow for estimation of functions of any H\"older smoothness $\beta\in(0,\infty)$. Assumption~\ref{ass:kernel}, which is used in Theorem~\ref{thm:decomp}, asks for a $d$-dimensional kernel $K$ of order $p+1$ for arbitrary $d\in\N, p\in\N_0$ with $s_1(K)$ invertible.  We establish the existence of such a kernel below.  It is convenient to let $\bar{p}^*:=\binom{d+2p}{2p}$ and define $Q^{\mathrm{ext}}:\R^d\to\R^{\bar{p}^*}$ by $Q^{\mathrm{ext}}(u):= \bigl(\frac{1}{\alpha!}\prod_{r=1}^d u_r^{\alpha_r}\bigr)_{\|\alpha\|_1\leq 2p}$, with components ordered in increasing graded lexicographic order in $\alpha$~i.e.~for $\alpha_1,\alpha_2\in\N_0^d$ we say $u^{\alpha_1}<_{\mathrm{gr.lex}}u^{\alpha_2}$ if $\|\alpha_1\|_1<\|\alpha_2\|_1$ or $\|\alpha_1\|_1=\|\alpha_2\|_1$ and $\alpha_1<_{\mathrm{lex}}\alpha_2$, where $<_{\mathrm{lex}}$ denotes the lexicographic order. It will also help to define the shorthand $\mathrm{C}^n_r:={\binom{n}{r}}=\frac{n!}{r!(n-r)!}$. Our proof will also rest on a construction involving a set of algebraically independent numbers~\citep[e.g.][]{lang}. The real numbers $\lambda_1,\ldots,\lambda_n\in\R$ are \emph{algebraically independent over $\mathbb{Q}$} if no non-trivial polynomial in $\mathbb{Q}[x_1,\ldots,x_n]$ vanishes on $(\lambda_1,\ldots,\lambda_n)$.  
For a block matrix $M\in\R^{n\times n}$ of the form
\[
M := \begin{pmatrix}
    A & B^\top \\ B & C
\end{pmatrix},
\]
where $A\in\R^{n_1\times n_1}, B\in\R^{n_2\times n_1}$ and $C\in\R^{n_2\times n_2}$  with $n_1+n_2=n$ with $C$ invertible, define the Schur complement of the block $C$ with respect to $M$ is defined as $A-B^\top C^{-1}B$. We will use the result that $M$ is invertible if the Schur complement of the block $C$ with respect to $M$ is invertible~\citep[see e.g.][Proposition 10.10.2]{samworth-shah}.

\begin{proposition}\label{prop:kernel-existence}
    For any $d\in\N$, $p\in\N_0$ there exists a bounded kernel $K:\R^d\to\R$ of order $p+1$ with $\supp K\subseteq\mathcal{B}_0(1)$ such that $s_1(K)$ is invertible. 
\end{proposition}

\begin{proof}
If $p=0$, then the kernel $K$ given by $K(\nu) \propto \max(1 - \|\nu\|,0)$ satisfies the requirements, since $s_1(K) = 1$ in that case.  Henceforth assume that $p \in \mathbb{N}$.  For $\epsilon > 0$ and $u \in \mathbb{R}^d$, let $\phi_{\epsilon}(u):=\ind_{\mathcal{B}_0(\epsilon)}(u)/\epsilon^d$. 
    By Lemma~\ref{lem:G-inverse}, there exists $\epsilon^*\in(0,1/2]$ and $x_1,\ldots,x_{\bar{p}^*}\in\mathcal{B}_0(1/2)$ such that
    \begin{equation*}
        G := 
        \begin{pmatrix}
            g_1(x_1) & \cdots & g_1(x_{\bar{p}^*})
            \\
            \vdots & \ddots & \vdots
            \\
            g_{\bar{p}^*}(x_1) & \cdots & g_{\bar{p}^*}(x_{\bar{p}^*})
        \end{pmatrix}
    \end{equation*}
    is invertible, where 
    \[
    g_m(x):=\int_{\mathcal{B}_0(\epsilon^*)} Q_m^{\mathrm{ext}}(\nu)\phi_{\epsilon^*}(x-\nu)\,d\nu=\int_{\mathcal{B}_0(1)} Q_m^{\mathrm{ext}}(x+\epsilon^* \nu)\,d\nu
    \]
    for $m\in[\bar{p}^*]$.  Given an arbitrary vector $(c_{\alpha'})_{\alpha' \in \mathbb{N}_0^d:p+1\leq \|\alpha'\|_1\leq 2p}$ and $\alpha\in\N_0^d$ with $\|\alpha\|_1\leq 2p$, define
    \begin{equation*}
        \mu_\alpha := \ind_{\{\alpha=\boldsymbol{0}\}}
        +
        \sum_{\substack{\alpha'\in\N_0^d,\\p+1\leq \|\alpha'\|_1\leq 2p}}\frac{c_{\alpha'}}{\alpha'!}\ind_{\{\alpha'=\alpha\}}.
    \end{equation*}
    Now set $\mu:=(\mu_\alpha)_{\alpha\in\N_0^d,\, \|\alpha\|_1\leq 2p}$, with components in graded lexicographical order.  Further, let $d = (d_1,\ldots,d_{\bar{p}^*})^\top := G^{-1}\mu$ and define a bounded function $K:\mathbb{R}^d \rightarrow \mathbb{R}$ with $\supp K\subseteq\mathcal{B}_0(1)$ by
    \begin{equation*}
        K(\nu) := \sum_{j=1}^{\bar{p}^*}d_j \phi_{\epsilon^*}(x_j - \nu).
    \end{equation*}
    For $\alpha\in\N_0^d$ with $\|\alpha\|_1\leq 2p$, let $m(\alpha)\in[\bar{p}^*]$ denote the index for which $Q^{\mathrm{ext}}(\nu)_{m(\alpha)}=\nu^\alpha$.  Then for any such $\alpha$,
    \begin{align}
        \int_{\R^d} \nu^{\alpha} K(\nu)\,d\nu
        &=
        \alpha! \mathrm{e}_{m(\alpha)}^\top \int_{\R^d} K(\nu)Q^{\mathrm{ext}}(\nu)\,d\nu
        =
        \alpha! \mathrm{e}_{m(\alpha)}^\top \sum_{j=1}^{\bar{p}^*}d_j\int_{\mathcal{B}_0(\epsilon^*)}\phi_{\epsilon^*}(x_j - \nu)Q^{\mathrm{ext}}(\nu)\,d\nu
        \notag
        \\
        &=
        \alpha! \mathrm{e}_{m(\alpha)}^\top Gd = \alpha!\mu_{\alpha}=\begin{cases}
            1 &\quad\text{if }\alpha=\boldsymbol{0}
            \\
            0 &\quad\text{if }1\leq\|\alpha\|_1\leq p
            \\
            c_{\alpha} &\quad\text{if }p+1\leq \|\alpha\|_1\leq 2p
        \end{cases}.
        \label{eq:kernel-terms}
    \end{align}
    Thus for arbitrary  $(c_{\alpha})_{\alpha\in\N_0^d,p+1\leq \|\alpha\|_1\leq 2p}$ there exists a bounded kernel $K$ of order $p+1$ with $\supp K\subseteq\mathcal{B}_0(1)$ satisfying~\eqref{eq:kernel-terms}. 
    
    It remains to show we can choose $(c_\alpha)$ such that $s_1(K)$ is invertible.  To see this, we choose $(c_{\alpha})_{\alpha\in\N_0^d,p+1\leq \|\alpha\|_1\leq 2p}$ to be a sequence of algebraically independent real numbers over $\mathbb{Q}$.  Then $s_1(K) \in \mathbb{R}^{\bar{p}^* \times \bar{p}^*}$ can be decomposed as
    \begin{gather*}
        s_1(K)=\begin{pmatrix}
            1 & {\boldsymbol{0}^\top}
            \\
            \boldsymbol{0} & s_1^{(-)}(K)
        \end{pmatrix},
        \quad \text{where} \ 
        s_1^{(-)}(K)
        :=
        \begin{pmatrix}
            {\boldsymbol{0}} & {\boldsymbol{0}} & \cdots & {\boldsymbol{0}}& {\boldsymbol{0}} & S_{1,p}
            \\
            {\boldsymbol{0}} & {\boldsymbol{0}} & \cdots & {\boldsymbol{0}} & S_{2,p-1} & S_{2,p}
            \\
            \vdots & \vdots & \ddots & \vdots & \vdots & \vdots 
            \\
            {\boldsymbol{0}} & {\boldsymbol{0}} & \cdots & S_{p-2,p-2} & S_{p-2,p-1} & S_{p-2,p}
            \\
            {\boldsymbol{0}} & S_{p-1,2} & \cdots & S_{p-1,p-2} & S_{p-1,p-1} & S_{p-1,p}
            \\
            S_{p,1} & S_{p,2} & \cdots & S_{p,p-2} & S_{p,p-1} & S_{p,p}
        \end{pmatrix},
    \end{gather*}
    and where 
    \[    S_{k,\ell} := (c_{\alpha_1+\alpha_2})_{\alpha_1,\alpha_2\in\N_0^d\,:\, \|\alpha_1\|_1=k, \|\alpha_2\|_1=\ell}\in\R^{\mathrm{C}^{k+d-1}_{d-1}\times\mathrm{C}^{\ell+d-1}_{d-1}},
    \]
    with components in lexicographic order.  Moreover, $S_{k,\ell}$ consists of $\mathrm{C}^{k+\ell+d-1}_{d-1}$ unique entries, such that each component of $(c_{\alpha})_{\alpha\in\N_0^d:\|\alpha\|_1=k+\ell}$ appears in exactly one of the `antidiagonals' of $S_{k,\ell}$. 
    Thus for any $k,\ell\in[p]$ with $p+1\leq k+\ell \leq 2p$, we see from Lemma~\ref{lem:alg-indep} that any $\min\bigl(\mathrm{C}^{k+d-1}_{d-1} , \mathrm{C}^{\ell+d-1}_{d-1}\bigr) \times \min\bigl(\mathrm{C}^{k+d-1}_{d-1} , \mathrm{C}^{\ell+d-1}_{d-1}\bigr)$ submatrix of $S_{k,\ell}$ is invertible, so 
    \begin{equation}\label{eq:rank}
        \mathrm{rank}(S_{k,\ell})=\min\bigl(\mathrm{C}^{k+d-1}_{d-1} , \mathrm{C}^{\ell+d-1}_{d-1}\bigr).
    \end{equation}
    
    If $p$ is odd, then $S_{(p+1)/2,(p+1)/2}$ is invertible by~\eqref{eq:rank}. If $p$ is even, then we claim that 
    \begin{equation*}
        U := \begin{pmatrix}
            {\boldsymbol{0}} & S_{p/2,p/2+1}  
            \\S_{p/2+1,p/2} & S_{p/2+1,p/2+1}
        \end{pmatrix},
    \end{equation*}
    is invertible.  Indeed, $S_{p/2+1,p/2+1}$ is invertible, and the Schur complement 
    of $S_{p/2+1,p/2+1}$ with respect to~$U$ is
    \begin{equation*}
        -S_{p/2+1,p/2}^\top S_{p/2+1,p/2+1}^{-1}S_{p/2+1,p/2} \in \mathbb{R}^{\mathrm{C}^{p/2+d-1}_{d-1} \times \mathrm{C}^{p/2+d-1}_{d-1}},
    \end{equation*}
    which is invertible as
    \begin{align*}
        \mathrm{rank}\bigl(S_{p/2+1,p/2}^\top S_{p/2+1,p/2+1}^{-1}S_{p/2+1,p/2}\bigr)
        =
        \mathrm{rank}\bigl(S_{p/2+1,p/2}\bigr)
        &=
        \mathrm{C}^{p/2+d-1}_{d-1}.
    \end{align*}

    We now apply Lemma~\ref{lem:pos-def-inductive} inductively.  Suppose that
    \begin{equation*}
        S_{[q,q_0]}:=\begin{pmatrix}
            {\boldsymbol{0}} & {\boldsymbol{0}} & \cdots & {\boldsymbol{0}}& {\boldsymbol{0}} & S_{q_0,q}
            \\
            {\boldsymbol{0}} & {\boldsymbol{0}} & \cdots & {\boldsymbol{0}} & S_{q_0+1,q-1} & S_{q_0+1,q}
            \\
            \vdots & \vdots & \ddots & \vdots & \vdots & \vdots 
            \\
            {\boldsymbol{0}} & {\boldsymbol{0}} & \cdots & S_{q-2,q-2} & S_{q-2,q-1} & S_{q-2,q}
            \\
            {\boldsymbol{0}} & S_{q-1,q_0+1} & \cdots & S_{q-1,q-2} & S_{q-1,q-1} & S_{q-1,q}
            \\
            S_{q,q_0} & S_{q,q_0+1} & \cdots & S_{q,q-2} & S_{q,q-1} & S_{q,q}
        \end{pmatrix},
    \end{equation*}
    is invertible for some $2\leq q_0\leq q \leq p-1$ with $q+q_0-1=p$. Then by Lemma~\ref{lem:pos-def-inductive} with
    \begin{gather*}
        A := S_{[q,q_0]},
        \quad
        B := S_{q+1,q_0-1},
        \quad
        D := \begin{pmatrix}
            S_{q+1,q_0}&\cdots&S_{q+1,q}
        \end{pmatrix},
    \end{gather*}
    and $M_{\Phi}:\R^{\mathrm{C}^{2q+d+3}_{d-1}}\to\R^{\mathrm{C}^{q+d+1}_{d-1}\times \mathrm{C}^{q+d+1}_{d-1}}$ of the form in Lemma~\ref{lem:alg-indep} with $M_{\Phi}\bigl((c_{\alpha})_{\alpha \in \mathbb{N}_0^d:\|\alpha\|_1=2q+2}\bigr)=S_{q+1,q+1}$, we have that $S_{[q+1,q_0-1]}$ is invertible.  By induction, we conclude that $s_1^{(-)}(K)=S_{[p,1]}$ is invertible, so $s_1(K)$ is invertible. 
    \end{proof}

\begin{lemma}\label{lem:G-inverse}
    There exist $\epsilon^* \in (0,1/2]$ and $x_1,\ldots,x_{\bar{p}^*}\in\mathcal{B}_0(1/2)$ such that the matrix $G\in\R^{\bar{p}^*\times\bar{p}^*}$ with $(m,j)$th entry
    \begin{equation*}
        G_{mj} := \int_{\mathcal{B}_0(1)} Q^{\mathrm{ext}}_m(x_j+\epsilon^* u)\,du
    \end{equation*}
    is invertible. 
\end{lemma}

\begin{proof}
Since $Q^{\mathrm{ext}}_1,\ldots,Q^{\mathrm{ext}}_{\bar{p}^*}$ are linearly independent, by~\citet[Lemma~11]{rose}, there exist $x_1,\ldots,x_{\bar{p}^*} \in \mathcal{B}_0(1/2)$ such that 
\[
A := \begin{pmatrix}
           Q^{\mathrm{ext}}_1(x_1) & \cdots & Q^{\mathrm{ext}}_1(x_{\bar{p}^*})
            \\
            \vdots&\ddots&\vdots
            \\
            Q^{\mathrm{ext}}_{\bar{p}^*}(x_1) & \cdots & Q^{\mathrm{ext}}_{\bar{p}^*}(x_{\bar{p}^*})
        \end{pmatrix}
\]
is invertible.  Now define the continuous function $G:[0,\infty) \to\R^{\bar{p}^*\times\bar{p}^*}$ by
    \begin{equation*}
        G_{mj}(\epsilon) := \int_{\mathcal{B}_0(1)} Q^{\mathrm{ext}}_m(x_j+\epsilon u)\,du.
    \end{equation*}
Then $\det G(0) = \det A \neq 0$, and $\det$ is a continuous function, so there exists $\epsilon^* \in (0,1/2]$ such that $\det G(\epsilon^*) \neq 0$, as required.
\end{proof}

\begin{lemma}\label{lem:alg-indep}
    Let $n \geq 2$, let $\Phi:[n]\to\mathcal{P}([m] \times [m])$ be such that $\Phi(1),\ldots,\Phi(n)$ partition $[m] \times [m]$, and such that for each $i \in [n]$, either $\Phi(i) \subseteq \{(j,m+1-j):j \in [m]\}$ or $\Phi(i) \subseteq \{(j,m+1-k):j,k \in [m] \text{ are distinct}\}$.  Now define $M_\Phi:\mathbb{R}^n \rightarrow \mathbb{R}^{m \times m}$ by $M_\Phi(a_1,\ldots,a_n)_{jk} := \sum_{i=1}^n a_i \cdot \mathbbm{1}_{\{(j,k) \in \Phi(i)\}}$.  Then for all $a\in\R^n$ with components consisting of algebraically independent numbers over $\mathbb{Q}$, the matrix $M_\Phi(a)$ is invertible.
\end{lemma}

\begin{proof}
    Define the determinant function $\mathcal{D}:\R^n\to\R$ by
    \begin{equation*}
        \mathcal{D}(a) := \det M_{\Phi}(a).
    \end{equation*}
    Then $\mathcal{D}$ is a polynomial with integer coefficients. Moreover, $\mathcal{D}\not\equiv 0$, because there exists $a^*\in\{0,1\}^n$ such that 
    $\mathcal{D}(a^*)=\det \bigl((\mathbbm{1}_{\{j+k=m+1\}})_{j,k \in [m]}\bigr)=1$. Therefore, for arbitrary $a\in\R^n$ whose components are algebraically independent numbers over $\mathbb{Q}$ (thus also over $\mathbb{Z}$), we have $\mathcal{D}(a)\neq 0$, so $M_\Phi(a)$ is invertible.
\end{proof}

\begin{lemma}\label{lem:pos-def-inductive}
    Let $m_1\leq m_2 \leq m_3$ and $n_c \geq 2$, and let $A\in\R^{m_2\times m_2}$, $B\in\R^{m_3\times m_1}$, $D\in\R^{m_3\times m_2}$, and $M:\R^{n_c}\to\R^{m_3\times m_3}$. Suppose that
    \begin{enumerate}[label=(\roman*)]
        \item $A$ is invertible;
        \item The matrix function $M$ is of the form $M_\Phi$ in Lemma~\ref{lem:alg-indep} (with $m$ replaced with $m_3$); 
        \item $B$ has linearly independent columns.
    \end{enumerate}
    Then for any $a\in\R^{n_c}$ with algebraically independent components, the matrix
    \begin{equation*}
        W := \begin{pmatrix}
            {\boldsymbol{0}_{m_1\times m_1}}
            &
            {\boldsymbol{0}_{m_1\times m_2}}
            &
            B^\top
            \\
            {\boldsymbol{0}_{m_2\times m_1}}
            & A & D^\top
            \\
            B & D & M(a)
        \end{pmatrix}
    \end{equation*}
    is invertible.
\end{lemma}
\begin{proof}
    Take $a\in\R^{n_c}$. Since $A$ is invertible, the matrix
    \begin{equation*}
        W_1 := \begin{pmatrix}
        A & D^\top
        \\
        D & M(a)
    \end{pmatrix}
    \end{equation*}
    is invertible if and only if the Schur complement of $A$, namely
    \begin{equation*}
        W_2 := M(a)-DA^{-1}D^\top,
    \end{equation*}
    is invertible.  Now define the symmetric permutation matrix $\tilde{Q} = (\tilde{Q}_{ij}) \in \mathbb{R}^{m_3 \times m_3}$ by $\tilde{Q}_{ij} := \mathbbm{1}_{\{i+j=m_3+1\}}$.  The determinant function
    \begin{equation*}
        d(a) := \det\bigl(M(a)-DA^{-1}D^\top\bigr)
    \end{equation*}
    is a polynomial in $a$.  Moreover, for any $\lambda>0$,
    there exists $a_\lambda^*\in\R^{n_c}$ such that $M(a_\lambda^*)=\lambda \tilde{Q}$, and thus 
    \begin{equation*}
        d(a_{\lambda}^*) = \det\bigl(\lambda \tilde{Q} -DA^{-1}D^\top\bigr) = \det\bigl(\lambda I_{m_3} -DA^{-1}D^\top\tilde{Q}\bigr).
    \end{equation*}
    It follows that $d\not\equiv0$, since otherwise the matrix $DA^{-1}D^\top\tilde{Q}$ has infinitely many eigenvalues. 
    Therefore for any $a\in\R^{n_c}$ whose components are algebraically independent over $\mathbb{Q}$, we have $d(a)\neq 0$, so~$W_2$ and consequently~$W_1$ are both invertible.

    The Schur complement of~$W_1$ with respect to $W$ is then 
    \begin{equation*}
        W_3 :=
        -
        \begin{pmatrix}
            {\boldsymbol{0}}_{m_1\times m_2} & B^\top
        \end{pmatrix}
        \begin{pmatrix}
            A & D^\top 
            \\
            D & M(a)
        \end{pmatrix}^{-1}
        \begin{pmatrix}
            {\boldsymbol{0}}_{m_2\times m_1} \\ B
        \end{pmatrix} \in \mathbb{R}^{m_1 \times m_1}
        ,
    \end{equation*}
    whose rank is equal to $\mathrm{rank}(B)=m_1$, so~$W_3$ is full rank, so invertible. 
    We conclude that the matrix~$W$ is invertible.
\end{proof}

\section{Proof of Theorem~\ref{thm:do-no-worse}}\label{appsec:proof-noworse}


For $\omega:\R\to\R$ and $\beta>0$, define the $\beta$-H\"older seminorm $$\|\omega\|_{C_\beta}:=\sup_{x\neq y}\frac{|\omega^{(\beta_0)}(x)-\omega^{(\beta_0)}(y)|}{\|x-y\|^{\beta-\beta_0}},$$
where $\beta_0:=\ceil{\beta}-1$. 


\begin{proof}[Proof of Theorem~\ref{thm:do-no-worse}]

As shorthand, we write $\beta := \beta$ and $L := L$.  
For $x \in \mathcal{X}$, define $\sigma_P^2(x):=\Var_P(Y\given X=x)$. 
By Theorem~\ref{thm:decomp}, for $\bar{f} \in \cH(\beta,L)$, 
\begin{equation*}
    \hat{f}_n^{\text{Outrig}}(x_0)-\bar{f}(x_0) = X_n'
    + R_n',
    \qquad
    X_n' := B(\bar{f},x_0,K,h)h^{\beta^*}+\frac{1}{\sqrt{nh^d}}\bigg\{\frac{R_2(K)V_{P,x_0}^{(\lambda)}(\rho)}{p_X(x_0)}\bigg\}^{1/2}Z_n'
\end{equation*}
while by more standard arguments,
\begin{equation*}
    \hat{f}_n^{\mathrm{LP}}(x_0)-\bar{f}(x_0) = X_n''
    + R_n'',
    \qquad
    X_n'' := 
    B(\bar{f},x_0,K,h)h^{\beta^*}+\frac{1}{\sqrt{nh^d}}\bigg\{\frac{R_2(K)\sigma_P^2(x_0)}{p_X(x_0)}\bigg\}^{1/2}Z_n'', 
\end{equation*}
where $Z_n',Z_n'',R_n',R_n''$ are random variables satisfying
\begin{equation*}
    \sup_{P\in\cP}\sup_{x_0\in\cX}\sup_{(h,\lambda)\in\cH_n\times\Lambda_n}\sup_{t\in\R}\big|\Pr_P( Z_n' \leq t)-\Phi(t)\big|\to0
    ,\;
    \sup_{P\in\cP}\sup_{x_0\in\cX}\sup_{(h,\lambda)\in\cH_n\times\Lambda_n}\sup_{t\in\R}\big|\Pr_P( Z_n'' \leq t)-\Phi(t)\big| \to 0,
\end{equation*}
as $n\to\infty$, where the supremum is taken over $P=(P_X,P_{
\varepsilon|X},\bar{f})$, and
\begin{equation*}
    R_n',\,R_n'' = \op\bigg(h^{\beta^*}+\frac{1}{\sqrt{nh^d}}\bigg).
\end{equation*}
Take an arbitrary $\zeta>0$.  Then using the fact that $(a+b)^2 \leq (1+\zeta)(a^2 + b^2/\zeta)$ for $a,b \in \mathbb{R}$, we have on the event $\bigl\{|R_n'|\leq \zeta s_{n,h}^{1/2}\bigr\}$ that

\begin{align*}
    \frac{(X_n'+R_n')^2}{s_{n,h}} \wedge M
    &\leq
    \biggl\{\frac{(1+\zeta)}{s_{n,h}}\Big(X_n'^2+\frac{1}{\zeta}R_n'^2\Big)\biggr\}\wedge M
    \leq
    \biggl\{\frac{(1+\zeta)X_n'^2}{s_{n,h}} \wedge M\biggr\} + \frac{1+\zeta}{\zeta}\frac{R_n'^2}{s_{n,h}}
    \\
    &\leq
    (1+\zeta)\Big(\frac{X_n'^2}{s_{n,h}}\wedge M\Big) + \zeta(1+\zeta)
    \leq
    \biggl(\frac{X_n'^2}{s_{n,h}} \wedge M\biggr)
    + \zeta M + \zeta(1+\zeta).
\end{align*}
Thus
\begin{align*}
    \E_P\bigl\{ \bigl(s_{n,h}^{-1}(X_n'+R_n')^2\bigr)\wedge M \bigr\}
    &=
    \E_P\Bigl[\bigl\{ \bigl(s_{n,h}^{-1}(X_n'+R_n')^2\bigr)\wedge M\bigr\} \ind_{\{|R_n'|\leq \zeta s_{n,h}^{1/2}\}} \Bigr]
    \\
    &\hspace{2cm}
    +
    \E_P\Bigl[\bigl\{ \bigl(s_{n,h}^{-1}(X_n'+R_n')^2\bigr)\wedge M\bigr\} \ind_{\{|R_n'| > \zeta s_{n,h}^{1/2}\}} \Bigr]
    \\
    &\leq
    \E_P\biggl(\frac{X_n'^2}{s_{n,h}} \wedge M\biggr) 
    + \zeta M
    + \zeta(1+\zeta)
    + M\,\Pr_P\bigl(|R_n'|> \zeta s_{n,h}^{1/2}\bigr).
\end{align*}
Analogously,
\begin{equation*}
    \E_P\biggl(\frac{X_n'^2}{s_{n,h}} \wedge M\biggr) 
    \leq 
    \E_P\bigl\{ \bigl(s_{n,h}^{-1}(X_n'+R_n')^2\bigr)\wedge M \bigr\}
    + \zeta M
    + \zeta(1+\zeta)
    + M\,\Pr_P\bigl(|R_n'|> \zeta s_{n,h}^{1/2}\bigr).
\end{equation*}
Therefore
\begin{align*}
\limsup_{n\to\infty}\sup_{P\in\cP} \sup_{x_0\in\cX}&\sup_{(h,\lambda)\in\cH_n\times\Lambda_n}\Bigl|
    \E_P\bigl\{ \bigl(s_{n,h}^{-1}(X_n'+R_n')^2\bigr) \wedge M \bigr\}
    -
    \E_P\bigl\{ \bigl(s_{n,h}^{-1}X_n'^2\bigr)\wedge M \bigr\}
    \Bigr|
    \\
    &\leq
    \zeta M + \zeta(1+\zeta) + M\limsup_{n\to\infty}\sup_{P\in\cP}\sup_{x_0\in\cX}\sup_{(h,\lambda)\in\cH_n\times\Lambda_n}\Pr_P(|s_{n,h}^{-1/2}R_n'|>\zeta)
    \\
    &=
    \zeta M + \zeta(1+\zeta).
\end{align*}
Since $\zeta > 0$ was arbitrary,
\begin{equation}\label{eq:Xn'Rn'}
    \lim_{n\to\infty}\sup_{P\in\cP}\sup_{x_0\in\cX}\sup_{(h,\lambda)\in\cH_n\times\Lambda_n}\Bigl|
    \E_P\bigl\{ \bigl(s_{n,h}^{-1}(X_n'+R_n')^2\bigr)\wedge M \bigr\}
    -
    \E_P\bigl\{\bigl(s_{n,h}^{-1}X_n'^2\bigr)\wedge M \bigr\}
    \Bigr| = 0.
\end{equation}
By almost identical arguments,
\begin{equation}\label{eq:Xn''Rn''}
    \lim_{n\to\infty}\sup_{P\in\cP}\sup_{x_0\in\cX}\sup_{h\in\cH_n}\Bigl|
    \E_P\bigl\{ \bigl(s_{n,h}^{-1}(X_n''+R_n'')^2\bigr)\wedge M \bigr\}
    -
    \E_P\bigl\{ \bigl(s_{n,h}^{-1}X_n''^2\bigr)\wedge M \bigr\}
    \Bigr| = 0.
\end{equation}

For $M>0$, define the function $G_M:\R\times[0,\infty)\to[0,\infty)$ by
\begin{equation}\label{eq:G_M}
G_M(\mu,V):= \E_{Z\sim N(\mu,V)}(Z^2\wedge M),
\end{equation}
where we interpret $Z \sim N(\mu,0)$ as a random variable satisfying $\mathbb{P}(Z=\mu) = 1$.  Further, for $\delta > 0$ as in Lemma~\ref{lem:supB-delta}, define
\begin{equation}\label{eq:CG}
    C_G :=
    \inf_{t\in[0,1]}G_M\Bigl(\delta\sqrt{1-t},\,
    \frac{R_2(K)c_3t}{C_X} \Bigr).
\end{equation}
Note in particular that $C_G$ does not depend on $(P,x_0,h,\lambda)$. 
We also claim that $C_G>0$. Indeed, $C_G$ is an infimum of a continuous function on a compact set, so the infimum is attained at some $t^*\in[0,1]$, and $C_G=G_M\big(\delta\sqrt{1-t^*},\,\tfrac{R_2(K)c_3t^*}{C_X}\big) > 0$ by Lemma~\ref{lem:GM}\ref{lem:GM-3}. Now let $\eta \in \bigl(0,\frac{1}{4}C_G\bigr]$.  By~\eqref{eq:Xn'Rn'}~and~\eqref{eq:Xn''Rn''} there exists $N_1\in\mathbb{N}$, not depending on $(P,x_0,h,\lambda)$, such that for $n\geq N_1$,
\begin{equation*}
    \frac{\mathcal{R}_{n,h,P,x_0,\epsilon,M}(\hat{f}^{\mathrm{Outrig}})}{\mathcal{R}_{n,h,P,x_0,\epsilon,M}(\hat{f}^{\mathrm{LP}})}
    \leq
    \frac{
    \sup_{\tilde{f}\in\cH_{f,\epsilon}^{\mathrm{loc}}(\beta,L)}
        \E_{(P_X,P_{\varepsilon|X},\tilde{f})}\bigl\{\big(s_{n,h}^{-1}X_n'^2\big)\wedge M\bigr\}+\eta}{
        \sup_{\tilde{f}\in\cH_{f,\epsilon}^{\mathrm{loc}}(\beta,L)}
        \E_{(P_X,P_{\varepsilon|X},\tilde{f})}\bigl\{\big(s_{n,h}^{-1}X_n''^2\big)\wedge M\bigr\}-\eta}.
\end{equation*}
Now take $Y \sim N(0,1)$ and let
\begin{align*}
    Y_n' \equiv Y_n'(\bar{f}) &:= B(\bar{f},x_0,K,h)h^{\beta^*}+\frac{1}{\sqrt{nh^d}}\biggl(\frac{R_2(K)V_{P,x_0}^{(\lambda)}(\rho)}{p_X(x_0)}\biggr)^{1/2}Y,
    \\
    Y_n'' \equiv Y_n''(\bar{f}) &:= B(\bar{f},x_0,K,h)h^{\beta^*}+\frac{1}{\sqrt{nh^d}}\biggl(\frac{R_2(K)\sigma_P^2(x_0)}{p_X(x_0)}\biggr)^{1/2}Y.
\end{align*}
By Lemma~\ref{lem:unif-portmanteau}, there exists $N_2\in\mathbb{N}$, not depending on $(P,x_0,h,\lambda)$, such that for all $n\geq N_1 \vee N_2$,
\begin{equation*}
    \frac{\mathcal{R}_{n,h,P,x_0,\epsilon,M}(\hat{f}^{\mathrm{Outrig}})}{\mathcal{R}_{n,h,P,x_0,\epsilon,M}(\hat{f}^{\mathrm{LP}})}
    \leq
    \frac{
    \sup_{\tilde{f}\in\cH_{f,\epsilon}^{\mathrm{loc}}(\beta,L)}
        \E\bigl[\bigl\{s_{n,h}^{-1}\bigl(Y_n'(\tilde{f})\bigr)^2\bigr\}\wedge M\bigr]+2\eta}{
        \sup_{\tilde{f}\in\cH_{f,\epsilon}^{\mathrm{loc}}(\beta,L)}
        \E\bigl[\bigl\{s_{n,h}^{-1}\bigl(Y_n'(\tilde{f})\bigr)^2\bigr\}\wedge M\bigr]-2\eta}.
\end{equation*}
From hereon we will take $n\geq N_1\vee N_2$. 
Noting that $G_M(\mu,V)=G_M(-\mu,V)$,
\begin{align*}
    \E\bigl[\bigl\{s_{n,h}^{-1}\bigl(Y_n'(\tilde{f})\bigr)^2\bigr\}\wedge M\bigr] &=
    G_M\bigg(|B(\tilde{f},x_0,K,h)|s_{n,h}^{-1/2}h^{\beta^*},\,
    \frac{R_2(K)V_{P,x_0}^{(\lambda)}(\rho)}{p_X(x_0)}\frac{1}{s_{n,h}nh^d}
    \bigg),
    \\
    \E\bigl[\bigl\{s_{n,h}^{-1}\bigl(Y_n'(\tilde{f})\bigr)^2\bigr\}\wedge M\bigr]
    &=
    G_M\bigg(|B(\tilde{f},x_0,K,h)|s_{n,h}^{-1/2}h^{\beta^*},\,
    \frac{R_2(K)\sigma_P^2(x_0)}{p_X(x_0)}\frac{1}{s_{n,h}nh^d}
    \bigg).
\end{align*}
By Lemma~\ref{lem:GM}\ref{lem:GM-1}, $\mu\mapsto G_M(\mu,V)$ is strictly increasing in $|\mu|$, for each $M,V > 0$. Therefore, for any $V>0$,
\begin{equation*}
    \sup_{\tilde{f}\in\cH_{f,\epsilon}^{\mathrm{loc}}(\beta,L)}G_M\bigl(|B(\tilde{f},x_0,K,h)|s_{n,h}^{-1/2}h^{\beta^*},\,V
    \bigr)
    =
    G_M\bigl(\bar{B}_{f,\epsilon,x_0,K,h}\,s_{n,h}^{-1/2}h^{\beta^*},\,V
    \bigr)
    ,
\end{equation*}
where we use the shorthand
$
    \bar{B}_{f,\epsilon,x_0,K,h} := 
    \sup_{\tilde{f}\in\cH_{f,\epsilon}^{\mathrm{loc}}(\beta,L)}|B(\tilde{f},x_0,K,h)|.
$
Therefore
\begin{equation*}
    \frac{
    \sup_{\tilde{f}\in\cH_{f,\epsilon}^{\mathrm{loc}}(\beta,L)}
        \E\big[\big(s_{n,h}^{-1}\{Y_n'(\tilde{f})\}^2\big)\wedge M\big]+2\eta}{
        \sup_{\tilde{f}\in\cH_{f,\epsilon}^{\mathrm{loc}}(\beta,L)}
        \E\big[\big(s_{n,h}^{-1}\{Y_n''(\tilde{f})\}^2\big)\wedge M\big]-2\eta}
        =
        \frac{G_M\Big(\bar{B}_{f,\epsilon,x_0,K,h}\,s_{n,h}^{-1/2}h^{\beta^*},\,
    \frac{R_2(K)V_{P,x_0}^{(\lambda)}(\rho)}{p_X(x_0)}\frac{1}{s_{n,h}nh^d}\Big)+2\eta}{G_M\Big(\bar{B}_{f,\epsilon,x_0,K,h}\,s_{n,h}^{-1/2}h^{\beta^*},\,
    \frac{R_2(K)\sigma_P^2(x_0)}{p_X(x_0)}\frac{1}{s_{n,h}nh^d}\Big)-2\eta},
\end{equation*}
and so
\begin{align}
    \sup_{P\in\cP}
        &\sup_{x_0\in\cX}
        \sup_{(h,\lambda)\in\cH_n\times\Lambda_n}
        \frac{\mathcal{R}_{n,h,P,x_0,\epsilon,M}(\hat{f}^{\mathrm{Outrig}})}{\mathcal{R}_{n,h,P,x_0,\epsilon,M}(\hat{f}^{\mathrm{LP}})}
        \notag
        \\
        &\leq
        \sup_{P\in\cP}\sup_{x_0\in\cX}\sup_{(h,\lambda)\in\cH_n\times\Lambda_n}
        \frac{G_M\Big(\bar{B}_{f,\epsilon,x_0,K,h}\,s_{n,h}^{-1/2}h^{\beta^*},\,
    \frac{R_2(K)V_{P,x_0}^{(\lambda)}(\rho)}{p_X(x_0)}\frac{1}{s_{n,h}nh^d}\Big)+2\eta}{G_M\Big(\bar{B}_{f,\epsilon,x_0,K,h}\,s_{n,h}^{-1/2}h^{\beta^*},\,
    \frac{R_2(K)\sigma_P^2(x_0)}{p_X(x_0)}\frac{1}{s_{n,h}nh^d}\Big)-2\eta}
    \notag
    \\
    &=
    \sup_{P\in\cP}\sup_{x_0\in\cX}\sup_{(h,\lambda)\in\cH_n\times\Lambda_n}
    \frac{G_M\Big(\bar{B}_{f,\epsilon,x_0,K,h}\,\alpha_1\big(nh^{2\beta^*\hspace{-0.1em}+d}\big),\,
    \frac{R_2(K)V_{P,x_0}^{(\lambda)}(\rho)}{p_X(x_0)}\alpha_2\big(nh^{2\beta^*\hspace{-0.1em}+d}\big)\Big)+2\eta}{G_M\Big(\bar{B}_{f,\epsilon,x_0,K,h}\,\alpha_1\big(nh^{2\beta^*\hspace{-0.1em}+d}\big),\,
    \frac{R_2(K)\sigma_P^2(x_0)}{p_X(x_0)}\alpha_2\big(nh^{2\beta^*\hspace{-0.1em}+d}\big)\Big)-2\eta},
    \label{eq:lemma-step1}
\end{align}
where $\alpha_1(u):=\sqrt{\frac{u}{1+u}}$ 
and $\alpha_2(u):=\frac{1}{1+u}$. By Lemma~\ref{lem:B(f)}\ref{lem:bias3} there exists $C_B(K,d,\beta_0^*,L) > 0$, not depending on $(P,x_0,h,\lambda)$, such that 
\[
|B(\tilde{f},x_0,K,h)|\alpha_1\big(nh^{2\beta^*\hspace{-0.1em}+d}\big) \leq C_B(K,d,\beta_0^*,L).
\]
By Lemma~\ref{lem:gauss-eq}, $V_{P,x_0}^{(\lambda)}(\rho_{\mathrm{score}})\leq\sigma_P^2(x_0)$, and moreover, by Lemma~\ref{lem:GM}, $G_M$ is coordinate-wise strictly increasing on the restricted domain $[0,C_B(K,d,\beta_0^*,L)]\times[0,\infty)$ for $M\geq \frac{25}{4}C_B(K,d,\beta_0^*,L)^2$.  Hence, for such $M$,
    \begin{align}
    \sup_{P\in\cP}&\sup_{x_0\in\cX}\sup_{(h,\lambda)\in\cH_n\times\Lambda_n}
\frac{G_M\Big(\bar{B}_{f,\epsilon,x_0,K,h}\,\alpha_1\big(nh^{2\beta^*\hspace{-0.1em}+d}\big),\,
    \frac{R_2(K)V_{P,x_0}^{(\lambda)}(\rho_{\mathrm{score}})}{p_X(x_0)}\alpha_2\big(nh^{2\beta^*\hspace{-0.1em}+d}\big)\Big)+2\eta}{G_M\Big(\bar{B}_{f,\epsilon,x_0,K,h}\,\alpha_1\big(nh^{2\beta^*\hspace{-0.1em}+d}\big),\,
    \frac{R_2(K)\sigma_P^2(x_0)}{p_X(x_0)}\alpha_2\big(nh^{2\beta^*\hspace{-0.1em}+d}\big)\Big)-2\eta}
    \notag
    \\
    &\leq
    \sup_{P\in\cP}\sup_{x_0\in\cX}\sup_{h\in\cH_n}
    \frac{G_M\Big(\bar{B}_{f,\epsilon,x_0,K,h}\,\alpha_1\big(nh^{2\beta^*\hspace{-0.1em}+d}\big),\,
    \frac{R_2(K)\sigma_P^2(x_0)}{p_X(x_0)}\alpha_2\big(nh^{2\beta^*\hspace{-0.1em}+d}\big)\Big)+2\eta}{G_M\Big(\bar{B}_{f,\epsilon,x_0,K,h}\,\alpha_1\big(nh^{2\beta^*\hspace{-0.1em}+d}\big),\,
    \frac{R_2(K)\sigma_P^2(x_0)}{p_X(x_0)}\alpha_2\big(nh^{2\beta^*\hspace{-0.1em}+d}\big)\Big)-2\eta}
    \notag
    \\
    &=
    1+\sup_{P\in\cP}\sup_{x_0\in\cX}\sup_{h\in\cH_n}
    \frac{4\eta}{G_M\Big(\bar{B}_{f,\epsilon,x_0,K,h}\,\alpha_1\big(nh^{2\beta^*\hspace{-0.1em}+d}\big),\,
    \frac{R_2(K)\sigma_P^2(x_0)}{p_X(x_0)}\alpha_2\big(nh^{2\beta^*\hspace{-0.1em}+d}\big)\Big)-2\eta}.
    \label{eq:donoworse-step1}
    \end{align}
    Now, by Lemma~\ref{lem:supB-delta}, there exists $N_3\in\mathbb{N}$ and $\delta>0$, neither depending on $(P,x_0,h,\lambda)$, such that for all $n\geq N_3$,
    \begin{equation*}
        \inf_{P\in\cP}\inf_{x_0\in\cX}\inf_{h\in\cH_n}\bar{B}_{f,\epsilon,x_0,K,h}\,
        =\inf_{P\in\cP}\inf_{x_0\in\cX}\inf_{h\in\cH_n}\sup_{\tilde{f}\in\cH_{f,\epsilon}^{\mathrm{loc}}(\beta,L)}|B(\tilde{f},x_0,K,h)|
        \geq\delta.
    \end{equation*}
    From hereon we take $n\geq N_1\vee N_2\vee N_3$. 
    Then for $M\geq\frac{25}{4}C_B(K,d,\beta_0^*,L)^2$,
    \begin{align}
        \inf_{P\in\cP} \inf_{x_0\in\cX}\inf_{h\in\cH_n}&G_M\bigg(\bar{B}_{f,\epsilon,x_0,K,h}\,\alpha_1\big(nh^{2\beta^*\hspace{-0.1em}+d}\big),\,
        \frac{R_2(K)\sigma_P^2(x_0)}{p_X(x_0)}\alpha_2\big(nh^{2\beta^*\hspace{-0.1em}+d}\big)\bigg)
        \notag
        \\
        &\geq
        \inf_{h\in\cH_n}G_M\bigg(\delta\alpha_1\big(nh^{2\beta^*\hspace{-0.1em}+d}\big),\,
        \frac{R_2(K)c_3}{C_X}\alpha_2\big(nh^{2\beta^*\hspace{-0.1em}+d}\big)\bigg)
        \notag
        \\
        &\geq
        \inf_{t\in[0,1]}G_M\biggl(\delta\sqrt{1-t},\,
        \frac{R_2(K)c_3t}{C_X}\biggr)
        =C_G
        ,
        \label{eq:CG-bound}
    \end{align}
    with the final inequality following from the facts that $\alpha_1=
    \sqrt{1-\alpha_2}$, and $\alpha_2([0,\infty])=[0,1]$. Thus for $\eta\in(0,\frac{1}{4}C_G]$, 
    \begin{align}
        &\quad
        \sup_{P\in\cP}\sup_{x_0\in\cX}\sup_{h\in\cH_n}
    \frac{4\eta}{G_M\Big(\bar{B}_{f,\epsilon,x_0,K,h}\,\alpha_1\big(nh^{2\beta^*\hspace{-0.1em}+d}\big),\,
    \frac{R_2(K)\sigma_P^2(x_0)}{p_X(x_0)}\alpha_2\big(nh^{2\beta^*\hspace{-0.1em}+d}\big)\Big)-2\eta}
    \leq
    \frac{8\eta}{C_G}.
        \label{eq:donoworse-step2}
    \end{align}
    Combining~\eqref{eq:lemma-step1},~\eqref{eq:donoworse-step1}~and~\eqref{eq:donoworse-step2}, 
    \begin{equation*}
        \sup_{P\in\cP}
        \sup_{x_0\in\cX}
        \sup_{(h,\lambda)\in\cH_n\times\Lambda_n}
        \frac{\mathcal{R}_{n,h,P,x_0,\epsilon,M}(\hat{f}^{\mathrm{Outrig}})}{\mathcal{R}_{n,h,P,x_0,\epsilon,M}(\hat{f}^{\mathrm{LP}})} \leq 1+\frac{8\eta}{C_G}.
    \end{equation*}
    Since $\eta\in\big(0,\frac{1}{4}C_G\big]$, $M\geq\frac{25}{4}C_B(K,d,\beta_0^*,L)^2$ and $\epsilon > 0$ were arbitrary, the result follows.
\end{proof}

\subsection{Auxiliary results for Theorem~\ref{thm:do-no-worse}}

\begin{lemma}\label{lem:supB-delta}
    Let $\beta,L,\epsilon > 0$ and suppose that the kernel $K$ satisfies Assumption~\ref{ass:kernel}.  Then there exists $\delta \equiv \delta(\beta,L,\epsilon,K) > 0$ such that 
    \begin{equation}\label{eq:suffices}
    \liminf_{n\to\infty}
    \inf_{f \in \mathcal{H}(\beta,L)}
    \inf_{x_0\in\cX} \inf_{h\in\cH_n}\sup_{\tilde{f}\in\cH_{f,\epsilon}^{\mathrm{loc}}(\beta,L)}|B(\tilde{f},x_0,K,h)| \geq \delta,
\end{equation}
where $B$ is as in~\eqref{eq:B-def}.
\end{lemma}
\begin{proof}
Define $\psi:\R\to[0,1]$ and $\omega:\R\to\R$ by
\begin{equation}\label{eq:omega-psi}
    \psi(u) := \begin{cases}
        1 &\quad\text{if } |u|<1
        \\
        \exp\Big(1-\frac{1}{|u|(2-|u|)}\Big) &\quad\text{if } 1\leq|u|\leq2
        \\
        0 &\quad\text{if }|u|>2
    \end{cases},
    \qquad
    \omega(u) := |u|^{\beta^*}\psi(u).
\end{equation}
Then $\|\omega\|_{C_{\beta^*}}<\infty$ by Lemma~\ref{lem:seminorm}. 
Further define
\begin{equation*}
    a_{\beta^*} := \begin{cases}
        1 &\quad\text{if $\beta_0^*=0$}
        \\
        \frac{(\beta^*)_{\beta_0^*-1}}{(\beta_0^*-1)!} &\quad \text{if $\beta_0^*\geq1$}
    \end{cases}.
\end{equation*}
We will prove~\eqref{eq:suffices} with $\delta := \frac{\xi}{2-\xi}\frac{a_{\beta^*}L}{d\|\omega\|_{C_{\beta^*}}}|\mu_{\beta^*}(K)|$, where $\xi := \frac{2\epsilon}{2L+\epsilon}\wedge1$. 
Take an arbitrary $x_0\in\cX$. If $|B(f,x_0,K,h)|\geq\delta$ then~\eqref{eq:suffices} follows immediately. We claim that if $|B(f,x_0,K,h)|<\delta$ then there exists a sequence $\tilde{f}_n\in\cH(\beta,L)$ satisfying 
\begin{equation*}
    \|\tilde{f}_n-f\|_\infty\leq\epsilon
    \quad\text{and}\quad
    \inf_{P\in\cP}\inf_{x_0\in\cX}\inf_{h\in\cH_n}|B(\tilde{f}_n,x_0,K,h)|\geq\delta
\end{equation*}
for sufficiently large $n\in\N$. 
Define $g:\R^d\to\R$ by
\begin{equation*}
    g(x;h) := \frac{Lh^{\beta^*}}{d\|\omega\|_{C_{\beta^*}}}\sum_{\ell=1}^d \omega\bigg(\frac{(x-x_0)_\ell}{h}\bigg).
\end{equation*}
Note the following properties of $g$:
\begin{enumerate}[label=(\roman*), leftmargin=*, align=left]
    \item $\|g(\cdot\,;h)\|_\infty \leq \frac{L\|\omega\|_\infty}{\|\omega\|_{C_{\beta^*}}}\,h^{\beta^*}$.
    \item For all $\alpha\in\N_0^d$ with $\|\alpha\|_1\leq\beta^*_0$ we have $\|\partial^\alpha g(\cdot\,;h)\|_\infty \leq \frac{L\|\partial^\alpha\omega\|_\infty}{d\|\omega\|_{C_{\beta^*}}}h^{\beta^*-\|\alpha\|_1}$.
    \item Take $\alpha\in\N_0^d$ with $\|\alpha\|_1=\beta_0^*$. 
    In the case $\beta_0^*=0$,
    \begin{equation*}
        \frac{|g(x;h)-g(y;h)|}{\|x-y\|^{\beta^*}}
        \leq 
        L.
    \end{equation*}
    Now consider the case $\beta_0^*\geq 1$. 
    If $\|\alpha\|_0 \geq 2$ then $\partial^\alpha g(x;h) = 0$. 
    If $\|\alpha\|_0=1$, then define $\ell^*\in[d]$ to be the index with $\alpha_{\ell^*} = \beta_0^*$. Then
    \begin{multline*}
        \frac{|\partial^\alpha g(x;h)-\partial^\alpha g(y;h)|}{\|x-y\|^{\beta^*-\beta_0^*}}
        =\frac{\Big|\frac{\partial^{\beta_0^*}}{\partial x_{\ell^*}^{\beta_0^*}}g(x;h)-\frac{\partial^{\beta_0^*}}{\partial x_{\ell^*}^{\beta_0^*}}g(y;h)\Big|}{\|x-y\|^{\beta^*-\beta_0^*}} 
        \\
        =
        \frac{Lh^{\beta^* - \beta_0^*}}{d\|\omega\|_{C_{\beta^*}}}\cdot\frac{\bigl|\omega^{(\beta_0)}\bigl(\frac{(x-x_0)_{\ell^*}}{h}\bigr) - \omega^{(\beta_0)}\bigl(\frac{\{y-x_0\}_{\ell^*}}{h}\bigr)\bigr|}{\|x-y\|^{\beta^*-\beta_0^*}}
        \leq \frac{L}{d}\cdot\bigg(\frac{|(x-x_0)_{\ell^*}|}{\|x-y\|}\bigg)^{\beta^*-\beta_0^*}
        \leq L.
    \end{multline*}
    \item Take $\alpha\in\N_0^d$ with $\|\alpha\|_1=\beta_0^*$. 
    Consider first the case $\beta_0^*\geq1$. 
    If $\|\alpha\|_0 \geq 2$, then
    \begin{equation*}
        \int_{\mathcal{B}_{0}(1)} K(\nu)\nu^\alpha \int_0^1(1-t)^{\beta_0^*-1}\bigl\{\partial^\alpha g(x_0+th\nu;h)-\partial^\alpha g(x_0;h)\bigr\} \, dt\,d\nu = 0.
    \end{equation*}
    Now consider the case $\|\alpha\|_0=1$, and let $\ell^*\in[d]$ be the index with $\alpha_{\ell^*} = \beta_0^*$. For $\nu\in\supp K$ we have $\|\nu\|\leq1$, and so $|t\nu_{\ell^*}|\leq1$ for all $t \in [0,1]$ and $\psi\big(\frac{\{(x_0+th\nu)-x_0\}_{\ell^*}}{h}\big) = \psi(t\nu_{\ell^*})=1$.  Hence
     \begin{equation}\label{eq:g-exp}
        g(x;h) = \frac{L}{d\|\omega\|_{C_{\beta^*}}}\sum_{\ell=1}^d|(x-x_0)_{\ell}|^{\beta^*}\psi\bigg(\frac{(x-x_0)_{\ell}}{h}\bigg)
        =
        \frac{L}{d\|\omega\|_{C_{\beta^*}}}\|x-x_0\|_{\beta^*}^{\beta^*}
    \end{equation}
    for $x\in\bigl\{x_0+th\nu\,:\,(\nu,t)\in(\supp K)\times[0,1]\bigr\}$.  It follows that 
    \begin{equation*}
        \partial^\alpha g(x;h) = \frac{L}{d\|\omega\|_{C_{\beta^*}}}(\beta^*)_{\beta_0^*}|(x-x_0)_{\ell^*}|^{\beta^*-\beta_0^*}\sgn^{\beta_0^*}\bigl((x-x_0)_{\ell^*})\bigr),
    \end{equation*}
    again for $x\in\{x_0+th\nu\,:\,(\nu,t)\in(\supp K)\times[0,1]\}$.  In particular, for $\nu\in\supp K$ and $t\in[0,1]$,
    \begin{equation*}
        \partial^\alpha g(x_0+th\nu;h) = \frac{Lt^{\beta^*-\beta_0^*}h^{\beta^*-\beta_0^*}}{d\|\omega\|_{C_{\beta^*}}}(\beta^*)_{\beta_0^*}|\nu_{\ell^*}|^{\beta^*-\beta_0^*}\sgn^{\beta_0^*}(\nu_{\ell^*}).
    \end{equation*}
We deduce that 
    \begin{align*}
        \int_{\mathcal{B}_{0}(1)} K(\nu)\nu^\alpha & \int_0^1(1-t)^{\beta_0^*-1}\bigl\{\partial^\alpha g(x_0+th\nu)-\partial^\alpha g(x_0)\bigr\} \,dt\,d\nu
        \\
        &= \frac{L(\beta^*)_{\beta_0^*}}{d\|\omega\|_{C_{\beta^*}}}h^{\beta^*-\beta_0^*}\int_{\mathcal{B}_{0}(1)} K(\nu)|\nu_{\ell^*}|^{\beta^*} \int_0^1 t^{\beta^*-\beta_0^*}dt\, d\nu\\
        &=
        \frac{L(\beta^*)_{\beta^*_0-1}}{d\|\omega\|_{C_{\beta^*}}}h^{\beta^*-\beta_0^*} \int_{\mathcal{B}_{0}(1)} K(\nu) |\nu_{\ell^*}|^{\beta^*}\, d\nu.
    \end{align*}
    Thus
    \begin{align*}
        B&\bigl(g(\cdot\,;h),x_0,K,h\bigr) 
        \\&= h^{\beta_0^*-\beta^*} \sum_{\substack{\alpha\in\N_0^d:\\\|\alpha\|_1=\beta_0^*,\|\alpha\|_0=1}}\frac{\beta_0^*}{\alpha!}\int_{\mathcal{B}_{0}(1)} K(\nu)\nu^\alpha
        \int_0^1(1-t)^{\beta_0^*-1}\bigl\{\partial^\alpha g(x_0+th\nu;h)-\partial^\alpha g(x_0;h)\bigr\}\,dt\,d\nu
        \\
        &=
        \frac{L(\beta^*)_{\beta^*_0-1}}{d\|\omega\|_{C_{\beta^*}}(\beta_0^*-1)!} \sum_{\ell=1}^d \int_{\mathcal{B}_{0}(1)} K(\nu)|\nu_\ell|^{\beta^*} \, d\nu
        =
        \frac{L(\beta^*)_{\beta^*_0-1}}{d\|\omega\|_{C_{\beta^*}}(\beta_0^*-1)!} \mu_{\beta^*}(K).
    \end{align*}

    On the other hand, if $\beta_0^*=0$, then
    \begin{equation*}
    B\bigl(g(\cdot\,;h),x_0,K,h\bigr)
    =
    h^{-\beta^*}\int_{\mathcal{B}_0(1)}K(\nu) \bigl\{g(x_0+h\nu;h)-g(x_0;h)\bigr\} \, d\nu 
    =
    \frac{L}{d\|\omega\|_{C_{\beta^*}}}\mu_{\beta^*}(K).
\end{equation*}
We conclude that
\begin{equation*}
    B\bigl(g(\cdot\,;h),x_0,K,h\bigr)
    =
    \frac{a_{\beta^*}L}{d\|\omega\|_{C_{\beta^*}}}\mu_{\beta^*}(K).
\end{equation*}
\end{enumerate}

\noindent
Now, for $n \in \mathbb{N}$, define 
\[
h_n := a_n \vee \biggl(\frac{\|\omega\|_{C_{\beta^*}}}{\|\omega\|_\infty}\biggr)^{1/\beta^*} \wedge b_n,
\]
and
$\tilde{f}_n:\R^d\to\R$ by
\begin{equation*}
    \tilde{f}_n(x) := (1-\xi)f(x)+\xi\,g(x;h_n).
\end{equation*}
Take $N\in\N$ large enough that  
\[
b_n\leq \biggl(\frac{\|\omega\|_{C_{\beta^*}}}{\|\omega\|_\infty}\biggr)^{1/\beta^*} \wedge\biggl(\frac{d\|\omega\|_{C_{\beta^*}}}{\sup_{\alpha\in\N_0^d:\|\alpha\|_1\leq\beta_0^*}\|\partial^\alpha\omega\|_\infty}\biggr)^{1/(\beta^*-\beta_0^*)}
\wedge \biggl(\frac{\epsilon\|\omega\|_{C_{\beta^*}}}{2L\|\omega\|_\infty}\biggr)^{1/\beta^*} \wedge 1
\]
for all $n\geq N$.  Then, for such $n$,
\begin{equation*}
    \|\tilde{f}_n\|_\infty
    \leq
    (1-\xi)\|f\|_\infty + \xi\|g(\cdot\,;h_n)\|_\infty
    \leq
    (1-\xi) L + \xi\frac{L\|\omega\|_\infty h_n^{\beta^*}}{\|\omega\|_{C_{\beta^*}}} \leq L.
\end{equation*}
Moreover, for $\beta_0^*\geq1$ and $n \geq N$,
\begin{align*}
\sup_{\alpha\in\N_0^d:\|\alpha\|_1\leq\beta^*_0}\|\partial^\alpha \tilde{f}_n\|_\infty
        &\leq
\sup_{\alpha\in\N_0^d:\|\alpha\|_1\leq\beta^*_0}\Bigl\{
        (1-\xi)\|\partial^\alpha f\|_\infty +  \xi\|\partial^\alpha g(\cdot\,;h_n)\|_{\infty}
        \Bigr\}
        \\
        &\leq
        (1-\xi)L + \xi Lh^{\beta^*-\beta^*_0}\frac{\sup_{\alpha\in\N_0^d:\|\alpha\|_1\leq\beta_0^*}\|\partial^\alpha\omega\|_\infty}{d\|\omega\|_{C_{\beta^*}}}
        \leq L.
\end{align*}
Finally, for all $\alpha\in\N_0^d$ with $\|\alpha\|_1=\beta_0^*$, 
\begin{align*}
        \sup_{x\neq y}\frac{|\partial^\alpha\tilde{f}_n(x)-\partial^\alpha\tilde{f}_n(y)|}{\|x-y\|^{\beta^*-\beta_0^*}}
        &\leq
        (1-\xi)\sup_{x\neq y}\frac{|\partial^\alpha f(x)-\partial^\alpha f(y)|}{\|x-y\|^{\beta^*-\beta_0^*}}
        +
        \xi\sup_{x\neq y}\frac{|\partial^\alpha g(x;h_n)-\partial^\alpha g(y;h_n)|}{\|x-y\|^{\beta^*-\beta_0^*}}
        \\
        &\leq (1-\xi)L + \xi L = L.
\end{align*}
Hence $\tilde{f}_n\in\cH(\beta,L)$ for $n\geq N$. Moreover,
\begin{align*}
        \|\tilde{f}_n-f\|_\infty
        \leq \xi (\|f\|_\infty+\|g(\cdot\,;h_n)\|_\infty) 
        \leq \xi(L+\epsilon/2) = \epsilon,
    \end{align*}
for all $n \geq N$, so $\tilde{f}_n\in\cH_{f,\epsilon}^{\mathrm{loc}}(\beta,L)$ for such $n$.  Finally, by the linearity of the bias function~\eqref{eq:B-def} in its first argument,
\begin{equation*}
    B(\tilde{f}_n,x_0,K,h_n) = (1-\xi)B(f,x_0,K,h_n)+\xi B\bigl(g(\cdot\,;h_n),x_0,K,h_n\bigr),
\end{equation*}
and so
\begin{align*}
    \bigl|B(\tilde{f}_n,x_0,K,h_n)\bigr| &\geq \xi\bigl|B(g(\cdot\,;h_n),x_0,K,h_n)\bigr| - (1-\xi)|B(f,x_0,K,h_n)|
    \\
    &> \frac{\xi a_{\beta^*}L}{d\|\omega\|_{C_{\beta^*}}}|\mu_{\beta^*}(K)|
    -
    (1-\xi)\delta
    =
    \delta,
\end{align*}
as required.
\end{proof}

\begin{lemma}\label{lem:unif-portmanteau}
    Consider a sequence of random variables $(Z_n)_{n\in\mathbb{N}}$ that depend on parameters $P\in\cP$, $x_0\in\cX$, $h\in\cH_n$, $\lambda\in\Lambda_n$ with
    \[
    \sup_{P\in\cP}\sup_{x_0\in\cX}\sup_{(h,\lambda)\in\cH_n\times\Lambda_n}\sup_{t\in\R}\big|\Pr_P( Z_n \leq t)-\Phi(t)\big| \rightarrow 0.
    \]
    Given a real sequence $(\mu_n)$ and a positive sequence $(\sigma_n)$, define $X_n := \mu_n + \sigma_n Z_n$ and $\tilde{X}_n :=\mu_n + \sigma_n Z$, where $Z \sim N(0,1)$.   For the loss function $\ell$ given by $\ell(u):=u^2\wedge M$ for some $M>0$, we have
    \begin{equation*}
\sup_{P\in\cP}\sup_{x_0\in\cX}\sup_{(h,\lambda)\in\cH_n\times\Lambda_n}\bigl|\E_P\bigl(\ell(X_n)\bigr) - \E\bigl(\ell(\tilde{X}_n)\bigr)\bigr| \to 0.
    \end{equation*}
\end{lemma}
\begin{proof}
    As $\ell$ takes values in $[0,M]$, we have
    \begin{align*}
        \bigl|\E_P\bigl(\ell(X_n)\bigr) &-
        \E\bigl(\ell(\tilde{X}_n)\bigr)\bigr|
        =
        \biggl|\int_0^M \bigl\{\Pr_P\bigl(\ell(X_n)>t\bigr)-\Pr\bigl(\ell(\tilde{X}_n)>t\bigr)\bigr\} \, dt\biggr| \\
        &\leq
        M \sup_{t\in[0,M]} \bigl|\Pr_P\bigl(\ell(X_n)>t\bigr)-\Pr\bigl(\ell(\tilde{X}_n)>t\bigr)\bigr|
        \\
        &=
        M\sup_{t\in[0,M]}\biggl|\Pr_P\biggl(\frac{-\sqrt{t}-\mu_n}{\sigma_n}\leq Z_n\leq \frac{\sqrt{t}-\mu_n}{\sigma_n}\biggr)-\Pr\biggl(\frac{-\sqrt{t}-\mu_n}{\sigma_n}\leq Z\leq \frac{\sqrt{t}-\mu_n}{\sigma_n}\biggr)\biggr|
        \\
        &\leq
        M\sup_{t\in[0,M]}\biggl|\Pr_P\biggl(Z_n\leq \frac{\sqrt{t}-\mu_n}{\sigma_n}\biggr)-\Pr\biggl(Z\leq \frac{\sqrt{t}-\mu_n}{\sigma_n}\biggr)\biggr|
        \\
        &\qquad\qquad +
        M\sup_{t\in[0,M]}\biggl|\Pr_P\biggl(Z_n < \frac{-\sqrt{t}-\mu_n}{\sigma_n}\biggr)-\Pr\biggl(Z < \frac{-\sqrt{t}-\mu_n}{\sigma_n}\biggr)\biggr|.
    \end{align*}
    Taking suprema and the limit as $n\to\infty$ yields the required result.
\end{proof}

\begin{lemma}\label{lem:seminorm}
    Adopt the notation of Lemma~\ref{lem:supB-delta}. The function $\omega$ of the form~\eqref{eq:omega-psi} has finite $\beta^*$-H\"older seminorm.
\end{lemma}

\begin{proof}
    Since $\psi$ is infinitely differentiable, with bounded derivatives of all orders, for any $j\in\mathbb{N}$ and $\beta_{\psi}\in(0,1]$, there exists $L_{j,\beta_{\psi}}\in(0,\infty)$ such that $\psi^{(j)}\in\cH(\beta_{\psi},L_{j,\beta_{\psi}})$.  Define 
    \[
    C_\psi := \max_{j\in\{0,1,\ldots,\beta_0^*\}}\sup_{u\in\R}\bigl|\psi^{(j)}(u)\bigr| <\infty.
    \]
    By the Leibniz rule,
    \begin{equation*}
        \omega^{(\beta_0^*)}(u) = \sum_{j=0}^{\beta_0^*}{\binom{\beta_0^*}{j}} (\beta^*)_{j}|u|^{\beta^*-j}(\sgn u)^j \psi^{(\beta_0^*-j)}(u),
    \end{equation*}
    where $(x)_j:=\prod_{k=1}^{j}(x+1-k)$ denotes the $j$th falling factorial. 
    We consider three cases: $\{|x|,|y|\geq2\}$, $\{|x|<2\}$ and $\{|y|<2\}$.
\begin{enumerate}[label=(\roman*), leftmargin=*, align=left]
\item
    If $|x|,|y|\geq 2$ then  $\omega^{(\beta_0^*)}(x) = \omega^{(\beta_0^*)}(y) = 0$.

\item  If $|x|<2$ then
    \begin{align*}
        \frac{|\omega^{(\beta_0^*)}(x)-\omega^{(\beta_0^*)}(y)|}{|x-y|^{\beta^*-\beta_0^*}}
        &\leq
        \sum_{j=0}^{\beta_0^*}\binom{\beta_0^*}{j}(\beta^*)_j \frac{\big||x|^{\beta^*-j}(\sgn x)^j\psi^{(\beta_0^*-j)}(x) - |y|^{\beta^*-j}(\sgn y)^j\psi^{(\beta_0^*-j)}(y)\big|}{|x-y|^{\beta^*-\beta_0^*}}
        \\
        &\leq
        \sum_{j=0}^{\beta_0^*}\binom{\beta_0^*}{j}(\beta^*)_j \frac{|x|^{\beta^*-j}}{|x-y|^{\beta^*-\beta_0^*}}\big|\psi^{(\beta_0^*-j)}(x)-\psi^{(\beta_0^*-j)}(y)\big|
        \\
        &\hspace{1.5cm} +
        \sum_{j=0}^{\beta_0^*}\binom{\beta_0^*}{j}(\beta^*)_j \frac{\big||x|^{\beta^*-j}(\sgn x)^j-|y|^{\beta^*-j}(\sgn y)^j\big|}{|x-y|^{\beta^*-\beta_0^*}}\big|\psi^{(\beta_0^*-j)}(y)\big|
        \\
        &\leq
        \sum_{j=0}^{\beta_0^*}L_{\beta_0^*-j,\beta^*-\beta_0^*}\binom{\beta_0^*}{j}(\beta^*)_j 2^{\beta^*-j}
        \\
        &\qquad +
        C_\psi\ind_{\{|y|<2\}}\sum_{j=0}^{\beta_0^*}\binom{\beta_0^*}{j}(\beta^*)_j \frac{\big||x|^{\beta^*-j}(\sgn x)^j-|y|^{\beta^*-j}(\sgn y)^j\big|}{|x-y|^{\beta^*-\beta_0^*}}.
    \end{align*}
    For $j\in\{0,1,\ldots,\beta_0^*\}$ define $k_j(u) := |u|^{\beta^*-j}(\sgn u)^j$ and 
    \[
    \bar{L}_j := \bigl((\beta^*-j)2^{\beta^*-j-1}\bigr)\vee \|k_{\beta_0^*}\|_{C_{\beta^*-\beta_0^*}}<\infty.
    \]
    Then when $|x|<2$ and $|y|<2$, 
    \begin{equation*}
        \frac{\big||x|^{\beta^*-j}(\sgn x)^j-|y|^{\beta^*-j}(\sgn y)^j\big|}{|x-y|^{\beta^*-\beta_0^*}}
        \leq
        \bar{L}_j\frac{|x-y|\vee|x-y|^{\beta^*-\beta^*_0}}{|x-y|^{\beta^*-\beta_0^*}} \leq 4\bar{L}_j.
\end{equation*}

    Therefore for $|x|<2$,
    \begin{equation}\label{eq:beta-Holder-norm}
        \frac{|\omega^{(\beta_0^*)}(x)-\omega^{(\beta_0^*)}(y)|}{|x-y|^{\beta^*-\beta_0^*}}
        \leq
        \sum_{j=0}^{\beta_0^*}2^{\beta^*-j}L_{\beta^*-j,\beta^*-\beta_0^*}\binom{\beta_0^*}{j}(\beta^*)_j 
        +
        4C_\psi\sum_{j=0}^{\beta_0^*}\bar{L}_j\binom{\beta_0^*}{j}(\beta^*)_j 
        <\infty.
    \end{equation}
\item  By symmetry~\eqref{eq:beta-Holder-norm} also holds when $|y|<2$. 
\end{enumerate}
    
We conclude that
    \begin{equation*}
        \|\omega\|_{C_{\beta^*}}
        \leq \sum_{j=0}^{\beta_0^*}2^{\beta^*-j}L_{\beta^*-j,\beta^*-\beta_0^*}\binom{\beta_0^*}{j}(\beta^*)_j 
        +
        4C_\psi\sum_{j=0}^{\beta_0^*}\bar{L}_j\binom{\beta_0^*}{j}(\beta^*)_j
        <\infty
        ,
    \end{equation*}
    as required.
    \end{proof}
    
\begin{lemma}\label{lem:GM}
    Recall the definition of $G_M$ from~\eqref{eq:G_M}. Then
    \begin{enumerate}[label=(\roman*), leftmargin=*, align=left]
        \item For any $V,M>0$, the function $\mu\mapsto G_M(\mu,V)$ is strictly increasing on $[0,\infty)$.
        \label{lem:GM-1}
        \item For any $B_\infty>0$, there exists $M(B_{\infty}):=\frac{25}{4}B_\infty^2>0$ such that for all $\mu \in [0,B_\infty]$ and $M \geq M(B_\infty)$, the function $V \mapsto G_M(\mu,V)$ is strictly increasing on $[0,\infty)$.
        \label{lem:GM-2}
        \item Fix any $M>0$. Then $G_M(\mu,V)\geq0$, with equality if and only if $(\mu,V)= (0,0)$.
        \label{lem:GM-3}
    \end{enumerate}
\end{lemma}
\begin{proof}
    \ref{lem:GM-1} We claim that non-central chi-squared distributions are stochastically ordered in the sense that for $\mu_1 < \mu_2$ and $k \in \mathbb{N}$, we have $\chi_k^2(\mu_1) \leq_{\mathrm{st}} \chi_k^2(\mu_2)$.  To see this, let $Z_1 \sim \mathrm{Poi}(\mu_1/2)$ and $Z_1' \sim \mathrm{Poi}\bigl((\mu_2 - \mu_1)/2\bigr)$ be independent and let $Z_2 := Z_1 + Z_1'$.  Now let $(E_n)$ be a sequence of independent $\mathrm{Exp}(1)$ random variables that are independent of $(Z_1,Z_1')$, and set $S_n := \sum_{i=1}^n E_i$ for $n \in \mathbb{N}$.  Then $Y_1 := S_{k+2Z_1} \sim \chi_k^2(\mu_1)$ and $Y_2 := S_{k+2Z_2} \sim \chi_k^2(\mu_2)$, with $Y_1 \leq Y_2$ almost surely.  This establishes the claim, and since $G_M(\mu,V) = \mathbb{E}\bigl\{V \chi_1^2(\mu^2/V) \wedge M\bigr\}$, the result follows.

    \medskip
    \ref{lem:GM-2} We have
    \begin{equation*}
        \frac{\partial}{\partial V}G_M(\mu,V) = \frac{1}{2V}\E_{Z\sim N(\mu,V)}\bigg((Z^2\wedge M)\bigg\{\biggl(\frac{Z-\mu}{\sqrt{V}}\bigg)^2-1\bigg\}\bigg)
        = \frac{1}{2V}\E_{\varepsilon\sim N(0,1)}\Bigl(\bigl\{(\mu+\sqrt{V}\varepsilon)^2\wedge M\bigr\}(\varepsilon^2-1)\Bigr).
    \end{equation*}
    Fix $\mu\in[0,B_\infty]$ and $M \geq M(B_\infty)$, so that $\mu/\sqrt{M} \in[0,2/5]$. 
    Now
    \begin{align*}
        \E_{\varepsilon\sim N(0,1)}\Bigl(\bigl\{(\mu+\sqrt{V}\varepsilon)^2\wedge M\bigr\}(\varepsilon^2-1)\Bigr)
        &=
        \E_{\varepsilon\sim N(0,1)}\Bigl(\min\bigl\{(\mu+\sqrt{V}\varepsilon)^2-M,0\bigr\}(\varepsilon^2-1)\Bigr)
        \\
        &=
        M\,\E_{\varepsilon\sim N(0,1)}\biggl(\min\biggl\{\biggl(\frac{\mu}{\sqrt{M}}+\sqrt{\frac{V}{M}}\varepsilon\biggr)^2-1,0\biggr\}(\varepsilon^2-1)\biggr),
    \end{align*}
    so it suffices to show that for all $(x,y)\in(0,\infty)\times[0,2/5]$, we have
    \begin{equation}\label{eq:suffices-GM}
        \E_{\varepsilon\sim N(0,1)}\Bigl(\min\bigl\{(y+x\varepsilon)^2-1,0\bigr\}(\varepsilon^2-1)\Bigr) > 0.
    \end{equation}
    By Stein's lemma~\citep[Lemma~1]{stein}, and writing $\phi$ and $\Phi$ denote the standard normal density and distribution functions, 
    \begin{align*}
        \E_{\varepsilon\sim N(0,1)}\Bigl(\min\bigl\{(y+x\varepsilon)^2 & -1,0\bigr\}(\varepsilon^2-1)\Bigr)
        \\
        &=
        2x\,\E_{\varepsilon\sim N(0,1)} \Bigl(\ind_{\{|y+x\varepsilon| < 1\}}\varepsilon(y+x\varepsilon)\Bigr)
        \\
        &= \int_{\frac{-1-y}{x}}^{\frac{1-y}{x}}\varepsilon (y+x\varepsilon)\phi(\varepsilon) \, d\varepsilon
        \\
        &= -\phi\bigg(\frac{1-y}{x}\bigg) - \phi\bigg(\frac{-1-y}{x}\bigg) + x\bigg\{\Phi\bigg(\frac{1-y}{x}\bigg) - \Phi\bigg(\frac{-1-y}{x}\bigg)\bigg\}
        \\
        &= x\bigg\{
        -\frac{1}{x}\phi\bigg(\frac{1-y}{x}\bigg) - \frac{1}{x}\phi\bigg(\frac{1+y}{x}\bigg) + \Phi\bigg(\frac{1-y}{x}\bigg) + \Phi\bigg(\frac{1+y}{x}\bigg)-1
        \bigg\}
        \\
        &= 
        xF_y\bigg(\frac{1}{x}\bigg)
        ,
    \end{align*}
    where
    \begin{equation*}
        F_y(x) := \Phi\big((1+y)x\big) + \Phi\big((1-y)x\big) - x\phi\big((1+y)x\big) - x\phi\big((1-y)x\big) - 1.
    \end{equation*}
    We have $F_y(0)=2\Phi(0)-1=0$, and $F_y$ is continuous for each $y\in[0,2/5]$.  Moreover,
    \begin{align*}
        F_y'(x) &=
        (1+y)\phi\big((1+y)x\big) + (1-y)\phi\big((1-y)x\big)
        - \phi\big((1+y)x\big) - x(1+y)\phi'\big((1+y)x\big)
        \\
        &\qquad
        - \phi\bigl((1-y)x\bigr) - x(1-y)\phi'\bigl((1-y)x\bigr)
        \\
        &= \bigl(y+(1+y)^2x^2\bigr)\phi\big((1+y)x\big) + \bigl(-y+(1-y)^2x^2\bigr)\phi\big((1-y)x\big)
        \\
        &=: \frac{1}{\sqrt{2\pi}}\exp\biggl(-\frac{(1+y)^2x^2}{2}\biggr)G_y(x^2)
        ,
    \end{align*}
    where
    \begin{equation*}
        G_y(\nu) :=  \bigl(y+(1+y)^2\nu\bigr) + \bigl(-y+(1-y)^2\nu\bigr)e^{2y\nu}.
    \end{equation*}
    Now $G_y(0)=0$, and $G_y$ is continuous for each $y\in[0,2/5]$. Moreover,
    \begin{align*}
        G_y'(\nu) &=
        (1+y)^2 + \big((1-y)^2-2y^2+2y(1-y)^2\nu\big)e^{2y\nu}
        \\
        &=\underbrace{(1+y)^2}_{\geq1} + \underbrace{\big(1-2y-y^2\big)}_{\geq 1/25\text{ for }y\in[0,2/5]}\,\underbrace{e^{2y\nu}}_{\geq1} + \underbrace{2y(1-y)^2}_{\geq0}\,\underbrace{\nu e^{2y\nu}}_{\geq0} \geq \frac{26}{25} > 0.
    \end{align*}
    Therefore $G_y(\nu)>0$ for all $(\nu,y)\in(0,\infty)\times[0,2/5]$, so $F_y(x)>0$ for all $(x,y)\in(0,\infty)\times[0,2/5]$, and so~\eqref{eq:suffices-GM} holds, from which~\ref{lem:GM-2} follows.
    
    \medskip
    \ref{lem:GM-3} The fact that $G_M(\mu,V) \geq 0$ follows immediately from the definition. If $V>0$ then $G_M(\mu,V)=\E_{Z\sim N(\mu,V)}(Z^2\wedge M) >0$, as $z^2\wedge M>0$ for $z\neq0$. Finally $G_M(\mu,0)=\mu^2\wedge M=0$ if and only if $\mu=0$.
\end{proof}

\begin{lemma}\label{lem:score-min}
    Fix $x_0 \in \mathcal{X}$ and assume that the conditional density $p(\cdot \given x_0)$ of $\varepsilon\given X = x_0$ is absolutely continuous with respect to Lebesgue measure with $i_P(x_0)<\infty$. Then
    $V_{P,x_0}^{(\infty)}(\cdot)$ is minimised over absolutely continuous functions $\varrho(\cdot \given x_0)$ satisfying $\E_P\bigl\{\varrho^2(\varepsilon\given X)\biggiven X=x_0\bigr\}<\infty$ by $\varrho(\cdot \given x_0)=\rho(\cdot \given x_0)$, where $\rho(\cdot \given x_0) := p'(\cdot \given x_0)/p(\cdot \given x_0)$ is the conditional score function. 
\end{lemma}
\begin{proof}
The fact that $p(\cdot\given x_0)$ is absolutely continuous means that  $\lim_{\varepsilon\to\pm\infty}p(\varepsilon\given x_0)=0$.  Fix $t_0\in\R$, and define $g(\varepsilon,t):=\ind_{\{\varepsilon<t\leq t_0\}}-\ind_{\{t_0<t\leq\varepsilon\}}$ for $\varepsilon,t\in\R$. Then $\int_\R g(\varepsilon,t)\varrho'(t\given x_0)\,dt = \varrho(t_0\given x_0)-\varrho(\varepsilon\given x_0)$ for all $\varepsilon\in\R$. Moreover $\int_\R p'(\varepsilon\given x_0)g(\varepsilon,t)\,d\varepsilon=p(t\given x_0)$ for all $t\in\R$.  By Cauchy--Schwarz, 
\[
\bigl[\E_P\bigl\{\varrho(\varepsilon \given X)\rho(\varepsilon\given X)\given X=x_0\bigr\}\bigr]^2 \leq \E_P\bigl\{\varrho^2(\varepsilon\given X)\biggiven X=x_0\bigr\} \cdot i_P(x_0) < \infty.
\]
Hence, applying Fubini's theorem,
\begin{align}
    \E_P\bigl\{\varrho(\varepsilon&\given X)\rho(\varepsilon\given X)\given X=x_0\bigr\}
    =
    \int_{\R}\varrho(\varepsilon\given x_0)p'(\varepsilon\given x_0)\,d\varepsilon
    \notag
    \\
    &=
    \varrho(t_0\given x_0)\int_{\R}p'(\varepsilon\given x_0)\,d\varepsilon - \int_{\R}p'(\varepsilon\given x_0)\int_{\R}g(\varepsilon,t)\varrho'(t\given x_0)\,dt\,d\varepsilon 
    \notag
    \\
    &= - \int_{\R} \varrho'(t\given x_0) \int_{\R} p'(\varepsilon\given x_0)g(\varepsilon,t)\,d\varepsilon\,dt
    = - \E_P\bigl\{\varrho'(\varepsilon\given X)\given X)\given X=x_0\bigr\}.
    \label{Eq:FITwoForms}
\end{align}
Therefore
\begin{align*}
        V_{P,x_0}^{(\infty)}(\varrho)
        =
        \frac{\E_P\bigl\{\varrho^2(\varepsilon\given X)\given X=x_0\bigr\}}{\bigl\{\E_P\bigl(\varrho'(\varepsilon\given X)\given X=x_0\bigr)\bigr\}^2}
        &=
        \frac{\E_P\bigl\{\varrho^2(\varepsilon\given X)\given X=x_0\bigr\}}{\bigl\{\E_P\bigl(\varrho(\varepsilon\given X)\rho(\varepsilon\given X)\given X=x_0\bigr)\bigr\}^2} 
        \geq
        \frac{1}{\E_P\bigl\{\rho^2(\varepsilon\given X)\given X=x_0\bigr\}},
    \end{align*}
    where the inequality follows by Cauchy--Schwarz, and equality holds if and only if $\varrho(\cdot\given x_0)\propto\rho(\cdot\given x_0)$.
\end{proof}

\begin{lemma}\label{lem:gauss-eq}
    Adopt the setup of Lemma~\ref{lem:score-min}, and assume that $\lim_{\varepsilon \rightarrow \pm \infty}\varepsilon p(\varepsilon \given x_0) = 0$.  Then
    \begin{equation*}
        \sup_{n\in\N}\sup_{\lambda\in\Lambda_n}V_{P,x_0}^{(\lambda)}(\rho) \leq \sigma_P^2(x_0),
    \end{equation*}
    with equality if and only if $\rho(\cdot\given x_0)$ is linear. 
\end{lemma}
\begin{proof}
    Recall the definition
    \begin{align*}
        V_{P,x_0}^{(\lambda)}(\rho)
        &=
        V_{P,x_0}^{(\infty)}(\rho)
        +
        \bigl(\sigma_P^2(x_0)- V_{P,x_0}^{(\infty)}(\rho)\bigr)\frac{R_2(\kappa_\lambda)}{R_2(K)}
        \\
        &=
        \sigma_P^2(x_0) - \biggl(1-\frac{R_2(\kappa_\lambda)}{R_2(K)}\biggr)\bigl(\sigma_P^2(x_0)- V_{P,x_0}^{(\infty)}(\rho)\bigr).
    \end{align*}
    Define $\mathrm{id}:\R\times\R^d\to\R$ by $\mathrm{id}(\varepsilon\given x) := \varepsilon$.  Then, by Lemma~\ref{lem:score-min},  $V_{P,x_0}^{(\infty)}(\rho)\leq V_{P,x_0}^{(\infty)}(\mathrm{id})=\sigma_P^2(x_0)$, and so
    \begin{equation*}
        \sup_{n\in\N}\sup_{\lambda\in\Lambda_n}V_{P,x_0}^{(\lambda)}(\rho)
        =
        \sigma_P^2(x_0) - \biggl(1-\frac{\sup_{n\in\N}\sup_{\lambda\in\Lambda_n}R_2(\kappa_\lambda)}{R_2(K)}\biggr)\bigl(\sigma_P^2(x_0)- V_{P,x_0}^{(\infty)}(\rho)\bigr)
        \leq \sigma_P^2(x_0),
    \end{equation*}
    with equality if and only if $\rho(\cdot\given x_0)\propto\mathrm{id}(\cdot\given x_0)$. 
\end{proof}

\section{Proof of Theorem~\ref{thm:improvements}}

\subsection{Proof of Theorem~\ref{thm:improvements}{\ref{thm:improved-ratio-2}}}

    Adopt the notation of the proof of Theorem~\ref{thm:do-no-worse}.  There, the following two facts are established. First, there exists $\mu_{\max}:=C_B(K,d,\beta_0^*,L) > 0$ 
    (defined in Lemma~\ref{lem:B(f)}\ref{lem:bias3}) such that
    \begin{equation*}
        \sup_{f\in\cH(\beta,L)} \sup_{\epsilon > 0} \sup_{h\in\cH_n} \bar{B}_{f,\epsilon,x_0,K,h}\, \leq \mu_{\max}.
    \end{equation*}
    Second, for any $\eta\in(0,\frac{1}{4}C_G]$ there exists $N\in\mathbb{N}$, not depending on $(P,x_0,h,\lambda)$, such that for any $n\geq N$, $M\geq\frac{25}{4}C_B(K,d,\beta_0^*,L)^2$ and error distribution $P_{\varepsilon|x_0}$, 
    \begin{align*}
    &\quad\;\sup_{P_X\in\cP_X}\sup_{f\in\cH(\beta,L)}\sup_{(h,\lambda)\in\cH_n'\times\Lambda_n}\frac{\mathcal{R}_{n,h,P,x_0,\epsilon,M}(\hat{f}_n^{\mathrm{Outrig}})}{\mathcal{R}_{n,h,P,x_0,\epsilon,M}(\hat{f}_n^{\mathrm{LP}})}
    \\
    &\leq
    \sup_{P_X\in\cP_X}\sup_{f\in\cH(\beta,L)}\sup_{(h,\lambda)\in\cH_n'\times\Lambda_n}
    \! \! \frac{G_M\Big(\bar{B}_{f,\epsilon,x_0,K,h}\,\alpha_1\big(nh^{2\beta^*\hspace{-0.1em}+d}\big),\,
    \frac{R_2(K)V_{P,x_0}^{(\lambda)}(\rho)}{p_X(x_0)}\alpha_2\big(nh^{2\beta^*\hspace{-0.1em}+d}\big)\Big)+2\eta}{G_M\Big(\bar{B}_{f,\epsilon,x_0,K,h}\,\alpha_1\big(nh^{2\beta^*\hspace{-0.1em}+d}\big),\,
    \frac{R_2(K)\sigma_P^2(x_0)}{p_X(x_0)}\alpha_2\big(nh^{2\beta^*\hspace{-0.1em}+d}\big)\Big)-2\eta}
    .
    \end{align*}
    Write~$\gamma:=\inf_{n\in\N}\inf_{\lambda\in\Lambda_n}\bigl(1-\frac{R_2(\kappa_{\lambda})}{R_2(K)}\bigr)>0$ and $\theta := \sup_{h\in\cH_n'}nh^{2\beta^*\hspace{-0.1em}+d}<\infty$. 
    We now apply Lemma~\ref{lem:G-ratio-bound} with $\xi:=\frac{R_2(K)}{p_X(x_0)}\alpha_2\big(nh^{2\beta^*\hspace{-0.1em}+d}\big)$ so that $\xi \geq \frac{R_2(K)}{C_X}\alpha_2(\theta)=:\xi_{\min}>0$ and $\xi\leq \frac{R_2(K)}{c_X}=:\xi_{\max}$, as well as $\mu=\bar{B}_{f,\epsilon,x_0,K,h}\alpha_1\bigl(nh^{2\beta^*\hspace{-0.1em}+d}\bigr)\leq\mu_{\max}$ so that $|\mu| \leq \mu_{\max}$, and $V_1:=V_{P,x_0}^{(\lambda)}(\rho) = 1/i_P(x_0)+\frac{R_2(\kappa_\lambda)}{R_2(K)}\bigl(\sigma_P^2(x_0)-1/i_P(x_0)\bigr)$, $V_2:=\sigma_P^2(x_0)$. Then
    \begin{align}
\sup_{P_X\in\cP_X}\sup_{f\in\cH(\beta,L)} & \sup_{(h,\lambda)\in\cH_n'\times\Lambda_n}\frac{\mathcal{R}_{n,h,P,x_0,\epsilon,M}(\hat{f}_n^{\mathrm{Outrig}})}{\mathcal{R}_{n,h,P,x_0,\epsilon,M}(\hat{f}_n^{\mathrm{LP}})}
        \notag
    \\
    &\leq
    \frac{1}{1-\Delta}\sup_{\lambda\in\Lambda_n}\biggl(
    1+\frac{\xi_{\min}\gamma\bigl(\sigma_P^2(x_0)-\frac{1}{i_P(x_0)}\bigr)}{\mu_{\max}^2+\xi_{\max}V_{P,x_0}^{(\lambda)}(\rho)+2\eta}
    -\frac{2\big(1+\tfrac{1}{1-\Delta}\big)\eta}{\xi_{\min}V_{P,x_0}^{(\lambda)}(\rho)+2\eta}
    \biggr)^{-1}
    \label{eq:partway}
    \\
    &\leq
    \frac{1}{1-\Delta}\biggl(
    1+\frac{\xi_{\min}\gamma\bigl(\sigma_P^2(x_0)-\frac{1}{i_P(x_0)}\bigr)}{\mu_{\max}^2+\xi_{\max}\sigma_P^2(x_0)+2\eta}
    -\frac{2\big(1+\tfrac{1}{1-\Delta}\big)\eta}{\xi_{\min}/i_P(x_0)+2\eta}
    \biggr)^{-1}.
    \notag
    \end{align}
    Since $\Delta \in (0,1)$ and $\eta \in (0,C_G/4]$ were arbitrary,
    \begin{align}
    {r}_{\cH'}(P_{\varepsilon|x_0})\leq \biggl(1+\frac{\xi_{\min}\gamma\bigl(\sigma_P^2(x_0)-\frac{1}{i_P(x_0)}\bigr)}{\mu_{\max}^2+\xi_{\max}\sigma_P^2(x_0)}\biggr)^{-1}
    \leq 1
    ,
    \label{eq:id-rho}
    \end{align}
    with equality if and only if $\sigma_P^2(x_0)=\frac{1}{i_P(x_0)}$, which occurs only if $P_{\varepsilon|x_0}$ is Gaussian by Lemma~\ref{lem:gauss-eq}. \hfill \qed

\subsection{Proof of Theorem~\ref{thm:improvements}{\ref{thm:improved-ratio-3}}}

    Adopt the notation of the proof of Theorem~\ref{thm:do-no-worse}, as well as the quantities $\xi_{\min}, \xi_{\max}$ and $\mu_{\max}$ as in the proof of Theorem~\ref{thm:improvements}{\ref{thm:improved-ratio-2}}. 
    Given $\tau \in (3,5]$, define a distribution $P_{\varepsilon|x_0}(\tau)$ for $\tau\in(3,5]$ with density function 
    \begin{equation*}
        p_{\varepsilon|x_0}(u\given x_0;\tau) := \frac{C_\tau}{|u|^{\tau}+1},
        \qquad \text{where} \ 
        C_\tau:=\frac{\tau\sin(\pi/\tau)}{2\pi}.
    \end{equation*}
    Then
    \begin{equation}\label{eq:sigma2-tau}
        \sigma_{P(\tau)}^2(x_0)=\E_{P(\tau)}(\varepsilon^2\given X=x_0) = 2C_\tau\int_0^\infty\frac{u^2}{u^\tau+1} \, du=\frac{\sin(\pi/\tau)}{\sin(3\pi/\tau)}.
    \end{equation}
    Further, the conditional score function is 
    $$\rho(\varepsilon\given x_0;\tau)=\frac{-p'(\varepsilon\given x_0;\tau)}{p(\varepsilon\given x_0;\tau)}=\frac{\tau|\varepsilon|^{\tau-1}\sgn\varepsilon}{|\varepsilon|^\tau+1},$$
    and so
    \begin{equation}\label{eq:i-tau}
        \frac{1}{i_{P(\tau)}(x_0)}=\frac{1}{\E_{P(\tau)}\{\rho^2(\varepsilon\given x_0;\tau)\}} = \bigg(2C_\tau\int_0^\infty\frac{\tau^2u^{2(\tau-1)}}{(u^\tau+1)^3}\,du\bigg)^{-1}=\frac{2}{\tau-1}
        \in [1/2,2)
        .
    \end{equation}
    Now define $\gamma(\underbar{\lambda}):=1-\frac{R_2(\kappa_{\underbar{\lambda}})}{R_2(K)}$. 
    Arguing similarly to~\eqref{eq:partway}, we have for each $\underbar{\lambda} > \lambda_0(K)$, $\Delta \in (0,1)$ and $\eta\in(0,C_G/4]$ that
    \begin{align*}
r_{\cH'}\bigl(P_{\varepsilon|x_0}(\tau),\underbar{\lambda}\bigr)
    \leq
    \frac{1}{1-\Delta}\biggl(
    1+\frac{\xi_{\min}\gamma(\underbar{\lambda})\bigl(\sigma_P^2(x_0)-\frac{1}{i_P(x_0)}\bigr)}{\mu_{\max}^2+\xi_{\max}V_{P,x_0}^{(\underbar{\lambda})}(\rho)+2\eta}
    -\frac{2\big(1+\tfrac{1}{1-\Delta}\big)\eta}{\xi_{\min}V_{P,x_0}^{(\underbar{\lambda})}(\rho)+2\eta}
    \biggr)^{-1}.
    \end{align*}
    Since $\Delta \in (0,1)$ and $\eta \in (0,C_G/4]$ were arbitrary, we deduce from~\eqref{eq:sigma2-tau}~and~\eqref{eq:i-tau} that
    \begin{align*}
        r_{\cH'}\bigl(P_{\varepsilon|x_0}(\tau),\underbar{\lambda}\bigr) 
        &\leq
        \biggl(1+\frac{\xi_{\min}\gamma(\underbar{\lambda})\bigl(\sigma_{P(\tau)}^2(x_0)-\frac{1}{i_{P(\tau)}(x_0)}\bigr)}{\mu_{\max}^2+\xi_{\max}\bigl\{\frac{1}{i_{P(\tau)}(x_0)}+\bigl(1-\gamma(\underbar{\lambda})\bigr)\bigl(\sigma_{P(\tau)}^2(x_0)-\frac{1}{i_{P(\tau)}(x_0)}\bigr)\bigr\}}\biggr)^{-1}
        \\
        &=\biggl(1+\frac{\xi_{\min}\gamma(\underbar{\lambda})}{\mu_{\max}^2+\xi_{\max}\bigl\{\frac{2}{\tau-1}+\bigl(1-\gamma(\underbar{\lambda})\bigr)\bigl(\frac{\sin(\pi/\tau)}{\sin(3\pi/\tau)}-\frac{2}{\tau-1}\bigr)\bigr\}}\biggl\{\frac{\sin(\pi/\tau)}{\sin(3\pi/\tau)}-\frac{2}{\tau-1}\biggr\}\biggr)^{-1}.
    \end{align*}
    Since $\gamma:[\lambda_0(K),\infty)\to(0,1]$ is a strictly increasing function with $\gamma\bigl(\lambda_0(K)\bigr)=0$ and $\lim_{\underbar{\lambda}\to\infty}\gamma(\underbar{\lambda})=1$, and $\frac{\sin(\pi/\tau)}{\sin(3\pi/\tau)}-\frac{2}{\tau-1}\to \infty$ as $\tau\searrow3$, we conclude that
    \begin{equation*}
        \lim_{\tau\searrow3,\,\underbar{\lambda}\to\infty}
        r_{\cH'}\bigl(P_{\varepsilon|x_0}(\tau),\underbar{\lambda}\bigr)
        = 0,
    \end{equation*}
    so the result follows.
\hfill \qed

\subsection{Additional lemma for Theorem~\ref{thm:improvements}}

\begin{lemma}\label{lem:G-ratio-bound}
Fix $\Delta\in(0,1)$ and $\eta\in(0,C_G/4]$, as well as $0 < \xi_{\min} \leq \xi_{\max}$, $\mu_{\max} \in \mathbb{R}$ and $0 < V_{\min} \leq V_{\max}$.  Suppose that $\xi \in [\xi_{\min},\xi_{\max}]$, $\mu \in [-\mu_{\max},\mu_{\max}]$ and $V_1,V_2 \in [V_{\min},V_{\max}]$ with $V_1 \leq V_2$ and $M \geq 3(\mu_{\max}^2 + \xi_{\max}V_{\max})/\Delta$. Then
    \begin{equation}
        \frac{G_M(\mu,\xi V_1)+2\eta}{G_M(\mu,\xi V_2)-2\eta} 
        \leq \frac{1}{1-\Delta}\bigg(1+\frac{\xi_{\min}(V_2-V_1)}{\mu_{\max}^2+\xi_{\max}V_1+2\eta}-\frac{2\big(1+\tfrac{1}{1-\Delta}\big)\eta}{\xi_{\min} V_1+2\eta}\bigg)^{-1}.
    \end{equation}
\end{lemma}
\begin{proof}
    For any $\mu \in \mathbb{R}$ and $V \geq 0$, we have
    \begin{equation*}
        G_M(\mu,V)=\mu^2+V-\E_{Z\sim N(\mu,V)}\bigl\{(Z^2-M)\vee0\bigr\},
    \end{equation*}
    so $G_M(\mu,V)\leq \mu^2+V$. Further,
    \begin{align*}
        \E_{Z\sim N(\mu,V)}\bigl\{(Z^2-M)\vee0\bigr\}
        &=
        \E_{Z\sim N(\mu,V)}\bigl(Z^2\ind_{\{|Z|>\sqrt{M}\}}\bigr) - M\mathbb{P}(|Z| > \sqrt{M})
        \\ 
        &\leq
        \frac{1}{M}\E_{Z\sim N(\mu,V)}(Z^4)
        =
        \frac{\mu^4+6\mu^2V+3V^2}{M},
    \end{align*}
    so
    \begin{equation*}
        \mu^2+V-\frac{\mu^4+6\mu^2V+3V^2}{M}
        \leq G_M(\mu,V) 
        \leq \mu^2+V.
    \end{equation*}
    On the other hand,
    \begin{equation*}
        \frac{\mu^4+6\mu^2\xi V_2+3\xi^2 V_2^2}{M}
        \leq 
        \frac{3(\mu^2 + \xi V_2)^2}{M} \leq \Delta(\mu^2 + \xi V_2),
    \end{equation*}
    so
    \begin{align*}
        G_M(\mu,\xi V_2) &\geq
        \mu^2+\xi V_2 - \frac{\mu^4+6\mu^2\xi V_2+3\xi^2V_2^2}{M}
        \geq
        (1-\Delta)(\mu^2+\xi V_2). 
    \end{align*}
    Therefore
    \begin{align*}
        \frac{G_M(\mu,\xi V_1)+2\eta}{G_M(\mu,\xi V_2)-2\eta}
        \leq
        \frac{1}{1 \!- \! \Delta}\biggl(\frac{\mu^2+\xi V_2-\tfrac{2\eta}{1-\Delta}}{\mu^2+\xi V_1+2\eta}\biggr)^{-1}
        &=
        \frac{1}{1 \!-\! \Delta}\biggl(1+\frac{\xi(V_2-V_1)}{\mu^2+\xi V_1+2\eta}-\frac{2\big(1+\tfrac{1}{1-\Delta}\big)\eta}{\mu^2+\xi V_1+2\eta}\biggr)^{-1}
        \\
        &\leq
        \frac{1}{1 \!-\! \Delta}\biggl(1+\frac{\xi_{\min}(V_2-V_1)}{\mu_{\max}^2+\xi_{\max}V_1+2\eta}-\frac{2\big(1+\tfrac{1}{1-\Delta}\big)\eta}{\xi_{\min} V_1+2\eta}\biggr)^{-1}
        ,
    \end{align*}
    as required.
\end{proof}

\section{Proofs of minimaxity results of Section~\ref{sec:minimax}}

    \begin{proof}[Proof of Theorem~\ref{thm:UB}]
        We take $K$ to be a bounded kernel of order $p+1$ supported on $\mathcal{B}_0(1)$, and define
        \begin{align*}
            \bar{\mu}_\beta(K)
            &:= \int_{\mathcal{B}_0(1)}|K(\nu)|\,\|\nu\|_1^{\beta_0}\|\nu\|^{\beta-\beta_0}\,d\nu, 
            \qquad
            a :=\frac{d\,\Gamma^2(\beta+1)R_2(K)}{2\beta 
            \Gamma^2(\beta-\beta_0+1)
            L^2\bar\mu_\beta^2(K) i_P(x_0)p_X(x_0)}, \\
            h \equiv h_n &:= 
            a^{1/(2\beta+d)}n^{-1/(2\beta+d)}.
    \end{align*} 
        Consider the decomposition of Theorem~\ref{thm:decomp}. 
        In the case $\beta\in(0,1]$, the bias term~\eqref{eq:B-def-beta01} satisfies
       \begin{align*}
            |B(f,x_0,K,h)|\,h^{\beta}
            &\leq
            \int_{\mathcal{B}_0(1)}|K(\nu)|\,|f(x_0+h\nu)-f(x_0)|\,d\nu
            \leq
            L\int_{\mathcal{B}_0(1)}|K(\nu)|\,\|h\nu\|^\beta\,d\nu
            \leq 
            Lh^{\beta}\bar{\mu}_{\beta}(K).
        \end{align*}
        On the other hand, for $\beta>1$ we have from~\eqref{eq:multinom-expansion} that
        \[
        \sum_{\substack{\alpha\in\N_0^d:\\\|\alpha\|_1=\beta_0}}\frac{1}{\alpha!} \int_{\mathcal{B}_0(1)}|K(\nu)|\,|\nu^\alpha|\,\|\nu\|^{\beta-\beta_0}\,d\nu
        =
        \frac{1}{\beta_0!}\int_{\mathcal{B}_0(1)}|K(\nu)|\,\|\nu\|_1^{\beta_0}\|\nu\|^{\beta-\beta_0}\,d\nu
        =\frac{\bar{\mu}_\beta(K)}{\beta_0!}.
        \]
        Thus, the bias term~\eqref{eq:B-def} satisfies
        \begin{align*}
            |B(f,x_0,K,h)|\,h^{\beta}
            &\leq
            h^{\beta_0} \sum_{\substack{\alpha\in\N_0^d:\\\|\alpha\|_1=\beta_0}}\frac{\beta_0}{\alpha!}\int_{\mathcal{B}_{0}(1)} |K(\nu)|\,|\nu^\alpha| 
            \,\int_0^1(1-t)^{\beta_0-1}\bigl|\partial^\alpha f(x_0+th\nu)-\partial^\alpha f(x_0)\bigr| \,dt\,d\nu
            \\
            &\leq
            L h^\beta  
            \int_0^1 t^{\beta-\beta_0} (1-t)^{\beta_0-1}\,dt
            \sum_{\substack{\alpha\in\N_0^d:\\\|\alpha\|_1=\beta_0}}\frac{\beta_0}{\alpha!} \int_{\mathcal{B}_0(1)}|K(\nu)|\,|\nu^\alpha|\,\|\nu\|^{\beta-\beta_0}\,d\nu
            \\
            &
            =
            Lh^{\beta}\frac{\Gamma(\beta-\beta_0+1)}{\Gamma(\beta+1)}\bar{\mu}_\beta(K). 
        \end{align*}

        Using the fact that $V_{P,x_0}^{(\infty)}(\rho)= 1/i_P(x_0)$, and arguing as in the proof of Theorem~\ref{thm:do-no-worse}, we therefore have
        \begin{align*}
            \limsup_{M\to\infty}&\limsup_{n\to\infty}\mathrm{MSE}_{n,M}(\hat{f}^{\mathrm{Outrig}}_n)
            \leq
            L^2 a^{2\beta/(2\beta+d)}\frac{\Gamma^2(\beta-\beta_0+1)}{\Gamma^2(\beta+1)}\bar\mu_\beta^{2}(K)
            +
            \frac{R_2(K)}{ i_P(x_0)p_X(x_0)\,a^{d/(2\beta+d)}}
            \\
            &=
            (2\beta+d)\biggl(\frac{1}{(2\beta)^{2\beta}d^d}
            \frac{\Gamma^{2d}(\beta-\beta_0+1)}{\Gamma^{2d}(\beta+1)}
            \bar\mu_\beta^{2d}(K)R_2^{2\beta}(K)
            \biggr)^{1/(2\beta+d)}
            \biggl(\frac{L^{d/\beta}}{i_P(x_0) p_X(x_0)}\biggr)^{2\beta/(2\beta+d)}.
        \end{align*}
        This establishes~\eqref{eq:LAM-UB}, and it remains to bound $C_{\beta,1}$.  Now 
        \begin{equation*}
            \bar\mu_\beta(K)
            =
        \int_{\mathcal{B}_0(1)}|K(\nu)|\,\|\nu\|_1^{\beta_0}\|\nu\|^{\beta-\beta_0}\,d\nu
            \leq
            d^{\beta_0(q-1)/q}\int_{\mathcal{B}_0(1)}|K(\nu)|\|\nu\|^{\beta}\,d\nu
            =: d^{\beta_0(q-1)/q}\tilde{\mu}_\beta(K)
            .
        \end{equation*}
Henceforth we take our kernel $K$ to be 
\begin{equation*}
            K(\nu) := \frac{\beta+d}{\beta}
\frac{\Gamma(1+d/q)}{2^d\Gamma^d(1+1/q)}\max\bigl(0,\,1-\|\nu\|^\beta\bigr),
        \end{equation*}
        so that
        \begin{equation*}
            R_2(K) = \frac{2(\beta+d)}{2\beta+d}\cdot \frac{\Gamma(1+d/q)}{2^d \Gamma^d(1+1/q)},
            \qquad
            \tilde{\mu}_{\beta}(K) = \frac{d}{2\beta+d}.
        \end{equation*}
        Hence
        \begin{equation*}
            \tilde{\mu}_\beta^d(K)R_2^\beta(K)
            =
            \frac{d^d(\beta+d)^\beta}{2^{\beta(d-1)}(2\beta+d)^{\beta+d}}\cdot\frac{\Gamma^{\beta}(1+d/q)}{\Gamma^{\beta d}(1+d/q)}.
        \end{equation*}
        Moreover, for $\beta\in(0,2]$, we have $\frac{\Gamma(\beta+1)}{\Gamma(\beta-\beta_0+1)}=1\vee\beta$, 
        and so~\eqref{eq:LAM-UB} holds with
        \begin{align}\label{eq:C-bd-(0,2]}
            C_{\beta,d}
            &=
            \biggl\{
            \frac{d^{(2\beta_0(q-1)/q+1)d}(\beta+d)^{2\beta}}{(2^d\beta)^{2\beta}(2\beta+d)^{d}(1\vee\beta)^{2d}}\cdot\frac{\Gamma^{2\beta}(1+d/q)}{\Gamma^{2\beta d}(1+1/q)}\biggr\}^{1/(2\beta+d)},
        \end{align}
        and in the special case $d=1$, this reduces to~\eqref{Eq:Cbetad}.
    \end{proof}

    For our corresponding lower bound we will make use of the following lemma. Our explicit constants involve the hypergeometric function $_2F_1$, where, for $a \in \mathbb{R}$, $c > b > 0$ and $z \in \mathbb{R}$, we have
    \[
    _2F_1(a,b,c;z) := \frac{1}{\mathrm{B}(b,c-b)} \int_0^1 x^{b-1}(1-x)^{c-b-1}(1-zx)^{-a} \, dx,
    \]
    where $\mathrm{B}(\cdot,\cdot)$ denotes the beta function.  For more details on hypergeometric functions, see e.g.,~\citet[Chapter 9]{hypergeometric}. We also introduce the following notation. Given a function $K:\R^d\to\R$ we define the $\beta$-H\"older semi-norm
    \begin{equation*}
        \|K\|_{C_{\beta}}:=\sup_{\alpha\in\N_0^d\,:\,\|\alpha\|_1=\beta_0}\sup_{x\neq y}\frac{|\partial^\alpha K(x)-\partial^\alpha K(y)|}{\|x-y\|^{\beta-\beta_0}},
    \end{equation*}
    where $\beta_0:=\ceil{\beta}-1$.

    \begin{lemma}\label{lem:LB-K}
        Let $x_0\in\cX$ and $\beta \in (0,\infty)$.  
        Fix an arbitrary function $K:\R^d\to\R$ with $\|K\|_{C_\beta}<\infty$, as well as $\sup_{\alpha\in\N_0^d\,:\,\|\alpha\|_1\leq\beta_0}\|\partial^\alpha K\|_\infty<\infty$. Then for any sequence $(\hat{f}_n)$ of estimators,
        \begin{equation}\label{eq:LB-cbetaK}
            \liminf_{n\to\infty} \mathrm{MSE}_n(\hat{f}_n)
            \geq
            c_{\beta,d}(K) \biggl(\frac{L^{d/\beta}}{  p_X(x_0) i_P(x_0) }\biggr)^{2\beta/(2\beta+d)},
        \end{equation}
        where
        \begin{equation*}
            c_{\beta,d}(K)
            :=
            a_{\beta,d}\cdot\biggl(\frac{(2\beta)^{2\beta}d^d}{(2\beta+d)^{2\beta+d}}\biggr)^{1/(2\beta+d)}\frac{K^2(0)}{\|K\|_{C_\beta}^{2d/(2\beta+d)}R_2(K)^{2\beta/(2\beta+d)}},
        \end{equation*}
        and $a_{\beta,d} > 0$ satisfies
        \begin{equation}\label{eq:a}
        \frac{1}{1.69} \leq \inf_{\beta \in (0,\infty)} \inf_{d \in \mathbb{N}} a_{\beta,d} \leq \lim_{\beta \searrow 0} \inf_{d \in \mathbb{N}} a_{\beta,d} = 1.
        \end{equation}
    \end{lemma}

\begin{proof}[Proof of Lemma~\ref{lem:LB-K}]
        For $t \in [-1,1]$, define $f_t:\mathbb{R}^d \rightarrow \mathbb{R}$ by
        \begin{equation}\label{eq:ft}
        f_t(x) := \frac{tLh^\beta}{\|K\|_{C_{\beta}}}K\Bigl(\frac{x-x_0}{h}\Bigr),
    \end{equation}
    with
    \begin{equation}\label{eq:h}
        h := \biggl(\frac{\lambda \|K\|_{C_\beta}^2}{L^2 R_2(K)p_X(x_0)i_P(x_0)n}\biggr)^{1/(2\beta+d)}
    \end{equation}
    for some $\lambda > 0$.  Then for all $\alpha\in\N_0^d$ with $\|\alpha\|_1\leq\beta_0$,
    \begin{equation*}
        \|\partial^\alpha f_t\|_\infty
        \leq \frac{L\|\partial^\alpha K\|_\infty h^{\beta-\|\alpha\|_1}}{\|K\|_{C_\beta}}
        \to 0
    \end{equation*}
    as $n\to\infty$. 
    Moreover,
    \begin{equation*}
        \sup_{\substack{\alpha\in\N_0^d:\\\|\alpha\|_1=\beta_0}}\sup_{x\neq y}\frac{|\partial^\alpha f_t(x)-\partial^\alpha f_t(y)|}{\|x-y\|^{\beta-\beta_0}}
        \leq
        \frac{Lh^{\beta-\beta_0}}{\|K\|_{C_\beta}}\sup_{\substack{\alpha\in\N_0^d:\\\|\alpha\|_1=\beta_0}}\sup_{x\neq y}\frac{|\partial^\alpha K(\frac{x-y}{h})-\partial^\alpha K(\frac{y-x_0}{h})|}{\|x-y\|^{\beta-\beta_0}}
        = L.
    \end{equation*}
    Hence there exists $N \in \mathbb{N}$ such that $f_t\in\cH(\beta,L)$ for all $t \in [-1,1]$ and $n \geq N$.  Let $P_t$ denote the joint distribution of independent and identically distributed pairs $(X_1,Y_1),\ldots,(X_n,Y_n)$ satisfying $X_1 \sim P_X$ and  
    \begin{equation}\label{eq:LB-model}
        Y_1 = f_t(X_1) + \varepsilon_1,
    \end{equation}
    where $\varepsilon_1 \given X_1 \sim P_{\varepsilon|X}$.  Suppose further that $P_t:=(P_X,P_{\varepsilon|X},f_t) \in \mathcal{P}$.  The Fisher information $I(t)$ of $P_t$ is given by
    \begin{align*}
        I(t) &= n\int_{\R^d}\int_\R\biggl(\frac{\partial}{\partial t}\log\bigl\{p_{\varepsilon|X}\bigl(y-f_t(x)\given x\bigr)p_X(x)\bigr\}\biggr)^2p_{\varepsilon|X}\bigl(y-f_t(x)\given x\bigr)p_X(x)\,dy\,dx
        \\
        &=
        \frac{nL^2h^{2\beta}}{\|K\|_{C_\beta}^2}\int_{\R^d}\int_\R\frac{\bigl\{p_{\varepsilon|X}'\bigl(y-f_t(x)\given x\bigr)\bigr\}^2}{p_{\varepsilon|X}\bigl(y-f_t(x)\given x\bigr)}K^2\Bigl(\frac{x-x_0}{h}\Bigr)p_X(x)\,dy\,dx
        \\
        &=
        \frac{nL^2h^{2\beta}}{\|K\|_{C_\beta}^2}\int_{\R^d}\E_{P_{\varepsilon|X}}\bigl\{\rho^2(\varepsilon\given X)\biggiven X=x\bigr\}K^2\Bigl(\frac{x-x_0}{h}\Bigr)p_X(x)\,dx
        \\
        &=
        \frac{nL^2h^{2\beta+d}}{\|K\|_{C_\beta}^2}\int_{\R^d}\E_{P_{\varepsilon|X}}\bigl\{\rho^2(\varepsilon\given X)\biggiven X=x_0+h\nu\bigr\}K^2(\nu)p_X(x_0+h\nu)\,d\nu
        \\
        &\to
        \lambda 
        ,
    \end{align*}
    as $n\to\infty$ by the dominated convergence theorem, where we recall that $i_P(x_0):=\E_P\bigl(\rho^2(\varepsilon\given X)\given X=x_0\bigr)$. Therefore for all $\epsilon'\in(0,1)$ there exists $N_{\epsilon'} \in \mathbb{N}$ such that 
    \begin{equation*}
        I(t)
        \geq (1-\epsilon')\lambda
    \end{equation*}
    for all $n\geq N_{\epsilon'}$ and $t\in[-1,1]$. 
    Define $\hat{t}_n:=\frac{\|K\|_{C_\beta}}{Lh^\beta K(0)}\hat{f}_n(x_0)$. Then for $n\geq N\vee N_{\epsilon'}$,
    \begin{align*}
        &\sup_{\substack{f\in\cH(\beta,L)}}\E_P\bigl\{\bigl(\hat{f}_n(x_0)-f(x_0)\bigr)^2\bigr\}
        \geq
    \sup_{t\in[-1,1]}\E_{P_t}\bigl\{\bigl(\hat{f}_n(x_0)-f_t(x_0)\bigr)^2\bigr\}
        \\
        &=
        \frac{L^2h^{2\beta}K^2(0)}{\|K\|_{C_\beta}^2}\sup_{t\in[-1,1]}\E_{P_t}\bigl\{(\hat{t}_n-t)^2\bigr\}
        \\
        &=
        \frac{K^2(0)}{\bigl(\|K\|_{C_\beta}^{2d}R_2(K)^{2\beta}\bigr)^{1/(2\beta+d)}} \biggl(\frac{L^{d/\beta}}{  p_X(x_0) i_P(x_0) n}\biggr)^{2\beta/(2\beta+d)}
         \lambda^{2\beta/(2\beta+d)}\sup_{t\in[-1,1]}\E_{P_t}\bigl\{(\hat{t}_n-t)^2\bigr\}
         \\
         &\geq
         \frac{K^2(0)}{\bigl(\|K\|_{C_\beta}^{2d}R_2(K)^{2\beta}\bigr)^{1/(2\beta+d)}} \biggl(\frac{L^{d/\beta}}{  p_X(x_0) i_P(x_0) n}\biggr)^{2\beta/(2\beta+d)}
        \sup_{m>0}\frac{\lambda^{2\beta/(2\beta+d)}}{(m+1)^2\{_2F_1\bigl(-\frac{1}{2}, \frac{m}{2}, \frac{m}{2}+1; -\frac{(1-\epsilon')\lambda}{m^2}\bigr)\}^2},
    \end{align*}
    by the augmented van Trees inequality~\citep{young2026} (restated as Theorem~\ref{thm:avt} for convenience) with $\alpha(t) = (1 - |t|)^m$ and 
    \[
    \pi(t)=\frac{\sqrt{\lambda \alpha^2(t)+\alpha'(t)^2}}{\int_{-1}^1\sqrt{\lambda \alpha^2(\tau)+ \alpha'(\tau)^2}\,d\tau}.
    \]
    Since $\epsilon' \in (0,1)$ and $\lambda > 0$ were arbitrary, we deduce that 
    \begin{align*}
& \liminf_{n\to\infty} \mathrm{MSE}_n(\hat{f}_n)
        \\
        &\hspace{1cm} \geq
        g\biggl(\frac{2\beta}{2\beta+d}\biggr)
        \biggl(\frac{(2\beta)^{2\beta}d^d}{(2\beta+d)^{2\beta+d}}\biggr)^{1/(2\beta+d)}
        \frac{K^2(0)}{\bigl(\|K\|_{C_\beta}^{2d}R_2(K)^{2\beta}\bigr)^{1/(2\beta+d)}} \biggl(\frac{L^{d/\beta}}{  p_X(x_0) i_P(x_0)}\biggr)^{2\beta/(2\beta+d)},
    \end{align*}
    where
    \begin{equation*}
        g(\gamma) := \frac{1}{\gamma^\gamma(1-\gamma)^{1-\gamma}}\sup_{m>0}\sup_{\lambda>0}\frac{\lambda^\gamma}{(m+1)^2\bigl\{_2F_1\bigl(-\frac{1}{2}, \frac{m}{2}, \frac{m}{2}+1; -\frac{\lambda}{m^2}\bigr)\bigr\}^2}.
    \end{equation*}
    Now
    \begin{align*}
        g(\gamma)
        &=
        \frac{1}{\gamma^\gamma(1-\gamma)^{1-\gamma}}\sup_{m>0}\sup_{\lambda>0}\frac{4\lambda^\gamma}{(m+1)^2 \bigl(\int_0^1 t^{\frac{m}{2}-1}\sqrt{m^2+\lambda t}\,dt\bigr)^2}
        \\
        &\geq
        \frac{1}{\gamma^\gamma(1-\gamma)^{1-\gamma}}\sup_{m>0}\sup_{\lambda>0}\frac{m^2\lambda^\gamma}{(m+1)^2 (m^2+\lambda)} = \sup_{m>0} \frac{m^{2\gamma}}{(m+1)^2}
        = \gamma^{2\gamma}(1-\gamma)^{2(1-\gamma)}.
    \end{align*}
    This lower bound is convex in $\gamma$ and symmetric about $1/2$, so     \begin{equation*}
        \inf_{\gamma\in(0,1)}g(\gamma) \geq \frac{1}{4},
    \end{equation*}
which establishes~\eqref{eq:LB-cbetaK} with the stated expression for $c_{\beta,d}(K)$ and $a_{\beta,d} := g\bigl(\frac{2\beta}{2\beta+d}\bigr)$.   In fact, it can be verified numerically that
    \begin{equation*}
         \inf_{\gamma\in (0,1)}g(\gamma) \geq \frac{1}{1.69}.
    \end{equation*}
It remains to show the equality in~\eqref{Eq:Cratio}. Take $\delta\in(0,1)$.  Then, using the inequality $\sqrt{m^2+\lambda t\,}\geq (1-\delta)m+\delta\sqrt{\lambda t\,}$ for $m,\lambda > 0$ and $t \in [0,1]$, we have for any $\gamma \in (0,1)$ that
\begin{align*}
    \sup_{\lambda>0}&\sup_{m>0}\frac{\lambda^\gamma}{(m+1)^2\bigl(\frac{1}{2}\int_0^1 t^{\frac{m}{2}-1}\sqrt{m^2+\lambda t}\,dt\bigr)^2}
    \leq
    \sup_{\lambda>0}\sup_{m>0}\frac{\lambda^\gamma}{(m+1)^2\bigl((1-\delta)+\frac{\delta\sqrt{\lambda}}{m+1}\bigr)^2}
    \\
    &\leq
    \sup_{\lambda>0}\sup_{m>0}\frac{\lambda^\gamma}{(1-\delta)^2(m+1)^2+\delta^2\lambda}
    =
    \sup_{\lambda>0}\frac{\lambda^\gamma}{(1-\delta)^2+\delta^2\lambda}
    = 
    \frac{\gamma^\gamma(1-\gamma)^{1-\gamma}}{\delta^{2\gamma}(1-\delta)^{2(1-\gamma)}}.
\end{align*}
Therefore
\begin{equation*}
    \gamma^{2\gamma}(1-\gamma)^{2(1-\gamma)} \leq g(\gamma) \leq \frac{1}{\delta^{2\gamma}(1-\delta)^{2(1-\gamma)}}.
\end{equation*}
In particular,
\begin{equation*}
    1\leq \lim_{\gamma\searrow0}g(\gamma) \leq \frac{1}{(1-\delta)^2}.
\end{equation*}
As $\delta\in(0,1)$ was arbitrary, 
\begin{equation*}
    \lim_{\beta\searrow0}\inf_{d\in\N}a_{\beta,d}=\lim_{\gamma\searrow0}g(\gamma) =1,
\end{equation*}
which completes the proof.
\end{proof}

\begin{theorem}[Augmented van Trees inequality, \citealp{young2026}]\label{thm:avt}
Let $(\mathcal{X},\mathcal{A},\mu)$ be a $\sigma$-finite measure space.  Assume that:
\begin{enumerate}[label=(\roman*)]
\item $p:\mathcal{X} \times [-1,1] \rightarrow [0,\infty)$ is a measurable function such that $p(\cdot,t)$ is a (Lebesgue) density for each $t \in T$, and that $t \mapsto p(x,t)$ is absolutely continuous for $\mu$-almost all $x$.  Further suppose that the Fisher information
\[
\mathcal{I}(t) := \int_{\mathcal{X}} \biggl(\frac{\partial_t p(x,t)}{p(x,t)}\biggr)^2 p(x,t) \, d\mu(x) 
\]
satisfies $\int_{-1}^1 \mathcal{I}(t) \, dt < \infty$.

\item $\pi$ is a (Lebesgue) density on $[-1,1]$.

\item $\alpha:[-1,1] \rightarrow \mathbb{R}$ is absolutely continuous with $\alpha(-1) = \alpha(1) = 0$ and
\[
\max\biggl\{\int_{-1}^1 \alpha(t) \, dt,\int_{-1}^1 \frac{\alpha^2(t)}{\pi(t)} \, dt, \int_{-1}^1 \frac{\alpha'(t)^2}{\pi(t)} \, dt\biggr\} < \infty.
\]
\end{enumerate}
Then writing $P_t$ for the distribution with density $p(\cdot,t)$, we have for any measurable $\hat{t}:\mathcal{X} \rightarrow \mathbb{R}$ that
\[
\int_{-1}^1 \mathbb{E}_{P_t}\bigl\{\bigl(\hat{t}(X) - t\bigr)^2\bigr\} \pi(t) \, dt \geq \frac{\bigl\{\int_{-1}^1 \alpha(t) \, dt\bigr\}^2}{\int_{-1}^1 \frac{\mathcal{I}(t)\alpha^2(t) + \alpha'(t)^2}{\pi(t)} \, dt}.
\]
\end{theorem}

Theorem~\ref{thm:LAM-LB} follows by applying Lemma~\ref{lem:LB-K} for a specific choice of bump function $K$. 

\begin{proof}[Proof of Theorem~\ref{thm:LAM-LB}]
The result follows from Lemma~\ref{lem:LB-K} by taking~$K(\nu)=\prod_{j=1}^d \bigl\{e^{-1/(1-\nu_j^2)}\mathbbm{1}_{\{|v_j| < 1\}}\bigr\}$.  In that case, $K(0)=e^{-d}$, $R_2(K)\leq 2^{d}e^{-2d}<\infty$ and $\|K\|_{C_\beta}$ is finite and depends only on $(\beta,d)$~\citep[e.g.][Exercise 8.13]{samworth-shah}. 
In the case $\beta\in(0,1]$, take $K = K^*$, where
    \begin{equation}\label{eq:LB-K}
        K^*(\nu):=\max\bigl(1-\|\nu\|^{\beta},0\bigr).
    \end{equation}
    Then $K^*(0)=1$. Moreover, for $u,v\in\mathcal{B}_0(1)$, 
    \[
    |K^*(u)-K^*(v)|=\bigl|\|u\|^\beta-\|v\|^\beta\bigr| \leq \bigl|\|u\|-\|v\|\bigr|^\beta \leq \|u-v\|^\beta,\]
    and if $u\in\mathcal{B}_0(1)$, $v\not\in\mathcal{B}_0(1)$, define $v_{\perp}:=\argmin_{\nu\in\mathcal{B}_0(1)}\|\nu-v\|$. Then
    \[|K^*(u)-K^*(v)|=|K^*(u)-K^*(v_{\perp})|\leq\|u-v_{\perp}\|^\beta
    \leq \|u-v\|^\beta.
    \]
    Therefore $\|K^*\|_{C_\beta}=1$.  Finally, 
    \begin{equation*}
        R_2(K^*) = \frac{2^{d+1}\beta^2}{(\beta+d)(2\beta+d)}\cdot\frac{\Gamma^d(1+1/q)}{\Gamma(1+d/q)}.
    \end{equation*}
    Thus
    \begin{align*}
        \biggl(\frac{(2\beta)^{2\beta}d^d}{(2\beta+d)^{2\beta+d}}\biggr)^{1/(2\beta+d)}\frac{K^{*2}(0)}{\bigl(\|K^*\|_{C_\beta}^{2d}R_2(K^*)^{2\beta}\bigr)^{1/(2\beta+d)}}
        &=
        \biggl(\frac{(\beta+d)^{2\beta}d^d}{(2^d\beta)^{2\beta}(2\beta+d)^d}\cdot\frac{\Gamma^{2\beta}(1+d/q)}{\Gamma^{2\beta d}(1+1/q)}\biggr)^{1/(2\beta+d)}.
    \end{align*}
    Applying Lemma~\ref{lem:LB-K} and recalling~\eqref{eq:C-bd-(0,2]}, we obtain~\eqref{eq:LB-cbetaK} with
    \begin{equation*}
        c_{\beta,d}(K^*) = a_{\beta,d}C_{\beta,d},
    \end{equation*}
    so the result follows from~\eqref{eq:a}.
    \end{proof}









\end{document}